\PassOptionsToPackage{unicode}{hyperref}
\PassOptionsToPackage{hyphens}{url}
\PassOptionsToPackage{dvipsnames,svgnames,x11names}{xcolor}

\documentclass[12pt]{article}

\usepackage{NS_General}
\usepackage{mathrsfs}

\usepackage{dsfont}
\newcommand{\mathbbm}[1]{\text{\usefont{U}{dsrom}{m}{n}#1}}
\usepackage{amsmath,amssymb}
\usepackage{mathtools}
\usepackage[hypertexnames=false]{hyperref}
\usepackage[labelfont=bf]{caption}
\usepackage{enumerate}
\usepackage{relsize}
\usepackage{ragged2e}
\usepackage{scalefnt}
\usepackage{iftex}
\ifPDFTeX
  \usepackage[T1]{fontenc}
  \usepackage[utf8]{inputenc}
  \usepackage{textcomp} 
\else
  \usepackage{unicode-math}
  \defaultfontfeatures{Scale=MatchLowercase}
  \defaultfontfeatures[\rmfamily]{Ligatures=TeX,Scale=1}
\fi
\usepackage{lmodern}
\ifPDFTeX\else  
   
\fi
\IfFileExists{upquote.sty}{\usepackage{upquote}}{}
\IfFileExists{microtype.sty}{
  \usepackage[]{microtype}
  \UseMicrotypeSet[protrusion]{basicmath} 
}{}
\makeatletter
\@ifundefined{KOMAClassName}{
  \IfFileExists{parskip.sty}{
    \usepackage{parskip}
  }{
    \setlength{\parindent}{0pt}
    \setlength{\parskip}{6pt plus 2pt minus 1pt}}
}{
  \KOMAoptions{parskip=half}}
\makeatother
\usepackage{xcolor}
\makeatletter
\ifx\paragraph\undefined\else
  \let\oldparagraph\paragraph
  \renewcommand{\paragraph}{
    \@ifstar
      \xxxParagraphStar
      \xxxParagraphNoStar
  }
  \newcommand{\xxxParagraphStar}[1]{\oldparagraph*{#1}\mbox{}}
  \newcommand{\xxxParagraphNoStar}[1]{\oldparagraph{#1}\mbox{}}
\fi
\ifx\subparagraph\undefined\else
  \let\oldsubparagraph\subparagraph
  \renewcommand{\subparagraph}{
    \@ifstar
      \xxxSubParagraphStar
      \xxxSubParagraphNoStar
  }
  \newcommand{\xxxSubParagraphStar}[1]{\oldsubparagraph*{#1}\mbox{}}
  \newcommand{\xxxSubParagraphNoStar}[1]{\oldsubparagraph{#1}\mbox{}}
\fi
\makeatother

\usepackage{longtable,booktabs,array}
\usepackage{calc} 
\usepackage{etoolbox}
\makeatletter
\patchcmd\longtable{\par}{\if@noskipsec\mbox{}\fi\par}{}{}
\makeatother
\IfFileExists{footnotehyper.sty}{\usepackage{footnotehyper}}{\usepackage{footnote}}
\makesavenoteenv{longtable}
\usepackage{graphicx}
\graphicspath{ {./figures/} }
\makeatletter
\def\maxwidth{\ifdim\Gin@nat@width>\linewidth\linewidth\else\Gin@nat@width\fi}
\def\maxheight{\ifdim\Gin@nat@height>\textheight\textheight\else\Gin@nat@height\fi}
\makeatother
\setkeys{Gin}{width=\maxwidth,height=\maxheight,keepaspectratio}
\makeatletter
\def\fps@figure{htbp}
\makeatother

\makeatletter
\@ifpackageloaded{caption}{}{\usepackage{caption}}
\AtBeginDocument{%
\ifdefined\contentsname
  \renewcommand*\contentsname{Table of contents}
\else
  \newcommand\contentsname{Table of contents}
\fi
\ifdefined\listfigurename
  \renewcommand*\listfigurename{List of Figures}
\else
  \newcommand\listfigurename{List of Figures}
\fi
\ifdefined\listtablename
  \renewcommand*\listtablename{List of Tables}
\else
  \newcommand\listtablename{List of Tables}
\fi
\ifdefined\figurename
  \renewcommand*\figurename{Figure}
\else
  \newcommand\figurename{Figure}
\fi
\ifdefined\tablename
  \renewcommand*\tablename{Table}
\else
  \newcommand\tablename{Table}
\fi
}
\@ifpackageloaded{float}{}{\usepackage{float}}
\floatstyle{ruled}
\@ifundefined{c@chapter}{\newfloat{codelisting}{h}{lop}}{\newfloat{codelisting}{h}{lop}[chapter]}
\floatname{codelisting}{Listing}

\makeatother
\makeatletter
\@ifpackageloaded{caption}{}{\usepackage{caption}}
\@ifpackageloaded{subcaption}{}{\usepackage{subcaption}}
\makeatother

\ifLuaTeX
  \usepackage{selnolig}
\fi
\usepackage[]{natbib}
\usepackage{bookmark}

\IfFileExists{xurl.sty}{\usepackage{xurl}}{}
\hypersetup{
  pdftitle={Title},
  pdfauthor={Author 1; Author 2},
  pdfkeywords={3 to 6 keywords, that do not appear in the title},
  colorlinks=true,
  linkcolor={blue},
  filecolor={Maroon},
  citecolor={Blue},
  urlcolor={Blue},
  pdfcreator={LaTeX via pandoc}}
\usepackage[strings]{underscore}

\renewcommand{\E}{\mathbb{E}}

\newcommand{\anon}{1}

\begin{document}

\def\spacingset#1{\renewcommand{\baselinestretch}%
{#1}\small\normalsize} \spacingset{1}

\if1\anon
{
  \title{\bf Nonparametric Efficient Estimation of Dynamic Treatment Regimes with Competing Risks}
  
  \author{
    Nitya Shah$^1$, 
    Laura D. Carbone$^{2,3}$, 
    Howard A. Fink$^{4,5,6}$, \\
    John T. Schousboe$^{7,8}$, 
    and 
    Jared D. Huling$^1$\thanks{Corresponding Author: \href{mailto:huling@umn.edu}{huling@umn.edu}}
  }

  \date{}

  \maketitle

    \smaller[1]
    \begin{center}
    $^1$Division of Biostatistics \& Health Data Science, School of Public Health, University of Minnesota, Minneapolis, MN \\
    $^2$Center of Innovation for Complex Chronic Healthcare, Edward Hines Jr. VA Hospital, Hines, IL \\
    $^3$Medical College of Georgia at Augusta University, Augusta, GA \\
    $^4$Geriatric Research Education and Clinical Center, Veterans Affairs Health Care System, Minneapolis, MN\\
    $^5$Center for Care Delivery \& Outcomes Research, Veterans Affairs Health Care System, Minneapolis, MN\\
    $^6$Department of Medicine, University of Minnesota, Minneapolis, MN\\
    $^7$HealthPartners Institute, Bloomington, MN \\
    $^8$Division of Health Policy \& Management, School of Public Health, University of Minnesota, Minneapolis, MN
    \end{center}
    \normalsize
} \fi

\if0\anon
{
  \bigskip
  \bigskip
  \bigskip
  \begin{center}
    {\LARGE\bf Nonparametric Efficient Estimation of Dynamic Treatment Regimes with Competing Risks}
\end{center}
  \medskip
} \fi

\bigskip
\begin{abstract}
Many observational studies evaluate the risks and benefits of time-varying interventions. In these longitudinal settings, the primary outcome of interest is often a clinical event that is precluded by mortality, which acts as a competing risk. Evaluating dynamic treatment regimes (DTRs) is complicated by time-varying confounding, right-censoring, and progressive sample size reduction over time, for which traditional methods such as the g-formula or inverse probability weighting may be misspecified or highly unstable. In this work, we develop the nonparametric efficiency theory and derive the efficient influence function for the cumulative incidence of an event under a DTR in the presence of a competing risk and right-censoring. Building on this result, we propose a sequentially doubly robust estimator that accommodates flexible machine learning for nuisance estimation while retaining $\sqrt{n}$-consistency and asymptotic normality. We illustrate the utility of our framework by estimating long-term cumulative incidence of clinical fractures under different bisphosphonate ``drug holiday'' regimes for osteoporosis.
\end{abstract}

\vfill

\newpage
\spacingset{1.8}

\section{Introduction}
\label{sec:intro}

Evaluating time-varying interventions using longitudinal observational data is a central challenge in modern causal inference. Of particular interest are dynamic treatment regimes (DTRs), which formalize these interventions as individualized treatment rules that adapt over time in response to a patient's evolving clinical history \citep{murphy_marginal_2001}. In fields such as public health and personalized medicine, the primary objective under this framework is to estimate the counterfactual cumulative risk of an outcome of interest under population-wide adherence to a given DTR.

Estimating long-term risks presents substantial methodological challenges, however, which consequently limit the utility of standard survival analysis techniques. In longitudinal settings, healthcare providers frequently modify exposures based on patients' evolving characteristics and prior outcomes. This leads to complex time-dependent confounding, where covariate measurements simultaneously act as confounders for future treatment assignment and are themselves affected by past treatment history \citep{hernan_causal_2023}. Furthermore, because longitudinal studies often follow individuals over extended periods of time, right-censoring due to loss to follow-up can preclude the observation of the primary outcome. If attrition depends on a patient's clinical history, it can introduce severe selection bias \citep{hernan_hazards_2010, hernan_causal_2023}. Moreover, especially in older populations, mortality acts as a competing risk that prevents the outcome of interest from occurring. Failure to account for this, such as by treating death as a censoring event, can lead to biased cumulative incidence estimates for the target event \citep{fine_proportional_1999}.

Causal inference methods such as the g-formula and inverse probability weighted (IPW) estimators provide formal strategies for estimating causal effects while addressing time-varying confounding and censoring, and have been widely applied in numerous longitudinal studies \citep{robins_analytic_1993, robins_marginal_2000, hernan_causal_2023}. However, these methods require correct specification of nuisance functions, including models for treatment assignment, censoring, and outcome(s) of interest, which can be especially difficult to accomplish in time-varying settings. Parametric or semiparametric models may lack the flexibility to capture complex clinical decision-making and heterogeneity in patient characteristics, which can introduce bias and reduce the interpretability of results \citep{bang_doubly_2005}. Moreover, IPW estimators are less statistically efficient and highly sensitive to extreme weights, which can be particularly problematic when estimating the risk of rare outcomes \citep{tsiatis_semiparametric_2006, funk_doubly_2011}. These limitations have motivated statistical methods that remain valid under flexible estimation of nuisance functions, while achieving desirable large-sample properties such as consistency, double robustness, and semiparametric efficiency \citep{tsiatis_semiparametric_2006, funk_doubly_2011, van_der_laan_targeted_2006, kennedy_semiparametric_2023}.

Despite substantial theoretical advances in causal inference for longitudinal data, extending DTRs to settings with competing risks remains limited. Efforts to address this gap include the dynamic-regime marginal structural models proposed by \citet{morzywolek_estimation_2022}, which tailor traditional marginal structural models to handle competing events. However, estimation in this framework relies on IPW and is thus susceptible to the same issues of high variance, instability, and vulnerability to nuisance model misspecification mentioned previously. Alternatively, the Longitudinal Modified Treatment Policy (LMTP) framework \citep{diaz_causal_2025} allows for flexible, nonparametric estimation with time-varying exposures under less restrictive modeling assumptions. However, the causal estimands targeted by LMTPs typically correspond to population-level treatment modifications and do not directly address the evaluation of prespecified and individualized DTRs. Taken together, existing methods do not provide a unified framework for the robust and efficient estimation of causal effects for DTRs that explicitly account for competing events, motivating the development of new semiparametric approaches tailored to complex longitudinal decision-making settings.

To highlight these abstract methodological gaps and illustrate the need for statistical efficiency, we consider a motivating application in the long-term management of osteoporosis. Osteoporosis is a chronic skeletal condition characterized by decreased bone mass and the deterioration of bone microarchitecture, leading to higher risk of fractures \citep{morin_osteoporosis_2025}. As the population continues to age, osteoporosis and low bone mass present a growing public health burden, affecting over 53 million adults in the United States alone \citep{cauley_public_2013, wright_recent_2014}. Osteoporotic fractures can contribute to functional impairment, increased mortality, and considerable economic burden on healthcare systems \citep{burge_incidence_2007, michael_lewiecki_hip_2018}. As the gold standard for fracture prevention, bisphosphonate (BP) therapies are highly effective in reducing subsequent fractures in patients \citep{bauer_osteoporosis_2018, fink_long-term_2019}. However, determining the optimal duration of BP treatment is difficult. Prolonged BP use is associated with a modest increase in the risk of certain rare adverse events, such as atypical femoral fractures \citep{shane_atypical_2014, fink_long-term_2019}.

In consideration of these risks, clinical guidelines suggest reassessing ongoing BP therapy after three to five years, with a treatment pause, or ``drug holiday,'' advised for patients at low fracture risk to mitigate long-term adverse effects \citep{adler_managing_2016}. Determining the optimal timing and duration for such a holiday based on a patient's evolving clinical history represents a classic individualized DTR problem. Because the target population consists primarily of older adults, mortality acts as a prominent competing risk, frequently precluding the observation of clinical fractures and other adverse events. Accurate estimation of the cumulative risk of these events is essential for informing clinical decisions that weigh the benefits of fracture prevention against the low, but potentially severe, risk of harm \citep{black_atypical_2020, schilcher_bisphosphonate_2011}. Evaluating drug holiday regimes highlights another practical challenge in longitudinal causal inference. Although databases may initially contain millions of records, sequentially conditioning on patient histories can drastically reduce the available sample size when evaluating a specific DTR, since the cohort is continually diminished over time by treatment deviations, loss to follow-up, mortality, and any events of interest. When evaluating rare outcomes under this kind of structural sparsity, traditional estimators such as IPW may suffer from variance inflation and weight instability. This progressive sample size reduction uniquely motivates the necessity for a nonparametrically efficient estimator that minimizes variance and provides interpretable causal estimates.

To address these challenges in the context of BP use and its associated adverse events in osteoporosis patients, we develop nonparametric efficiency theory and derive the efficient influence function for the cumulative incidence under competing risks. Building on this rigorous foundation, we propose a novel nonparametric efficient estimator for the cumulative incidence of an adverse event under a DTR. Our efficient influence function-based approach offers a robust, broadly applicable framework for estimating cumulative risks in the presence of time-varying exposure, a competing event, and censoring due to loss to follow-up, all while leveraging extensive longitudinal data integrating both clinical records and claims data. Our estimator provides clinically interpretable cumulative risk estimates under varying treatment durations and scenarios, such as drug holidays, which can readily be extended to other settings where similar methodological challenges arise. 

The remainder of this paper is organized as follows. Section \ref{sec:methods} introduces our statistical framework, beginning with notation and data structure (Section \ref{sec:notation}), followed by a formal definition of the causal estimand (Section \ref{sec:estimand}) and the necessary identifying assumptions (Section \ref{sec:assumptions}). We then describe our proposed methodology and comparator approaches, covering identification and parametric estimation (Section \ref{sec:identification_and_para}) as well as nonparametric efficient estimation (Section \ref{sec:nonparametric_efficiency_theory}). Section \ref{sec:sims} presents a simulation study to assess the empirical performance of our estimator relative to established methods. In Section \ref{sec:rwe}, we discuss the application of our framework to the estimation of fracture incidence under different BP treatment regimes. Finally, Section \ref{sec:discussion} concludes with a discussion of our findings and directions for future research.

\section{Methods}
\label{sec:methods}

\subsection{Notation and Data Structure}
\label{sec:notation}

Following \citet{morzywolek_estimation_2022}, we assume that we observe data on discrete time points indexed by $k=0, \dots, K$. For each $k$ and for each individual, outcomes and covariates are measured and treatments are assigned in a study of some duration $K > 0$. Measurements collected prior to the start of the study are indexed at $k = 0$, whereas $k = K$ corresponds to values recorded at the end of the study. Let $Y_k$, $k = 0,...,K$, represent whether the outcome of interest (e.g., fracture) has occurred by time $k$; $Y_k = 1$ indicates that the outcome has been observed by time $k$ and $Y_k = 0$ means that it has not. In the same manner, $Z_k$, $k = 0,...,K$, is an indicator for whether the competing event (e.g., death) has been observed by time $k$. For any time-varying $V_k,k = 0,...,K$, let $\overline{V}_k \equiv (V_0,...,V_k)$ be a vector of all previously measured values of $V_k$ up to and including time $k$. In a similar vein, let $\underline{V}_k \equiv (V_k,...,V_K)$ be a vector of all future values of $V_k$, including time $k$. For all $k$, $Y_k = 1$ and $Z_k = 1$ imply that $\underline{Y}_{k+1} = 1$ and $\underline{Z}_{k+1} = 1$, respectively. We also specify that $Y_{0} = Z_{0} = 0$ for all participants to avoid including any individuals that have already experienced the outcome of interest or passed away prior to the start of the study. Additionally, let $X_0$ be the set of patient characteristics available at the start of the study, including baseline values for time-varying covariates (e.g., medication use and biomarker status) as well as static covariates (e.g., race and ID information). Let $X_k$, for $k = 1,...,K$, represent patient characteristics at time $k$, or as measured in the interval between time points $k-1$ and $k$. Define $\Hcal_k \equiv \left\{ Z_k, Y_k, X_k \right\}$ as the set of patient outcomes and covariates at time $k$.

The random variable $A_k$, $k = 0,...,K-1$, is a treatment indicator at time point $k$; $A_k = 0$ represents no treatment given, while $A_k = 1$ indicates treatment use at time $k$. In the context of a drug holiday, for every time point $k$ at which an individual $i$ is taking a break from bisphosphonate use, $A_{k,i} = 1$. Finally, in this study, we allow for censoring due to loss to follow-up via a censoring indicator $U_k$, for $k = 0,...,K$. If a study participant has been censored by time $k$, $U_k = 0$; otherwise, if they remain uncensored by time $k$, $U_k = 1$. We assume that, for all $k$, $U_k = 0 \implies \underline{U}_{k+1} = 0$ (i.e., once an individual is censored, they remain censored for the remainder of the study). At every time point $k$, we observe the variables $O_k = \{Z_k, Y_k, X_k, A_k, U_k\}$, in that temporal order.

\subsection{Causal Estimand}
\label{sec:estimand}

In the field of causal inference, treatment effects are generally defined and estimated using the potential outcomes framework \citep{splawa-neyman_application_1990, rubin_estimating_1974}. Suppose $a = (a_0, a_1,...,a_{k-1})$ represents a vector of treatment assignments, one at each discrete time point from baseline through time $k-1$. Under this framework, we can define $Y^a_k$ as the potential outcome for the event of interest which we would observe if all study participants had received treatment $A = a$ through time $k-1$. Similarly, $Z^a_k$ would be the potential outcome for the competing event had all participants received the sequence of treatments $A = a$ through time $k-1$. In our motivating example, not only does an individual's treatment vary with each time point $k$, but their treatment assignment at each time point is dependent on their previous treatment, outcome, and covariate history. This dependence introduces the potential for time-varying confounding \citep{hernan_causal_2023}. 

Similar to \citet{morzywolek_estimation_2022}, our goal is to evaluate the impact of a series of treatment decisions on the cumulative incidence of some event of interest. A dynamic treatment regime (DTR) $d$ is a sequence of treatment decision rules that specify the treatment an individual would receive at each time point based on their previous clinical history. Defining the treatment decision made at each time $k$ as $d_k(\overline{\Hcal}_k, \overline{A}_{k-1})$, or $d_k$ for ease of notation, the complete DTR is the ordered collection of these rules across the study ($d \equiv \overline{d}_{K-1}$). Using potential outcomes notation, let $Y^d_k$ be the potential outcome for the event of interest which we would observe if all participants had followed DTR $d$, i.e., had been assigned treatments according to $d$, until $k-1$. Similarly, let $Z^d_k$ be the potential outcome for the competing event had all individuals followed $d$ until $k-1$. In our motivating application, we are interested in estimating the risk of fracture associated with a set of $m$ prespecified bisphosponate DTRs, denoted $\Dcal = \{d^{1},..., d^{m}\}$.  By identifying {$\Pr\left(Y_k^d = 1\right)$} for each $d \in \Dcal$, we can identify the optimal course of treatment; i.e., select the regime that minimizes incidence of the target event by some time $k$ in the population. For $k = 1,...,K$, {$\Pr\left(Y_k^d = 1\right)$}, or the cumulative incidence of the outcome of interest at time $k$, can be rewritten as: 
\begin{equation}\allowdisplaybreaks
    \Pr\left(Y_k^d = 1\right) = \sum_{t=1}^{k} \left[ h_t^{(Y)}(d) \left\{1 - h_t^{(Z)}(d)\right\} \times \prod_{j=1}^{t-1}\left\{1 - h_j^{(Y)}(d)\right\}\left\{1 - h_j^{(Z)}(d)\right\}\right],
    \label{eq:cumu_inc}
\end{equation}
which involves the estimation of the discrete-time counterfactual event-specific hazards $h_t^{(Y)}(d)$ and $h_t^{(Z)}(d)$ for the outcome of interest and the competing event, respectively \citep{morzywolek_estimation_2022,young_causal_2020}. These are the event-specific hazards at time point $k$ under DTR $d$ among all individuals who are in the population and also have yet to experience the outcome of interest or the competing event by time $k$, defined as:
\begin{gather}\allowdisplaybreaks
    h_t^{(Y)}(d) = \Pr\left(Y_t^d = 1\: \big| \: Y_{t-1}^d = 0, Z_{t}^d = 0\right) \text{ and}\label{eq:h1}\\
    h_t^{(Z)}(d) = \Pr\left(Z_t^d = 1\: \big| \: Z_{t-1}^d = 0, Y_{t-1}^d = 0\right)\label{eq:h2}.
\end{gather}
The counterfactual event-specific hazard at time $t$ for the competing event, $h_t^{(Z)}(d)$, is conditional on no events through the previous time point, whereas the the counterfactual event-specific hazard at time $t$ for the outcome of interest, $h_t^{(Y)}(d)$, is conditional on no outcome of interest through the previous time point and no competing event through the \textit{current} time point $t$. This is consistent with our assumed temporal ordering of the observed variables $O_t$. Furthermore, although the time-varying hazards \eqref{eq:h1} and \eqref{eq:h2} are defined using potential outcomes, they lack a causal interpretation because they are susceptible to an inherent selection bias \citep{morzywolek_estimation_2022}. Following baseline, the cause-specific hazards are estimated among those who survive in the study population at each time point, potentially leading to inaccurate estimates if individuals who have a higher risk of the outcome are differentially excluded over time \citep{hernan_hazards_2010}. Nevertheless, \eqref{eq:h1} and \eqref{eq:h2} are useful as intermediate quantities in the estimation of the cumulative incidence \eqref{eq:cumu_inc}, which better reflects the overall risk of the outcome of interest.

Ideally, we would compare all hypothetical $d^{1},..., d^{p} \in \Dcal$ with a randomized trial with $p$ arms, one for each regime. This, however, becomes logistically infeasible as the number of regimes $p$ and the number of time points $K$ increase. Instead, we evaluate each DTR retrospectively using observational data by identifying individuals whose observed treatment histories coincide exactly with the sequence specified by regime $d$. Similar to the work of \citet{morzywolek_estimation_2022}, this is done by defining a compatibility indicator ${C}^{(d)}_{k} = \mathbbm{1}\left( \overline{A}_k = \overline{d}_k\right)$, which is 1 if the individual’s treatment history up to time $k$ aligns with $d$ and 0 otherwise. This is used to implement artificial censoring; individuals are censored at the time point at which their observed treatment path first deviates from $d$ \citep{robins_analytic_1993, murphy_marginal_2001}.

\subsection{Causal Assumptions}
\label{sec:assumptions}

In order to estimate the hazards of the outcome of interest and the competing event at each time point $k$ as a function of our observed data, we are required to make three standard causal assumptions, which are similar to those made by \citet{morzywolek_estimation_2022} and \citet{young_causal_2020}. For each regime $d$, we assume the following for all $k = 0,..., K$.

\begin{assumption}{1}{Sequential Ignorability}\label{sequential-ignorability}
\begin{gather*}
    (\underline{Y}^d_{k+1}, \underline{Z}^d_{k+1}) \indep A_k \; | \; \overline{X}_{k} = \overline{x}_{k}, \overline{Y}_{k} = 0, \overline{Z}_{k} = 0, \overline{U}_{k-1} = 1, \overline{C}_{k-1} = 1 \text{, and}\\
    (\underline{Y}^d_{k+1}, \underline{Z}^d_{k+1}, \underline{X}^d_{k+1}, \underline{A}^d_{k+1}) \indep U_k \; | \; \overline{X}_{k} = \overline{x}_{k}, \overline{Y}_{k} = 0, \overline{Z}_{k} = 0, \overline{A}_{k} = \overline{a}_{k}, \overline{U}_{k-1} = 1.
\end{gather*}
\end{assumption}

\begin{assumption}{2}{Positivity}\label{positivity}
\begin{gather*}
    \Pr(A_k = d_k\; | \; \overline{X}_{k} = \overline{x}_{k}, \overline{Y}_{k} = 0, \overline{Z}_{k} = 0, \overline{U}_{k-1} = 1, \overline{C}_{k-1} = 1) > 0 \;\; \text{w.p.} \; 1 \text{, and}\\
    \Pr(U_k = 1\; | \; \overline{X}_{k} = \overline{x}_{k}, \overline{Y}_{k} = 0, \overline{Z}_{k} = 0, {\overline{A}_{k} = \overline{a}_{k}}, \overline{U}_{k-1} = 1) > 0 \;\; \text{w.p.} \; 1,
\end{gather*}
\noindent {where ``w.p. 1'' denotes ``with probability 1'' with respect to the true data distribution.}
\end{assumption}

\begin{assumption}{3}{Consistency}\label{consistency}
\begin{equation*}
    \text{If\;\;}U_k = 1, \text{\;then\;\;} \overline{Y}_{k} = \overline{Y}_{k}^{\overline{A}_{k-1}},\;\; \overline{Z}_{k} = \overline{Z}_{k}^{\overline{A}_{k-1}}, \text{\;and\;\;} \overline{X}_{k} = \overline{X}_{k}^{\overline{A}_{k-1}}.
\end{equation*}
\end{assumption}
The first part of Assumption \ref{sequential-ignorability} requires that at any time point $k$, there is no unmeasured confounding. In other words, for individuals that have not yet experienced the event of interest or the competing event, their future potential outcomes for both events are independent of their treatment assignment at time $k$, conditional on all previously measured characteristics. The second part of Assumption \ref{sequential-ignorability} requires that whether an individual is lost to follow up at time $k$ is independent of all future potential outcomes, patient covariates, and treatment assignments. Positivity, or Assumption \ref{positivity}, states that at any time $k$, there is a nonzero probability of individuals being compatible with any DTR $d$, as well as of individuals not being lost to follow up, both conditional on previously measured information. Finally, Assumption \ref{consistency} states that, for all uncensored individuals, their observed outcomes and covariates correspond to their potential outcomes under their observed treatment history. This also implies that an individual's observed outcomes are not influenced by the treatments given to others, which is the standard causal assumption of no interference. 

\subsection{Identification and Parametric Estimation}
\label{sec:identification_and_para}

\subsubsection{Weighting Identification and Estimation}

Artificial censoring due to incompatibility with a treatment regime $d$ at a given time $k$ can increase selection bias if the remaining patients are no longer representative of the original population. \citet{robins_correcting_2000} demonstrated that inverse probability of censoring (IPC) weighting can effectively correct for this bias by reweighting the observed data to account for the inherent censoring mechanism. We define IPC weights for our dataset under each regime $d$ which reflect both the artificial censoring and any censoring due to loss to follow-up:
\begin{equation}\label{eq:ipcw}
    W^{(d)}_k = \frac{{C}^{(d)}_{k}}{{\prod_{j=0}^k \pi^{C^{(d)}}_j}} \times \frac{\prod_{j=0}^k{U}_{j}}{\prod_{j=0}^k \pi^U_j},
\end{equation}
where
\begin{gather}
    \pi^{C^{(d)}}_j = \Pr\left(C^{(d)}_j = 1 \,\middle|\, \overline{\Hcal}_j, \overline{U}_{j-1} = 1, \overline{C}^{(d)}_{j-1} = 1 \right), \text{ and} \label{eq:pi_C}\\
    \pi^U_j = \Pr\left(U_j = 1 \,\middle|\, \overline{\Hcal}_j, \overline{U}_{j-1} = 1, \overline{C}^{(d)}_j = 1 \right).\label{eq:pi_U}
\end{gather}
For $j = 0,\dots,k$, $\pi^{C^{(d)}}_j$ and $\pi^U_j$ are the propensity scores for compatibility and being uncensored, respectively. By convention, we set $\pi^{C^{(d)}}_j = \pi^U_j = 1$ if $\overline{Y}_{j} \neq 0$ or $\overline{Z}_{j} \neq 0$. Here and throughout, we adopt the conventions $C^{(d)}_{-1} = U_{-1} = 1$ and $W^{(d)}_{j} = 1$ for $j < 0$, so that $\pi^{C^{(d)}}_0 = \Pr(A_0 = d_0 \mid \Hcal_0)$ and $\pi^U_0$ are well defined. The compatibility product in \eqref{eq:ipcw} begins at $j = 0$ because the numerator $C^{(d)}_k$ censors individuals whose baseline treatment $A_0$ deviates from $d_0$; reweighting by $\pi^{C^{(d)}}_0$ corrects for this baseline artificial censoring and gives $\E\big[W^{(d)}_0\big] = 1$, a property used repeatedly in the proofs. The following theorem formally justifies the use of the IPC weights in \eqref{eq:ipcw} to identify the cumulative incidence function.

\begin{theorem}[Lemma 1 of \cite{morzywolek_estimation_2022}]
\label{theorem:weighting_identification}
Under Assumptions \ref{sequential-ignorability}, \ref{positivity}, and \ref{consistency}, 
\begin{gather*}\allowdisplaybreaks
    {h}_k^{(Y)}(d) = \frac{\E\left[Y_{k}(1-Z_{k})(1-Y_{k-1})W^{(d)}_{k-1}\right]}{\E\left[(1-Z_{k})(1-Y_{k-1})W^{(d)}_{k-1}\right]}, \text{ and}\\
    {h}_k^{(Z)}(d) = \frac{\E\left[Z_{k}(1-Y_{k-1})(1-Z_{k-1})W^{(d)}_{k-1}\right]}{\E\left[(1-Y_{k-1})(1-Z_{k-1})W^{(d)}_{k-1}\right]}.
\end{gather*}
\end{theorem}

\noindent{In the proof of their Lemma 1, \citet{morzywolek_estimation_2022}} argue that, provided the propensity score models are correctly specified, the following IPC weighted estimators for \eqref{eq:h1} and \eqref{eq:h2} consistently estimate the respective expectation ratios presented in Theorem \ref{theorem:weighting_identification}:
\begin{gather}
    \hat{h}_k^{(Y)}(d) = \frac{\sum_{i = 1}^n Y_{i,k}(1-Z_{i,k})(1-Y_{i,k-1})\hat{W}^{(d)}_{i,k-1}}{\sum_{i = 1}^n (1-Z_{i,k})(1-Y_{i,k-1})\hat{W}^{(d)}_{i,k-1}}, \text{ and}\label{eq:h1_hat}\\ 
    \hat{h}_k^{(Z)}(d) = \frac{\sum_{i = 1}^n Z_{i,k}(1-Y_{i,k-1})(1-Z_{i,k-1})\hat{W}^{(d)}_{i,k-1}}{\sum_{i = 1}^n(1-Y_{i,k-1})(1-Z_{i,k-1})\hat{W}^{(d)}_{i,k-1}}.\label{eq:h2_hat}
\end{gather}
In Section 4.4 of their main text, the authors define the following IPC weighted Aalen-Johansen estimator for the cumulative incidence:
\begin{equation}
    \hat\Pr\left(Y_k^d = 1\right) = \sum_{t=1}^{k} \left[ \hat{h}_t^{(Y)}(d) \left\{1 - \hat{h}_t^{(Z)}(d)\right\} \times \prod_{j=1}^{t-1}\left\{1 - \hat{h}_j^{(Y)}(d)\right\}\left\{1 - \hat{h}_j^{(Z)}(d)\right\}\right].
    \label{eq:ipcw_cumu_inc}
\end{equation}
From \eqref{eq:ipcw_cumu_inc}, we can see that the cumulative incidence at time $k$ can be estimated by substituting \eqref{eq:h1_hat} and \eqref{eq:h2_hat} into \eqref{eq:cumu_inc} and summing across all time points through time $k$. 

\subsubsection{G-formula Identification and Estimation}\label{sec:g-formula_identification}

Alongside IPC weighting, the g-formula, first proposed by \cite{robins_new_1986}, provides an alternative identification approach that expresses the cumulative incidence function as a function of the observed data by modeling the outcome and covariate distributions. To formalize this alternative identification strategy, we first define the expectations: 
\begin{gather}
    \mu_{Y_k}(\overline{x}_{k-1}) \equiv \E \left[Y_k \:\big|\: \overline{Z}_k = 0, \overline{Y}_{k-1} = 0, \overline{U}_{k-1} = 1, \overline{C}_{k-1} = 1, \overline{X}_{k-1} = \overline{x}_{k-1} \right],\text{ and}\label{eq:mean_of_Y_k}\\
    \nu_{1- Z_k}(\overline{x}_{k-1}) \equiv \E \left[ 1- Z_k\:\big|\: \overline{Y}_{k-1} = 0, \overline{Z}_{k-1} = 0, \overline{U}_{k-1} = 1, \overline{C}_{k-1} = 1, \overline{X}_{k-1} = \overline{x}_{k-1}\right].\label{eq:mean_of_1-Z_k}
\end{gather}
We introduce vectors $G^{(k)} = \left\{G_1^{(k)}, G_2^{(k)}, ..., G_k^{(k)},  G_{k+1}^{(k)}\right\}$ and $R^{(k)} = \left\{R_1^{(k)}, R_2^{(k)}, ..., R_{k-1}^{(k)}\right\}$ to store iterated expectation terms, as defined below. Following the recursive format used by \citet{diaz_causal_2025} in the identification step of the LMTP framework, we let $G_{k+1}^{(k)} = 1$. We further define $G_k^{(k)} \equiv \mu_{Y_k}(\overline{x}_{k-1})\; \nu_{1 - Z_k}(\overline{x}_{k-1})$, and specify $R_j^{(k)} \equiv \mu_{1 - Y_j}(\overline{x}_{j-1})\; \nu_{1 - Z_j}(\overline{x}_{j-1})$. Let $\overline{\Ocal}_k \equiv \left\{ \overline{Y}_k = 0, \overline{Z}_k = 0, \overline{U}_{k-1} = 1, \overline{C}_{k-1} = 1, \overline{X}_{k-1} \right\}$ be the event representing being in the risk set and not having experienced the event of interest nor the competing event prior to time $k$. For $t = k-1, k-2, ..., 1$, we recursively define: 
\begin{equation}\label{eq:G_recursive}
    G_t^{(k)} = \E \left[ R_t^{(k)} \times G_{t+1}^{(k)}\: \big| \: \overline{\Ocal}_t \right].
\end{equation}
\noindent{When} fully iterated to include all data observed up to time $k$, $G_1^{(k)}$ is equal to:
\begin{equation}\label{eq:g-formula}
    \E \left[ \E \left[ \cdots \E \left[ \mu_{Y_k}(\overline{X}_{k-1})\; \nu_{1 - Z_k}(\overline{X}_{k-1})\; \prod_{j=1}^{k-1} \mu_{1 - Y_{j}}(\overline{X}_{j-1})\; \nu_{1 - Z_{j}}(\overline{X}_{j-1})\: \big| \: \overline{\Ocal}_{k-1} \right] \cdots \big|\; \overline{\Ocal}_{2}\right] \; \Big|\; \overline{\Ocal}_{1} \right].
\end{equation}

The following result, adapted from \cite{young_causal_2020}, shows that, under standard identification assumptions, the g-formula expression in (\ref{eq:g-formula}) is equal to the cumulative incidence.
\begin{theorem}\label{theorem:g-formula_identification}
Under Assumptions \ref{sequential-ignorability}, \ref{positivity}, and \ref{consistency}, the incidence at time $k$ is identified by the g-formula:
\begin{equation*}
    h_k^{(Y)}(d) \left\{1 - h_k^{(Z)}(d)\right\} \times \prod_{j=1}^{k-1}\left\{1 - h_j^{(Y)}(d)\right\}\left\{1 - h_j^{(Z)}(d)\right\} = \E \left[G_1^{(k)}\right].
\end{equation*}
\end{theorem}

\noindent{The} proof of Theorem \ref{theorem:g-formula_identification} is given in the Supplementary Material (Section \ref{sec:identification}). We can estimate the \textit{cumulative} incidence at time $k$ by calculating and summing \eqref{eq:g-formula} for all time points through time $k$, represented in the following corollary. 

\begin{corollary}
    Under the same assumptions used in Theorem \ref{theorem:g-formula_identification}, the cumulative incidence at time $k$ is identified by the g-formula:
\begin{equation*}
    \sum_{t=1}^k \left[ h_t^{(Y)}(d) \left\{1 - h_t^{(Z)}(d)\right\} \times \prod_{j=1}^{t-1}\left\{1 - h_j^{(Y)}(d)\right\}\left\{1 - h_j^{(Z)}(d)\right\} \right] = \sum_{t=1}^k \E\left[G_1^{(t)}\right].
\end{equation*}
\end{corollary}

\noindent{Our} g-formula estimator for the incidence at time $k$ is then given by:
\begin{equation}\label{eq:g-formula_est}
\begin{split}
    \hat{G}_1^{(k)} = \frac{1}{n} \sum_{i=1}^n \hat{\E} \bigg[ \hat\E \bigg[ \cdots &\hat\E \bigg[ \hat{\mu}_{Y_k}(\overline{X}_{i,k-1})\; \hat{\nu}_{1 - Z_k}(\overline{X}_{i,k-1})\\ 
    &\, \quad \times \prod_{j=1}^{k-1} \hat{\mu}_{1 - Y_{j}}(\overline{X}_{i,j-1})\; \hat{\nu}_{1 - Z_{j}}(\overline{X}_{i,j-1})\: \big| \: \overline{\Ocal}_{i,k-1} \bigg] \cdots \big|\; \overline{\Ocal}_{i,2}\bigg] \big|\; \overline{\Ocal}_{i,1}\bigg].
\end{split}
\end{equation}
Consistency of \eqref{eq:g-formula_est} for the cumulative incidence relies on correctly specified outcome models \eqref{eq:mean_of_Y_k} and \eqref{eq:mean_of_1-Z_k}, which are estimated using regression techniques \citep{young_causal_2020}. The representations of the cumulative incidence function via IPC weighting and the g-formula naturally motivate estimation strategies that draw on elements of both approaches. 

\subsection{Nonparametric Estimation}
\label{sec:nonparametric_efficiency_theory}

\subsubsection{Efficient Influence Function}
To construct an estimator with desired asymptotic properties, we begin by deriving the efficient influence function (EIF) for the cumulative incidence function within a nonparametric framework. The EIF characterizes the most efficient regular, asymptotically linear estimator by characterizing the sensitivity of a parameter to perturbations in the underlying data distribution. Using the EIF, we can construct estimators that are consistent, asymptotically normal, and achieve the nonparametric efficiency bound under standard regularity conditions \citep{tsiatis_semiparametric_2006, kennedy_semiparametric_2023}. We denote the incidence at time $k$ under a regime $d$ as $\psi_k$:
\begin{equation}\label{eq:incidence}
    \psi_k \equiv h_k^{(Y)}(d) \left\{1 - h_k^{(Z)}(d)\right\} \times \prod_{j=1}^{k-1} \left\{1 - h_j^{(Y)}(d)\right\}\left\{1 - h_j^{(Z)}(d)\right\}.
\end{equation}
Let $\varphi_k$ represent the EIF for $\psi_k$. 

\begin{theorem}\label{theorem:EIF}
    Under Assumptions \ref{sequential-ignorability}, \ref{positivity}, and \ref{consistency}, the EIF for $\psi_k$, $k = 1,\dots,K$, is given by:
\begin{equation}
    \begin{split}\label{eq:EIF}
    \varphi_k &= W_{k-1} \left\{ \prod_{j=0}^{k-1} (1-Y_j)(1-Z_j) \right\} \Big\{ Y_k (1-Z_k) - G_k^{(k)} \Big\}\\
    &\qquad + \sum_{t = 1}^{k-1} W_{t-1} \left\{ \prod_{j=0}^{t-1} (1-Y_j)(1-Z_j) \right\} \Big\{ (1 - Y_t) (1-Z_t)G_{t+1}^{(k)} - G_t^{(k)} \Big\}\\
    &\qquad + G_1^{(k)} - \psi_k.
\end{split}
\end{equation}
\end{theorem}

For $t=1,\dots,k$, the $G_t^{(k)}$ terms are defined in \eqref{eq:G_recursive}. Let $\psi_k^*$ be the \textit{cumulative} incidence at time $k$, i.e., $\psi_k^* = \sum_{t=1}^k \psi_t$, as defined in \eqref{eq:cumu_inc}. We can obtain the EIF for $\psi_k^*$ by calculating and summing \eqref{eq:EIF} for all time points through time $k$, represented in the following corollary. 

\begin{corollary}
     Under Assumptions \ref{sequential-ignorability}, \ref{positivity}, and \ref{consistency}, the EIF for $\psi_k^*$ is given by $ \varphi^*_k = \sum_{t=1}^k \varphi_t$.
\end{corollary}

We demonstrate in the Supplementary Material (Section \ref{sec:supp-EIF}) that \eqref{eq:EIF} is the EIF for the incidence at $k$ by establishing that it is the pathwise derivative of \eqref{eq:incidence} \citep{kennedy_semiparametric_2016}.

\subsubsection{EIF-based Estimator and its Asymptotic Properties}
\label{sec:estimator_and_asymptotics}

The EIF for the incidence in (\ref{eq:EIF}) has mean zero, i.e., $\E [\varphi_k] = 0$, and includes the subtraction of $\psi_k$. This motivates a simple one-step estimator derived from the property $\E [\varphi_k + \psi_k] = \psi_k$. The EIF-based estimator for \eqref{eq:incidence} is thus given by the following sample average:
\begin{align}\label{eq:EIF-based_estimator}
    \hat{\psi_k} =\, &\frac{1}{n}\, \sum_{i = 1}^n \left[\hat{W}_{i,k-1} \left\{ \prod_{j=0}^{k-1} (1-Y_{i,j})(1-Z_{i,j}) \right\} \Big\{ Y_{i,k} (1-Z_{i,k}) - \hat{G}_{i,k}^{(k)} \Big\}\nonumber\right.\\
    &\qquad\qquad + \sum_{t = 1}^{k-1} \hat{W}_{i,t-1} \left\{ \prod_{j=0}^{t-1} (1-Y_{i,j})(1-Z_{i,j}) \right\} \Big\{ (1 - Y_{i,t}) (1-Z_{i,t})\hat{G}_{i,t+1}^{(k)} - \hat{G}_{i,t}^{(k)} \Big\}\nonumber\\
    &\qquad\qquad + \left.\vphantom{\prod_{j=0}^{k-1}}\hat{G}_{i,1}^{(k)}\right].
\end{align}
The weight terms $\hat{W}_{i,t}, \; t = 0,\dots,k-1,$ in \eqref{eq:EIF-based_estimator} are estimates of the true IPC weights in \eqref{eq:ipcw} and are obtained by estimating the propensity scores \eqref{eq:pi_C} and \eqref{eq:pi_U} at each time for each individual. Additionally, the $\hat{G}_{i,t}^{(k)}$ terms, $t=1,\dots,k$, are estimates of the respective quantities defined in \eqref{eq:G_recursive} and are obtained by estimating the outcome models \eqref{eq:mean_of_Y_k} and \eqref{eq:mean_of_1-Z_k}, then \eqref{eq:g-formula_est}. The EIF-based estimator for $\psi_k^*$ is given by $\hat\psi_k^* = \sum_{t=1}^k \hat{\psi_t}$. 

The estimator in \eqref{eq:EIF-based_estimator} achieves desirable theoretical large-sample properties, as summarized in Theorem \ref{theorem:EIF_asymptotics}. However, estimation of the nuisance functions (which increase in number as $k$ increases) requires further assumptions that limit the complexity and magnitude of the resulting estimators. Using the same sample both to estimate nuisance parameters and to evaluate the empirical mean of the influence function can lead to overfitting. One approach to control this is to assume that the function class involved is sufficiently ``simple,'' i.e., Donsker \citep{kennedy_semiparametric_2023}. We assume here that the nuisance functions are estimated from a separate, independent sample. In practice, sample splitting, or cross-fitting \citep{clark_transportability_2024,zeng_efficient_2024} is often used; our assumption of the use of an independent sample for estimation can be trivially extended to handle cross-fitting. This involves dividing the data into separate parts such that the nuisance parameters are estimated on one part of the sample and the mean of the influence function is computed on a different, independent part. By formally separating the two estimation steps, sample splitting breaks the dependence that leads to overfitting. This separation allows us weaken assumptions and only require consistency of the nuisance parameter estimates \citep{kennedy_semiparametric_2023, zeng_efficient_2024}.
\begin{theorem}\label{theorem:EIF_asymptotics}
\allowdisplaybreaks
    Suppose that all nuisance functions, consisting of the outcome models $\mu_{Y_k}$, {$\nu_{1-Z_k}$}, and $\left\{ \mu_{1 - Y_j}, \nu_{1 - Z_j} \right\}$ for $j = 1, \dots, k-1$, and the treatment and censoring propensity models $\left\{\pi^{C(d)}_j, \pi^U_j\right\}$ for {$j =0, \dots, k-1$}, are estimated from a distinct and independent sample. Furthermore, assume the following conditions hold:
    \begin{enumerate}
        \item For $t = 1, \dots, k$, $\left\|\hat{W}_{t-1} - {W}_{t-1}\right\| = o_p(1)$, and $\left\|\hat{G}_t^{(k)} - {G}_t^{(k)}\right\| = o_p(1)$, {where  $\|f\|^2 = \int f^2 \, \mathrm{d}\P$}.
        \item Letting $\Bcal_t = (1-Y_t)(1-Z_t)$ and $\Bcal^*_k = Y_k(1-Z_k)$, there exists an $M < \infty$ such that $\left|\Bcal_k^{*} - \hat{G}_k^{(k)}\right| < M$ w.p. 1; for $t=1,\dots,k$, $\left|W_{t-1}\right| < M$ w.p. 1; and, for $t=1,\dots,k-1$, $\left|\Bcal_t \hat{G}_{t+1}^{(k)} - \hat{G}_t^{(k)}\right| < M$ w.p. 1.
        \item There exists some $\epsilon > 0$ such that, for $t=1,\dots,k$, $\prod_{j=0}^{t-1}{\hat{\pi}_j^U \hat{\pi}_j^{C^{(d)}}} \geq \epsilon$ w.p. 1.
        \item For $t=1,\dots,k$, the nuisance function estimators satisfy the product rate condition:
        \begin{equation}\label{eq:sequential_rate_condition}
        \left\| G_t^{(k)} - \hat{G}_t^{(k)} \right\| \left\| \frac{1}{\hat{\pi}_{t-1}^U \hat{\pi}_{t-1}^C} - \frac{1}{\pi_{t-1}^U \pi_{t-1}^C} \right\| = o_p\left(\frac{1}{\sqrt{n}}\right).
    \end{equation}
    \end{enumerate}
    Then,
    \begin{equation*}
        \sqrt{n}\left(\hat\psi_k - \psi_k\right) = \frac{1}{\sqrt{n}}\sum_{i=1}^n \varphi_{i,k} + o_p(1),
    \end{equation*}
    where $\varphi_{i,k}$ is the $i$-th sample EIF of the incidence at time $k$ and, as a result, the EIF-based estimator in \eqref{eq:EIF-based_estimator} is $\sqrt{n}$-consistent, asymptotically normal, and its asymptotic variance, the variance of $\varphi_k$, achieves the non-parametric efficiency bound. 
\end{theorem}

For inference, the asymptotic variance of the EIF-based estimator in \eqref{eq:EIF-based_estimator} is equal to the variance of the EIF in \eqref{eq:EIF}. In practice, this can be estimated by computing the sample variance of the estimated EIF, after centering it by its sample mean. This can be shown to be a consistent estimator of the variance if all nuisance parameter models are correctly specified \citep{shook-sa_double_2024}. Due to the linearity of the EIF, it is straightforward to show that the asymptotic distribution of the EIF-based estimator of the cumulative incidence function is then $\sqrt{n}\left(\hat\psi_k^* - \psi_k^*\right) = \frac{1}{\sqrt{n}}\sum_{i=1}^n\sum_{t=1}^k \varphi_{i,t} + o_p(1)$. 

Additionally, Theorem \ref{theorem:EIF_asymptotics} details the conditions under which $\hat{\psi_k}$ is asymptotically normal and achieves $\sqrt{n}$-consistency. We show these results in the Supplementary Material (Section \ref{sec:asymptotics_proof}). The sequential product rate condition in \eqref{eq:sequential_rate_condition} highlights why our estimator achieves strong sequential double robustness. Rather than requiring $\sqrt{n}$-convergence across all time points simultaneously, \eqref{eq:sequential_rate_condition} dictates that at each time $t$, the \textit{product} of the estimation errors from the g-formula terms and the IPC weights must vanish faster than $n^{-1/2}$. This structure accommodates flexible, nonparametric machine learning algorithms whose individual convergence rates may be slower than the parametric rate. For instance, if one nuisance model converges at a slower rate of $o_p(n^{-1/4})$, the target parameter retains $\sqrt{n}$-consistency provided the second model also converges as $o_p(n^{-1/4})$ rate at that specific time. This product structure is also what protects against misspecification of one of the models. If the propensity score models are incorrectly specified, this error remains bounded ($O_p(1)$), meaning $\sqrt{n}$-consistency is preserved as long as the outcome models are correctly specified and converge at a rate of $o_p(n^{-1/2})$. Conversely, if the outcome models are wrong, the estimator remains consistent provided the propensity models converge at a rate of $o_p(n^{-1/2})$. Consequently, our proposed estimator achieves $\sqrt{n}$-consistency at each time point for which either the propensity score models \textit{or} outcome models are correctly specified at the appropriate rates \citep{diaz_causal_2024}. 

\section{Simulation Study}
\label{sec:sims}

The simulation study presented illustrates the properties of our EIF-based estimator and investigates its performance relative to the IPW Aalen-Johansen estimator \citep{morzywolek_estimation_2022} and the g-formula approach \citep{robins_new_1986}.

\subsection{Simulation Settings and Estimation} \label{sec:sims-dgp}

We conducted this simulation study to mimic the data structure expected of a longitudinal observational study. For each simulation replication, we fixed $K=8$ and generated $O_k$, for $k=0, \dots, 8$, (as defined in Section \ref{sec:notation}) for six distinct sample sizes: $N \in \left\{1000, 2000, 4000, 8000, 16000, 32000\right\}$. To evaluate the the IPW, g-formula, and proposed EIF-based estimators, we considered two distinct regimes. For time points $k=0, \dots, 8$, Regime 1 is defined as $d^1 = \left\{0,0,0,0,0,0,0,0,0\right\}$ and represents no treatment throughout the study. Regime 2 is defined as $d^2 = \left\{0,0,0,0,0,1,1,1,1\right\}$ and represents no treatment in the first half of the study followed by treatment in the second half.

We simulated four continuous, time-varying covariates $X_{i,k}^{(m)}$, and four continuous, time-static covariates, $X_{i}^{*(m)}$, $m=1, \dots, 4$, across all times $k=0, \dots, 8$ for all individuals $i = 1, \dots, N$. The time-varying treatment, $A_{i,k}$, for individual $i$ at time $k$ was simulated using a logistic model. $A_{i,k}$ was compared to each of the two pre-specified DTRs, $d^1$ and $d^2$, yielding a sequence of compatibility indicators for each regime, $C_{i,k}^{(1)}$ and $C_{i,k}^{(2)}$, respectively, for every individual $i$ at each time $k$. Once an individual became incompatible with a regime (i.e., $C_{i,k}^{(1)} = 0$ or $C_{i,k}^{(2)} = 0$), they remained incompatible with that regime for all subsequent time points $k' > k$. If the event or competing event occurred at time $k$ but a participant was compatible at time $k-1$, that individual remained compatible for all $k' > k$. 

Building on the simulation of covariates and treatment, we generated the time-varying outcome of interest and competing event ($Y_{i,k}$ and $Z_{i,k}$, respectively) for individual $i$ at time $k$ as Bernoulli random variables. At baseline, no individuals had experienced the event of interest or the competing event (i.e., $Y_{i,0} = Z_{i,0} = 0$). For all time points following an event or competing event, the corresponding outcome indicators were set to 1. Finally, we generated censoring using a logistic model as well. No individuals were censored at baseline and all who became censored at time $k$ remained censored for all subsequent time points $k' > k$. The exact baseline distributions, model equations, and parameter values for the data generation are detailed in Section \ref{sec:supp_sim_details} of the Supplementary Material. 

As part of our simulation framework, we estimated four models: propensity score models for treatment assignment (and, in turn, compatibility) and censoring, and outcome models for the event of interest and the competing event. To evaluate estimator performance under varying conditions of model reliability, we implemented three distinct estimation scenarios. 

In the nonparametric scenario, models were estimated using the SuperLearner algorithm, which is an ensemble machine learning method that generates a weighted average prediction from a collection of learners using cross-validation \citep{polley_superlearner_2024}. We implemented three folds for the internal cross-validation step. Our SuperLearner library included multivariate adaptive regression splines (MARS), random forest, a neural network, a generalized linear model (GLM) with all two-way interactions, and penalized regression. For both propensity score models, both the generalized linear model with two-way interactions and penalized regression were correctly specified based on our data-generating mechanism. In contrast, none of the learners in the library exactly matched the outcome model, so it was likely misspecified across all candidates. The second scenario utilized correctly specified parametric GLMs for the propensity scores and event-specific hazards \eqref{eq:h1} and \eqref{eq:h2}, incorporating the full set of static and time-varying covariates that were generated. Although these models included the full covariate set, complete specification of the hazards does not guarantee a correctly specified model for the cumulative incidence. Because the estimand is a nonlinear function of the event-specific hazards, it cannot be represented as a linear combination of predictors on the logit link scale. This structural discrepancy ensures that the overall g-formula is misspecified even when the underlying hazard models are correctly specified, since the functional form of the nested conditional means is not analytically available. Finally, the third scenario introduced a misspecified parametric GLM approach. In contrast to the structural misspecification of the second scenario, in this scenario, we deliberately restricted the covariate set for all models (propensity scores, censoring, and outcome) to include only one static covariate and one time-varying covariate. This omission of predictors allowed us to assess estimator performance under considerable model misspecification.

We compared the performance of three estimators: the proposed EIF-based estimator in (\ref{eq:EIF}), the IPW Aalen-Johansen (IPW-AJ) estimator \citep{morzywolek_estimation_2022}, and the g-formula \citep{robins_new_1986}. Performance was evaluated based on bias, coverage, and the Root Integrated Mean Squared Error (RIMSE). For $K$ discrete time points,
$$\text{RIMSE} = \sqrt{\frac{1}{K} \sum_{t=1}^K \text{MSE}(t)}\,,$$
where $\text{MSE}(t) = \frac{1}{B} \sum_{b=1}^{B} (\hat{\theta}_{b,t} - \theta_{0,t})^2$ is the mean squared error at time $t$ across $B$ simulation replications, and $\hat{\theta}_{b,t}$ and $\theta_{0,t}$ are the estimated and true cumulative incidence, respectively.

\subsection{Results}\label{sec:sims-results}

In Figure \ref{fig:RIMSE}, we present the RIMSE for the cumulative incidence for each estimator under both regimes ($d^1$ and $d^2$) across the three distinct model-specification scenarios as a function of the sample size $N$. For the IPW-AJ and EIF-based estimators in the nonparametric and correctly specified parametric scenarios, the negative linear relationship between $\log_2(\text{RIMSE})$ and $\log_2(\text{N})$ indicates $\sqrt{n}$-consistency for both methods. On the other hand, the g-formula exhibits slower convergence, reflecting the structural misspecification of the outcome process discussed in Section \ref{sec:sims-dgp}. In the misspecified parametric scenario, in which we have omitted covariates, the IPW-AJ and EIF-based estimators continue to outperform the g-formula, but no longer exhibit $\sqrt{n}$-convergence. Instead, the RIMSE for these estimators plateaus at larger sample sizes, particularly under $d^1$, suggesting that intentional model misspecification introduced an asymptotic bias that cannot be reduced by increasing sample size. Overall, while the IPW-AJ estimator demonstrates greater numerical stability in smaller samples, the EIF-based estimator performs comparably or marginally better as sample size increases. 

\begin{figure}[!htb]
\centering
  \includegraphics[width=0.95\textwidth]{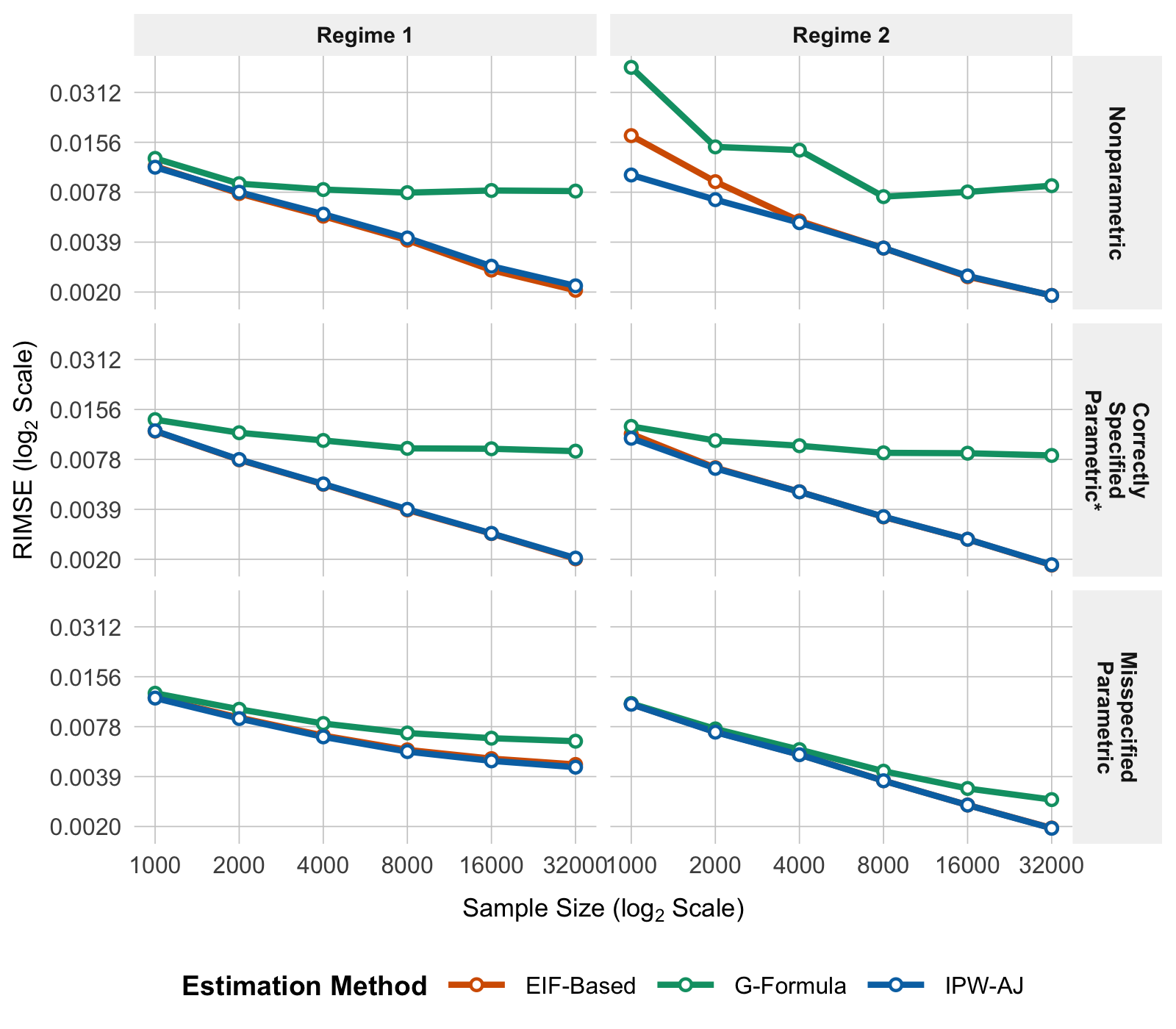} 
  \caption{\textbf{Root integrated mean squared error across regimes and estimation scenarios as a function of sample size, on the $\log_2$ scale.} $^*$Correctly specified refers to the individual propensity score and hazard models; the overall g-formula remains structurally misspecified due to the nonlinear functional form of the cumulative incidence.} 
  \label{fig:RIMSE}
\end{figure}

Figure \ref{fig:Reg1_AvgErr} depicts the bias, estimated by the average error, in the cumulative incidence over time under two different sample sizes and across the estimation scenarios for $d^1$. Consistent with theory, the EIF-based estimator maintains near-zero bias in both the nonparametric and correctly specified parametric scenarios, which is particularly evident as the sample size increases to $N = 32,000$. In the nonparametric and misspecified parametric scenarios, all three estimators exhibited a distinct increase in bias at later time points ($k \geq 6$), likely due to data sparsity as fewer individuals remain compatible with the regime toward the end of follow-up. In the misspecified parametric scenario, the EIF-based and IPW-AJ estimators have reduced bias relative to the g-formula, but do not achieve near-zero bias due to the omitted confounders. Despite the upward trend at later time points, the results under the nonparametric and correctly specified parametric approaches show that the EIF-based estimator achieves minimal bias as the sample size increases, provided at least one set of nuisance models is correctly specified. 

\begin{figure}[!htb]
\centering
    \includegraphics[width=0.95\textwidth]{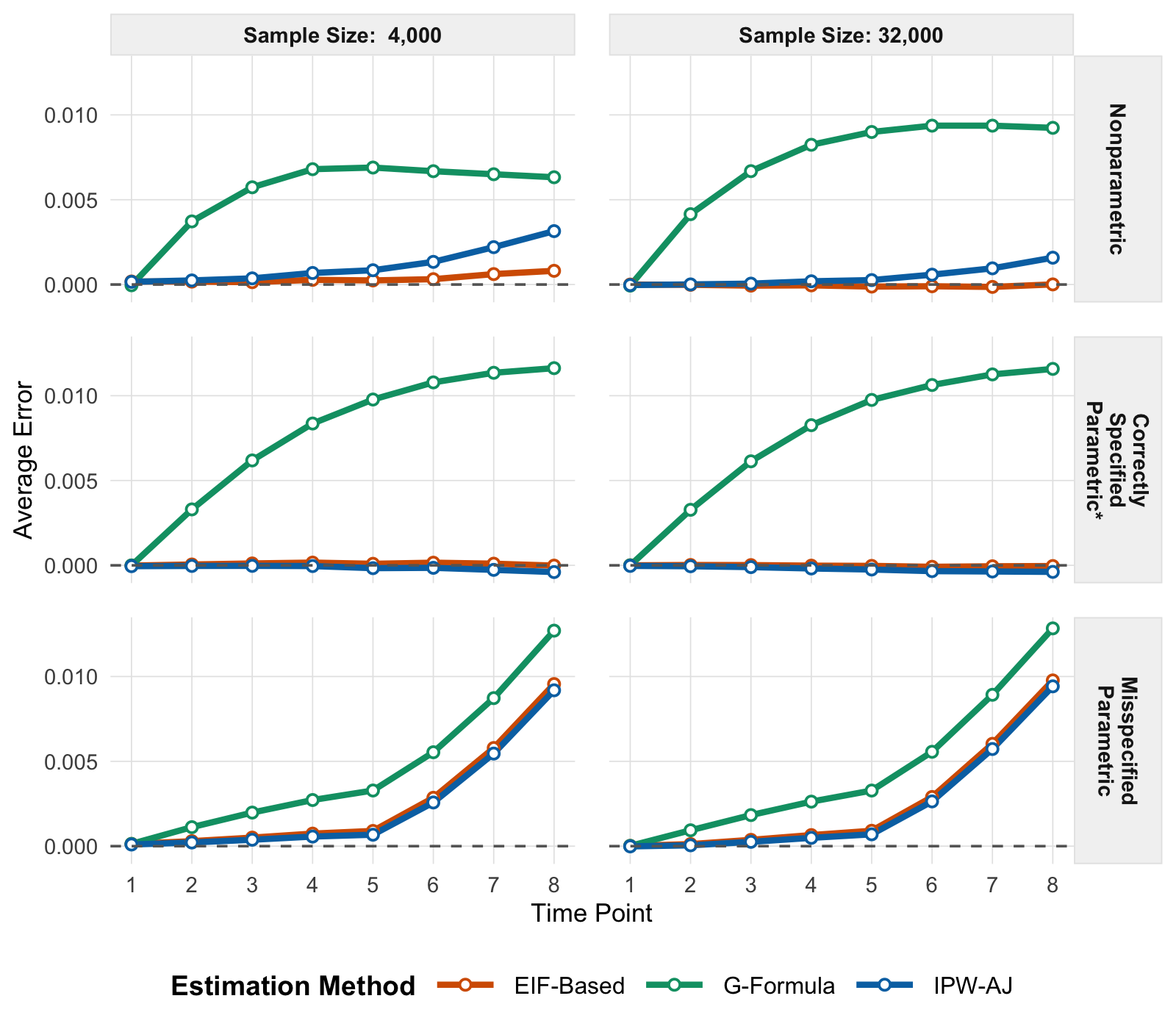}
    \caption{
        \textbf{Bias in cumulative incidence under Treatment Regime 1.} The dashed horizontal line indicates zero error (unbiasedness).
    }
    \label{fig:Reg1_AvgErr}
\end{figure}

Figure \ref{fig:Reg2_AvgErr} presents the bias across time and estimation scenarios for $d^2$. Across all three estimation scenarios, the IPW-AJ and EIF-based estimators exhibit near-zero bias throughout the follow-up period. In contrast, the g-formula exhibits systematically higher bias, with consistent trends across both sample sizes. In the misspecified scenario, all estimators demonstrate sensitivity to omitted confounders, trending toward negative bias by the end of follow-up. At $N = 4,000$, the EIF-based estimator shows noticeable finite-sample instability in the nonparametric scenario, characterized by a downward trend in bias at later time points ($k \geq 6$). However, when the sample size increases to $N = 32,000$, this instability largely resolves. In the nonparametric and correctly specified scenarios for the higher sample size, the EIF-based estimator is the most accurate, with bias closest to zero at all time points. These results suggest that while the EIF-based estimator requires a large enough sample size to stabilize under complex parametric model specification, it ultimately yields higher asymptotic accuracy by more effectively targeting the true cumulative incidence.

\begin{figure}[!htb]
\centering
    \includegraphics[width=0.95\textwidth]{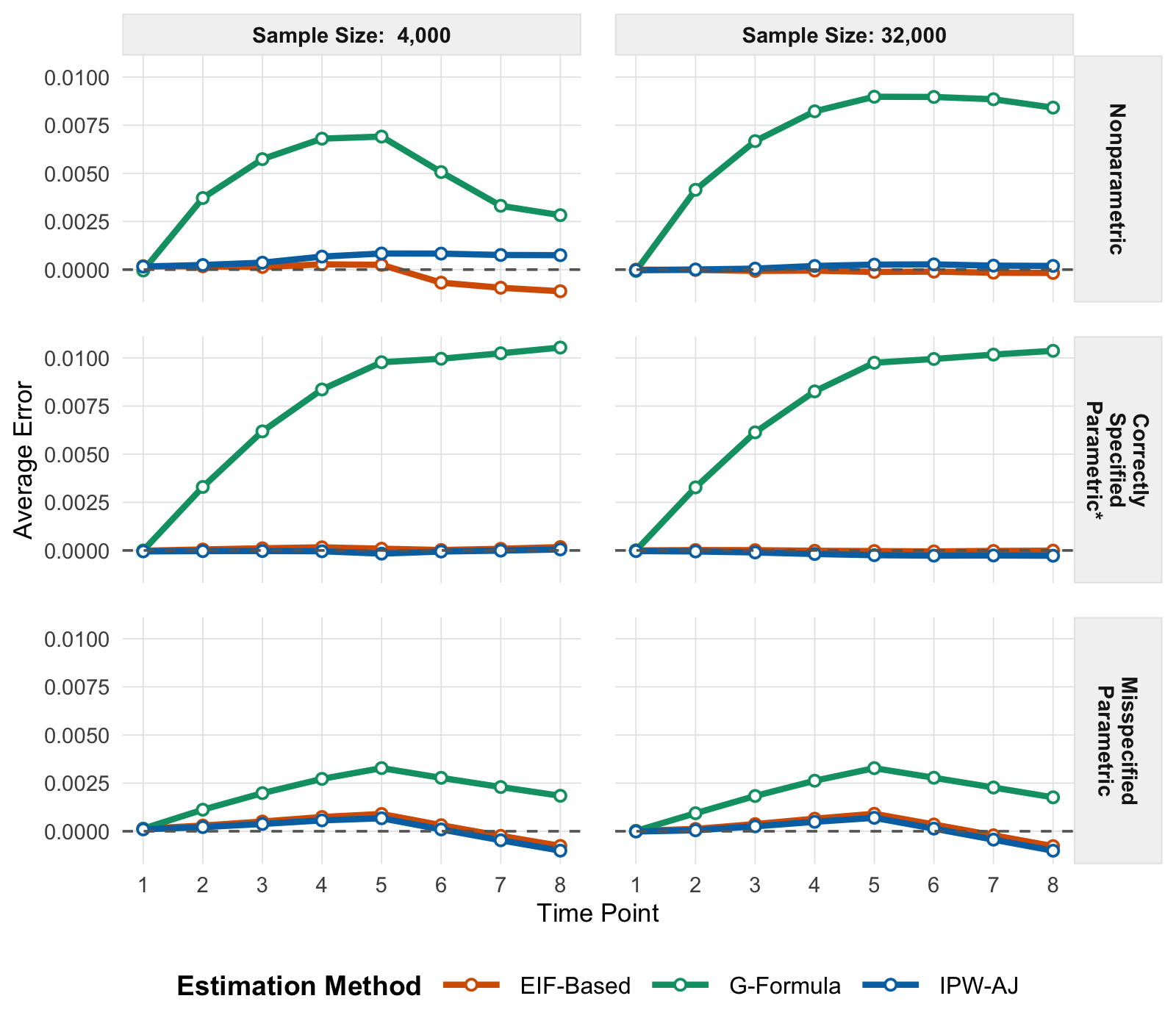}
    \caption{
        \textbf{Bias in cumulative incidence under Treatment Regime 2.} The dashed horizontal line indicates zero error (unbiasedness). 
    }
    \label{fig:Reg2_AvgErr}
\end{figure}

Figure \ref{fig:EIF_coverage} presents the 95\% coverage probabilities over time for the EIF-based estimator as a function of sample size under both regimes. Under $d^2$, the estimator maintains nominal or slightly conservative coverage throughout the follow-up period across all estimation scenarios, correctly characterizing uncertainty even under the misspecification in the third scenario. In contrast, under $d^1$, there is a sharp decrease in coverage at later time points as sample size increases for the misspecified parametric case. Because the confidence intervals shrink at larger sample sizes, but are centered on a biased estimate (Figure \ref{fig:Reg1_AvgErr}), they eventually fail to encompass the true value. This highlights that when persistent bias is present, larger samples can paradoxically increase the likelihood of excluding the true parameter value.

\begin{figure}[!htb]
\centering
\includegraphics[width=0.95\textwidth]{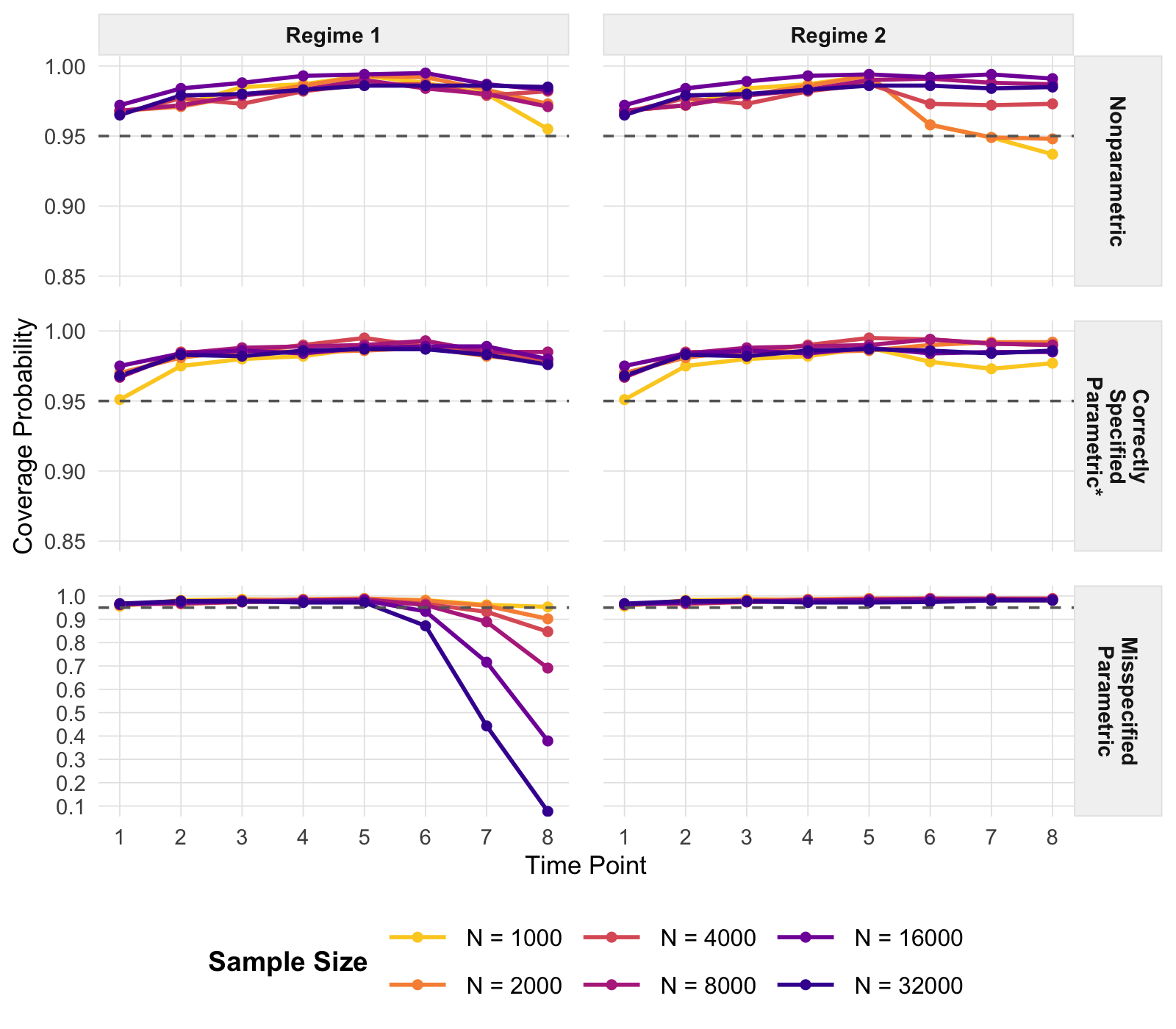}
    \caption{
        \textbf{Coverage of cumulative incidence using the EIF-based estimator across regimes and estimation scenarios.} The dashed line represents nominal 95\% coverage.
    }
    \label{fig:EIF_coverage}
\end{figure}

\section{Analysis of Time to Bisphosphonate Discontinuation}
\label{sec:rwe}

To further demonstrate the practical utility of our proposed EIF-based estimator, we applied our framework to a large longitudinal dataset from the Veterans Health Administration (VHA). This application investigates the fracture risk associated with continuing versus discontinuing oral bisphosphonate (BP) treatment among older male veterans who have already completed an initial 3-year course of treatment for osteoporosis.

\subsection{Design and Estimation}

This analysis uses real-world electronic health record data spanning the Veterans Affairs Health Care System, structurally linked with Fee for Service (FFS) Medicare and Medicare Advantage records to ensure comprehensive longitudinal capture of patient outcomes. The cohort consists of 109,286 male Veterans aged 50 or older who had been prescribed three consecutive years of an oral BP treatment with an average proportion of days (PDC) covered of at least 50\% during that time. Patient clinical histories were discretized into consecutive 90-day intervals, denoted by $k = 1,\dots,16$, representing a 4-year follow-up window immediately following the completion of the initial 3-year treatment course. 

In this follow-up window, BP treatment for each study participant ($A_k$) was defined as $\text{PDC} \ge 50\%$ within each 90-day interval. The primary clinical event of interest ($Y_k$) was a composite outcome capturing any incident clinical fracture in time interval $k$, while death was the terminal competing risk ($Z_k$). We evaluated two clinically distinct, time-varying treatment regimes. The first regime was that of complete discontinuation after three years of treatment, represented by no treatment throughout the 16 intervals of follow-up. The second regime was that of sustained oral BP therapy throughout the entire follow-up period (resulting in a total of seven years of treatment). The time-static covariates in this dataset included age at time of first oral BP prescription or script, BMI, race, ethnicity, and prevalent clinical fractures at the time of cohort entry. For each interval, the time-varying covariates for each subject were estimated glomerular filtration rate (EGFR), VA frailty index, and the presence of any of the following medical conditions: rheumatoid arthritis, diabetes mellitus, malabsorption, chronic liver disease, depression, anxiety, dementia, and chronic obstructive pulmonary disease (COPD).

In a real-world setting, patients frequently switch treatments or routes of administration. For this cohort, patients were censored (i.e., $U_k = 0$) if they were lost to follow-up, initiated an intravenous (IV) BP, or were prescribed a non-bisphosphonate osteoporosis drug. At each time interval $k$, study participants were still compatible with either of the two regimes if they had followed the exact treatment sequence specified by that regime up to that point, and were uncensored, alive, and fracture-free. Supplementary Figure \ref{fig:compat_sample_sizes} depicts the number of individuals compatible with each regime over time.

To estimate the cumulative incidence of clinical fractures in older male veterans under each of these two treatment regimes, we compared four different estimation methods. We first computed an unweighted regime-specific Aalen-Johansen estimator by using an indicator to identify the subset of regime-compatible, uncensored subjects at each time point. These individuals were given a weight of 1 in our estimation, while all others were given a weight of 0. To adjust for time-varying confounding and informative censoring, we used the three strategies previously evaluated in our simulation study: the IPW Aalen-Johansen (IPW-AJ) estimator, the g-formula, and the proposed EIF-based estimator. All underlying nuisance parameter models, including the propensity scores, censoring models, and g-formula conditional expectations, were estimated using high-dimensional random forests.

\subsection{Results}

The estimated cumulative incidence curves for the outcome of any clinical fracture under the two treatment regimes are presented in Figure \ref{fig:eif_cifs_for_bp_dtrs}. The primary analysis contrasts our proposed doubly robust EIF-based estimator with the naive, unweighted Aalen-Johansen estimates across the 16 follow-up intervals and between the two regimes. Throughout the four years of follow-up, the unweighted estimates (represented by the dashed lines) are substantially lower than the EIF-based estimates (represented by the solid lines) for both the discontinuation and continuation regimes. This notable downward separation is likely indicative of strong time-varying confounding and informative censoring within the cohort, where failing to adjust for the treatment and censoring mechanisms may have led to a severe underestimation of fracture risk.

\begin{figure}[!htb]
    \centering
    \includegraphics[width=0.95\textwidth]{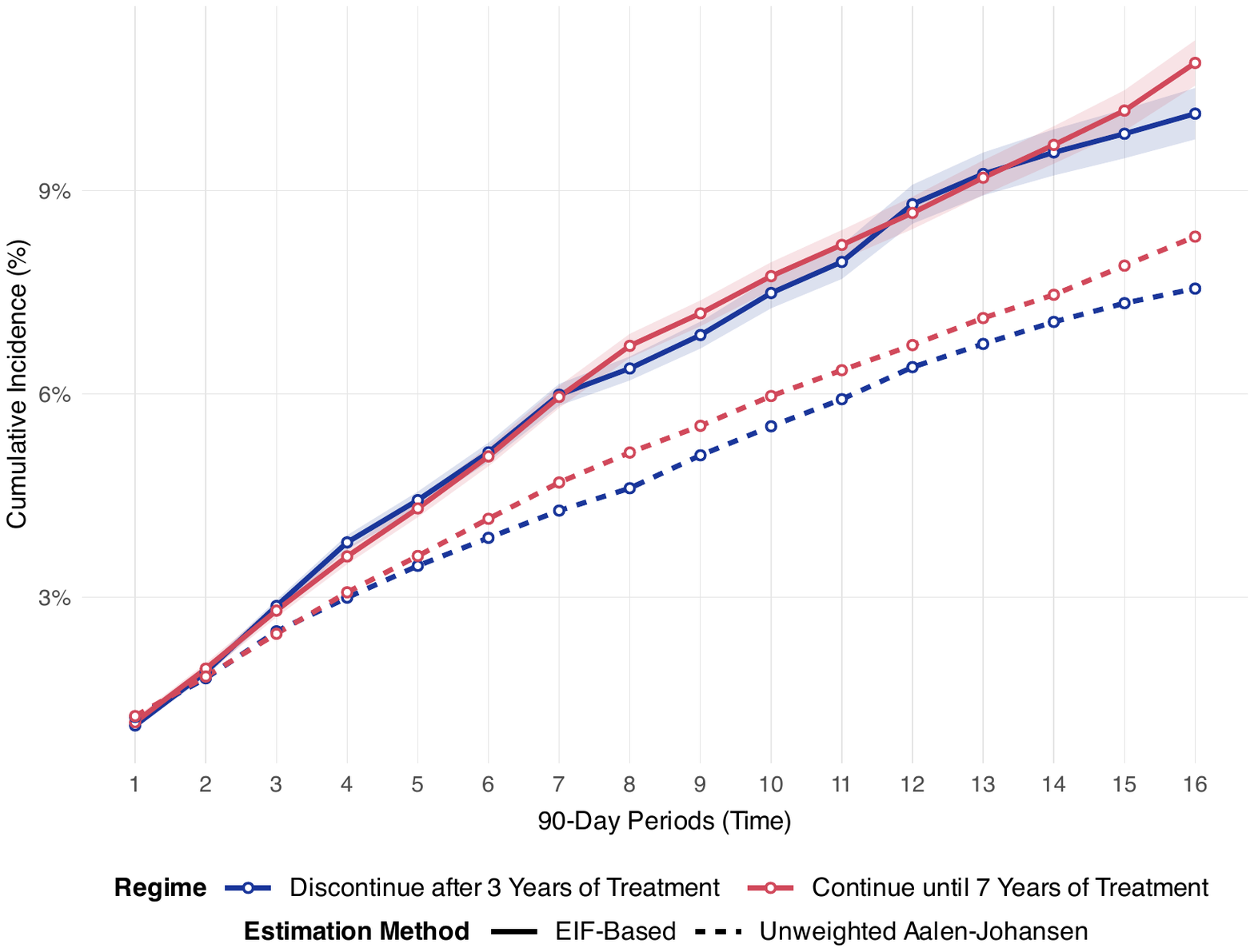}
    \caption{\textbf{Estimated cumulative incidence of incident clinical fractures under two different treatment regimes throughout four years of follow-up using our proposed EIF-based estimator.} The dashed lines represent the cumulative incidence using an unweighted Aalen-Johansen estimator.}
    \label{fig:eif_cifs_for_bp_dtrs}
\end{figure}

Additionally, we see that the difference in cumulative incidence under the two regimes is reduced after adjustment with our EIF-based method. Under the proposed method, the estimated cumulative incidence of fractures at the end of the 4-year follow-up window (interval 16) is approximately 10.13\% for the group discontinuing oral BP therapy after three years, compared to approximately 10.88\% for those continuing BP therapy for an additional four years. The closely overlapping 95\% confidence intervals indicate that among older male patients who have already successfully completed an initial 3-year course of oral BP treatment, there is no evidence of long-term clinical benefit, nor risk, associated with extending oral BP treatment for an additional four years. A secondary analysis comparing the cumulative incidence estimated using the standard causal approaches to the EIF-based estimator is provided in Supplementary Figure \ref{fig:three_method_comparison}.

\section{Discussion}
\label{sec:discussion}

The estimation of long-term risks for adverse events under complex, time-varying treatment strategies remains a critical methodological challenge for health outcomes research. Several causal inference approaches have been developed to address time-varying exposures, but no estimators achieve nonparametric efficiency while maintaining robustness against model misspecification in the presence of competing risks. While the g-formula, introduced by \citet{robins_analytic_1993}, can estimate cumulative incidence under time-dependent interventions, its reliance on a correctly specified outcome process that includes difficult-to-model conditional expectations of functions of outcome models makes it vulnerable to functional form misspecification. As demonstrated in our simulation study (Section \ref{sec:sims}), even with a complete set of covariates, the g-formula exhibits persistent bias because the recursive nature of the cumulative incidence is challenging to model well. In contrast, IPW estimators, such as the IPC weighted Aalen-Johansen estimator developed by \citet{morzywolek_estimation_2022}, can achieve $\sqrt{n}$-consistency, but rely on the correct specification of the treatment and censoring mechanisms and use significantly less data as compatibility sample sizes decay over time. 

The development of our novel nonparametric efficient estimator for the cumulative incidence of an adverse event under DTRs addresses these limitations. By deriving the EIF and validating its pathwise differentiability, we provide a rigorous framework for $\sqrt{n}$-consistent and asymptotically normal inference \citep{kennedy_semiparametric_2023}. Our simulations across model-specification scenarios provide critical insights into how these theoretical properties translate to finite-sample performance. Both the EIF-based and IPW-AJ estimators demonstrated negligible bias when nuisance models were correctly specified, while the g-formula demonstrated significant bias due to outcome model misspecification (Figures \ref{fig:Reg1_AvgErr} and \ref{fig:Reg2_AvgErr}). 

Our results when using a misspecified parametric estimation approach do, however, provide a clear look at the limits of double robustness. The sharp decline in coverage observed under Regime 1 (Figure \ref{fig:EIF_coverage}) highlights that the efficiency of the EIF-based estimator at larger sample sizes can lead to a total loss of nominal coverage if model misspecification introduces persistent asymptotic bias. Beyond model specification, the practical utility of our approach remains bound to our identification assumptions. Long follow-up periods and complex treatment histories can result in positivity violations, leading to extreme IPC weights. These operational challenges highlight the need for large longitudinal data resources, and further motivates the incorporation of flexible data-adaptive methods (such as SuperLearner) to ensure stable estimation. Additionally, our framework relies on the assumption of no unmeasured confounding for treatment, censoring, and the competing risk across all intervals. Despite these finite-sample, identification, and misspecification-related limitations, the sequential double robustness of our EIF-based estimator provides a crucial methodological advancement for estimating cumulative incidence in observational settings.

One potential extension of this work is the incorporation of covariate information, such as demographic factors, baseline health status, and comorbidity data, into an estimator of \textit{conditional} effects, allowing for the estimation of cumulative incidence in subpopulations. Such an approach would not only enhance our understanding of the distribution of these outcomes, but may also help identify high-risk subgroups that might otherwise be overlooked when relying solely on population-level estimates. Additional work could further investigate interactions among covariates, exploring how combinations of traits may amplify or mitigate the risk of adverse outcomes. Tailored estimates from such analyses could be especially valuable in clinical settings, supporting healthcare providers in crafting individualized care plans for patients, or in public health interventions, where they could inform risk-based screening strategies and treatment guidelines for high-risk subpopulations.

Our proposed EIF-based estimator offers a rigorous, generalizable, and clinically interpretable framework for estimating cumulative risks under dynamic treatment regimes. By explicitly handling time-varying exposure, a competing event, and censoring in a sequentially doubly robust manner, it is well-suited to leverage the types of extensive longitudinal data that are now available in modern healthcare settings. While our motivating example focuses on the critical public health challenge of optimizing BP treatment and holiday patterns to minimize risk of adverse events, this methodological approach is readily extendable to other clinical settings where similar challenges arise. 

\section{Disclosure Statement}\label{disclosure-statement}

The authors report there are no competing interests to declare.

\section{Acknowledgments}\label{acknowledgments}

This work was supported by grant R01AG079118 from the National Institute on Aging (NIA) and the Office of the Director at the National Institutes of Health. The views expressed in this work are those of the authors and do not necessarily reflect the position or policy of the Department of Veterans Affairs, the United States government or the Office of the National Institutes of Health.
Support for VA/CMS data provided by the Department of Veterans Affairs, Office of Research and Development, VA Information Resource Center (Project Numbers SDR 02-237 and 98-004).

\bibliography{Paper1_Resources_Truncated}

\pagebreak

\setcounter{section}{0}
\setcounter{figure}{0}
\setcounter{table}{0}

\renewcommand{\thesection}{S.\arabic{section}}
\renewcommand{\thefigure}{S.\arabic{figure}}
\renewcommand{\thetable}{S.\arabic{table}}

\begin{center}
    \Large{\textbf{SUPPLEMENTARY MATERIAL}}
\end{center}

\section{Notation, Conventions, and Preliminary Facts}\label{sec:supp-prelims}

This section summarizes some of the conventions and basic facts used throughout the proofs. Because the proofs repeatedly involve manipulations of conditional expectations given histories that combine random variables (covariates), events (survival, being uncensored, being compatible with the regime), we state these facts here.

\subsection{Conventions} Throughout the Supplementary Material we set $Y_{-1} = Z_{-1} = Y_0 = Z_0 = 0$, $U_{-1} = C^{(d)}_{-1} = 1$, and $W^{(d)}_j = 1$ for $j < 0$; empty sums equal zero and empty products equal one. As in the main text, $\pi^{C^{(d)}}_j = \pi^U_j = 1$ whenever $\overline{Y}_j \neq 0$ or $\overline{Z}_j \neq 0$. For a function $f$ of the observed data, $\|f\|^2 = \int f^2 \, \mathrm{d}\P$. Because all lemmas in Sections \ref{sec:identification} and \ref{sec:supp-EIF} involve a single arbitrary regime $d$, we drop the superscript $(d)$ on $W_t$, $C_t$, and $\pi^C_t$ when unambiguous. We abbreviate
\begin{equation*}
    \Acal_t = \prod_{j=0}^{t-1}(1-Y_j)(1-Z_j), \qquad \Bcal_t = (1-Y_t)(1-Z_t), \qquad \Bcal^*_k = Y_k(1-Z_k),
\end{equation*}
and note that, under the temporal ordering $O_j = \{Z_j, Y_j, X_j, A_j, U_j\}$, the collection of all variables realized strictly before $(Z_t, Y_t)$ is $\left(\overline{X}_{t-1}, \overline{Y}_{t-1}, \overline{Z}_{t-1}, \overline{A}_{t-1}, \overline{U}_{t-1}\right)$. Several arguments below condition on this full history, which we always write out explicitly; conditioning on the covariate history $\overline{X}_{t-1}$ by itself is generally not sufficient, as the tower-property steps we use also use outcome, treatment, and censoring information. Finally, define the at-risk event
\begin{equation*}
    E_t = \left\{\overline{Y}_{t-1} = 0,\; \overline{Z}_{t-1} = 0,\; \overline{U}_{t-1} = 1,\; \overline{C}_{t-1} = 1\right\}.
\end{equation*}
Because the regime $d$ assigns treatment as a deterministic function of the observed history, on the event $\{\overline{Y}_{t-1} = 0, \overline{Z}_{t-1} = 0\}$ the compatible treatment path is determined recursively by the covariates alone; we write $\overline{d}_{t-1}(\overline{x}_{t-1})$ for this path, that is, $d_0(x_0)$, $d_1\left(\overline{x}_1, \overline{d}_0(x_0)\right)$, and so on, with all outcome arguments set to zero.

\subsection{Basic Facts Used Throughout}

\begin{enumerate}
    \item[(P1)] \textit{Scores are conditionally centered.} For any partition of the data into a variable $V$ and a conditioning set $B$, the score of the conditional factor satisfies $\E\left[\ell'_\varepsilon(V \mid B; 0) \mid B\right] = 0$, because $\E_\varepsilon\left[\partial_\varepsilon \log p_\varepsilon(V \mid B) \mid B\right] = \partial_\varepsilon \int p_\varepsilon(v \mid B)\, \mathrm{d}v = 0$ at $\varepsilon = 0$ under the usual regularity conditions that allow for the interchange of derivative and integral. Consequently, for any integrable $Q$ that is a function of $B$,
    \begin{equation*}
        \E\left[Q \, \ell'_\varepsilon(V \mid B; 0)\right] = \E\left[Q \, \E\left\{\ell'_\varepsilon(V \mid B; 0) \mid B\right\}\right] = 0.
    \end{equation*}
    \item[(P2)] \textit{Conditioning on equivalent information on an event.} Let $B$ be a collection of observed variables, let $B_0$ be a sub-collection of $B$, and let $E$ be an event determined by $B$. Suppose that, on the event $E$, every variable in $B$ is equal to some fixed function of $B_0$. Then, for any integrable $V$,
    \begin{equation*}
        \E[V \mid B]\,\mathbbm{1}(E) = \E[V \mid B_0, E]\,\mathbbm{1}(E) \quad \text{a.s.},
    \end{equation*}
    where $\E[V \mid B_0, E]$ denotes the conditional expectation given $B_0$ and the occurrence of $E$. In words, once $E$ is known to have occurred, the value of $B$ contains no information beyond the value of $B_0$, so the two conditional expectations agree.

    \textit{Proof.} Both sides are functions of $B$. For any bounded function $h$ of $B$, the hypothesis implies $h(B)\mathbbm{1}(E) = g(B_0)\mathbbm{1}(E)$ for some function $g$. Hence
    \begin{align*}
        \E\left[\E\{V \mid B\}\, h(B)\mathbbm{1}(E)\right] &= \E\left[V h(B)\mathbbm{1}(E)\right] = \E\left[V g(B_0)\mathbbm{1}(E)\right]\\
        &= \E\left[\E\{V \mid B_0, E\}\, g(B_0)\mathbbm{1}(E)\right]\\
        &= \E\left[\E\{V \mid B_0, E\}\, h(B)\mathbbm{1}(E)\right],
    \end{align*}
    where the first and third equalities follow from a basic property of conditional expectation ($h(B)\mathbbm{1}(E)$ is a function of $B$, and $g(B_0)\mathbbm{1}(E)$ is a function of $B_0$ and $\mathbbm{1}(E)$). Since $h$ is arbitrary, the two sides agree almost surely on $E$. \hfill$\square$

    In particular, on $E_t$ the variables $\overline{Y}_{t-1}, \overline{Z}_{t-1}, \overline{U}_{t-1}$ are constant and $\overline{A}_{t-1} = \overline{d}_{t-1}(\overline{X}_{t-1})$ is a function of $\overline{X}_{t-1}$, so conditioning on {$\left(\overline{X}_{t-1}, \overline{Y}_{t-1}, \overline{Z}_{t-1}, \overline{A}_{t-1}, \overline{U}_{t-1}\right)$} and conditioning on $(\overline{X}_{t-1}, E_t)$ agree on $E_t$.
    \item[(P3)] \textit{The propensities are the relevant conditional means on compatible histories.} For each $j \geq 0$, define the events $F_j = \left\{\overline{C}_j = 1, \overline{U}_{j-1} = 1, \overline{Y}_j = 0, \overline{Z}_j = 0\right\}$ and $F'_j = \left\{\overline{C}_{j-1} = 1, \overline{U}_{j-1} = 1, \overline{Y}_j = 0, \overline{Z}_j = 0\right\}$. Then
    \begin{align*}
        \E\left[U_j \mid \overline{X}_j, \overline{Y}_j, \overline{Z}_j, \overline{A}_j, \overline{U}_{j-1}\right] \mathbbm{1}\left(F_j\right) &= \pi^U_j\, \mathbbm{1}\left(F_j\right),\\
        \E\left[\mathbbm{1}(A_j = d_j) \mid \overline{X}_j, \overline{Y}_j, \overline{Z}_j, \overline{A}_{j-1}, \overline{U}_{j-1}\right] \mathbbm{1}\left(F'_j\right) &= \pi^{C}_j\, \mathbbm{1}\left(F'_j\right).
    \end{align*}
    \textit{Proof.} Both displays are instances of (P2). For the first, take $E = F_j$: on $E$, the treatment history $\overline{A}_j$ equals the deterministic path $\overline{d}_j(\overline{X}_j)$ and the outcome and censoring histories are constant, so every conditioning variable is a function of $\overline{\Hcal}_j = \left(\overline{Z}_j, \overline{Y}_j, \overline{X}_j\right)$ on $E$; by (P2), the conditional expectation of $U_j$ given the full list agrees on $E$ with the conditional probability given $\left(\overline{\Hcal}_j, E\right)$, which is the definition \eqref{eq:pi_U} of $\pi^U_j$. The second display is analogous, with $E = F'_j$, on which $C_j = \mathbbm{1}(A_j = d_j)$ and $\pi^{C}_j$ from \eqref{eq:pi_C} is the conditional probability of that indicator. \hfill$\square$
    \item[(P4)] \textit{Functions of the past.} The following facts hold by construction: (i) $\Acal_t$, $\mathbbm{1}(E_t)$, and $W_{t-1}$ are functions of the history $\left(\overline{X}_{t-1}, \overline{Y}_{t-1}, \overline{Z}_{t-1}, \overline{A}_{t-1}, \overline{U}_{t-1}\right)$; for $W_{t-1}$ this holds because each $C_j = \mathbbm{1}\{\overline{A}_j = \overline{d}_j(\overline{\Hcal}_j, \overline{A}_{j-1})\}$, each $U_j$, and each propensity $\pi^C_j$, $\pi^U_j$ (a function of $\overline{\Hcal}_j$ and the treatment history) with $j \leq t-1$ is a function of these variables; (ii) $G_t^{(k)}$ and $R_t^{(k)}$ are functions of $\overline{X}_{t-1}$, by the recursive definition \eqref{eq:G_recursive}; (iii) a score $\ell'_\varepsilon(V \mid B; 0)$ is a function of $(V, B)$; (iv) $W_{t-1}\Acal_t = W_{t-1}\Acal_t \mathbbm{1}(E_t)$, because the numerator of $W_{t-1}$ contains $\prod_{j=0}^{t-1} U_j C_j$ and $\Acal_t$ contains the survival indicators.
    \item[(P5)] \textit{One-step conditional means of the sequential regressions.} For $t = 1, \dots, k-1$,
    \begin{align*}
        \E\left[\Bcal_t\, G_{t+1}^{(k)} \,\middle|\, {\overline{X}_{t-1}, \overline{Y}_{t-1}, \overline{Z}_{t-1}, \overline{A}_{t-1}, \overline{U}_{t-1}}\right] \mathbbm{1}(E_t) &= G_t^{(k)}\, \mathbbm{1}(E_t), \qquad \text{and}\\
        \E\left[\Bcal^*_k \,\middle|\, {\overline{X}_{k-1}, \overline{Y}_{k-1}, \overline{Z}_{k-1}, \overline{A}_{k-1}, \overline{U}_{k-1}}\right] \mathbbm{1}(E_k) &= G_k^{(k)}\, \mathbbm{1}(E_k).
    \end{align*}
    \textit{Proof.} By (P2), on $E_t$ we may replace conditioning on {the full history $\left(\overline{X}_{t-1}, \overline{Y}_{t-1}, \overline{Z}_{t-1}, \overline{A}_{t-1}, \overline{U}_{t-1}\right)$} by conditioning on $(\overline{X}_{t-1}, E_t)$. Peeling off the variables $(Z_t, Y_t, X_t)$ in temporal order,
    \begin{align*}
        \E\left[(1-Z_t)(1-Y_t)\, G_{t+1}^{(k)} \,\middle|\, \overline{X}_{t-1}, E_t\right] &= \Pr\left(Z_t = 0 \mid \overline{X}_{t-1}, E_t\right)\\
        &\qquad \times \E\left[(1-Y_t)\, G_{t+1}^{(k)} \,\middle|\, Z_t = 0, \overline{X}_{t-1}, E_t\right]\\
        &= \nu_{1-Z_t}\left(\overline{X}_{t-1}\right) \Pr\left(Y_t = 0 \mid \overline{Z}_t = 0, \overline{X}_{t-1}, E_t\right)\\
        &\qquad \times \E\left[G_{t+1}^{(k)} \,\middle|\, \overline{Y}_t = 0, \overline{Z}_t = 0, \overline{X}_{t-1}, E_t\right]\\
        &= \nu_{1-Z_t}\left(\overline{X}_{t-1}\right) \mu_{1-Y_t}\left(\overline{X}_{t-1}\right) \E\left[G_{t+1}^{(k)} \,\middle|\, \overline{\Ocal}_t\right]\\
        &= \E\left[R^{(k)}_t\, G_{t+1}^{(k)} \,\middle|\, \overline{\Ocal}_t\right] = G_t^{(k)},
    \end{align*}
    where the first two equalities condition on the discrete events $\{Z_t = 0\}$ and $\{Y_t = 0\}$ (the complementary events contribute zero because of the factors $1-Z_t$ and $1-Y_t$); the third equality identifies the two conditional probabilities as $\nu_{1-Z_t}$ and $\mu_{1-Y_t}$ from \eqref{eq:mean_of_Y_k}-\eqref{eq:mean_of_1-Z_k} and recognizes $\left\{\overline{Y}_t = 0, \overline{Z}_t = 0, \overline{X}_{t-1}, E_t\right\}$ as the conditioning set $\overline{\Ocal}_t$; the fourth pulls the factor $R^{(k)}_t = \mu_{1-Y_t}\nu_{1-Z_t}$, a function of $\overline{X}_{t-1}$, inside the conditional expectation; and the last is the definition \eqref{eq:G_recursive}. The second display of (P5) follows in the same way from $\E\left[Y_k(1-Z_k) \mid \overline{X}_{k-1}, E_k\right] = \Pr(Z_k = 0 \mid \overline{X}_{k-1}, E_k)\,\E\left[Y_k \mid \overline{Z}_k = 0, \overline{X}_{k-1}, E_k\right] = \nu_{1-Z_k}\left(\overline{X}_{k-1}\right)\mu_{Y_k}\left(\overline{X}_{k-1}\right) = G_k^{(k)}$. \hfill$\square$
\end{enumerate}

\section{G-formula for the Cumulative Incidence}\label{sec:identification}

To prove Theorem \ref{theorem:g-formula_identification}, we first establish a series of lemmas that allow us to rewrite the cumulative incidence into expectations using the observed data. Lemma \ref{lemma:young_lemma4} provides the algebraic expansion of the incidence at time $k$ as a function of the cause-specific hazards. 

\begin{lemma}\label{lemma:young_lemma4}
For any $k \in \{1, \dots, K\}$, under the conventions of Section \ref{sec:supp-prelims} (in particular $Y_{-1} = Z_{-1} = 0$ and $W_{-1} = 1$, so that every term below with index $-1$ is well defined),
\begin{align*}
{h}_{k}^{(Y)}(d)&\left\{1-{h}_{k}^{(Z)}(d)\right\}\prod_{j=1}^{k-1}\left\{1-{h}_{j}^{(Y)}(d)\right\}\left\{1-{h}_{j}^{(Z)}(d)\right\}\\
&= \E\left[Y_{k}(1-Z_{k})(1-Y_{k-1})W^{(d)}_{k-1}\right]\\
&\quad \times \prod_{j=1}^{k} \left[\frac{\E\left[(1-Y_{j-1})(1-Z_{j-1})W^{(d)}_{j-1}\right] - \E\left[Z_{j}(1-Y_{j-1})(1-Z_{j-1})W^{(d)}_{j-1}\right]}{\E\left[(1-Z_{j})(1-Y_{j-1})W^{(d)}_{j-1}\right]}\right]\\
&\quad \times \prod_{j=1}^{k} \left[\frac{\E\left[(1-Z_{j-1})(1-Y_{j-2})W^{(d)}_{j-2}\right] - \E\left[Y_{j-1}(1-Z_{j-1})(1-Y_{j-2})W^{(d)}_{j-2}\right]}{\E\left[(1-Y_{j-1})(1-Z_{j-1})W^{(d)}_{j-1}\right]}\right].
\end{align*}
\end{lemma}

\begin{proof}[Proof of Lemma \ref{lemma:young_lemma4}]
Following the framework of \cite{young_causal_2020}, we begin by substituting the definitions of the discrete-time hazard functions $h_k^{(Y)}(d)$ and $h_k^{(Z)}(d)$ into the left-hand side. {We work with the more general product $\prod_{j=t}^{k-1}$ for an arbitrary starting index $t \in \{1, \dots, k-1\}$ and fit it to $t = 1$ at the end of the proof; each hazard is replaced by its weighted-expectation representation from Theorem \ref{theorem:weighting_identification}.} By definition of these conditional expectations:
\allowdisplaybreaks
\begin{align*}
    &{h}_{k}^{(Y)}(d)\left\{1-{h}_{k}^{(Z)}(d)\right\}\prod_{j=t}^{k-1}\left\{1-{h}_{j}^{(Y)}(d)\right\}\left\{1-{h}_{j}^{(Z)}(d)\right\} \\
    &= \frac{\E\left[Y_{k}(1-Z_{k})(1-Y_{k-1})W^{(d)}_{k-1}\right]}{\E\left[(1-Z_{k})(1-Y_{k-1})W^{(d)}_{k-1}\right]} 
    \left[1-\frac{\E\left[Z_{k}(1-Y_{k-1})(1-Z_{k-1})W^{(d)}_{k-1}\right]}{\E\left[(1-Y_{k-1})(1-Z_{k-1})W^{(d)}_{k-1}\right]}\right]\\
    &\quad \times \left[1-\frac{\E\left[Y_{k-1}(1-Z_{k-1})(1-Y_{k-2})W^{(d)}_{k-2}\right]}{\E\left[(1-Z_{k-1})(1-Y_{k-2})W^{(d)}_{k-2}\right]}\right] 
    \left[1-\frac{\E\left[Z_{k-1}(1-Y_{k-2})(1-Z_{k-2})W^{(d)}_{k-2}\right]}{\E\left[(1-Y_{k-2})(1-Z_{k-2})W^{(d)}_{k-2}\right]}\right]\\
    &\quad \times\cdots\times \left[1-\frac{\E\left[Y_{t}(1-Z_{t})(1-Y_{t-1})W^{(d)}_{t-1}\right]}{\E\left[(1-Z_{t})(1-Y_{t-1})W^{(d)}_{t-1}\right]}\right] 
    \left[1-\frac{\E\left[Z_{t}(1-Y_{t-1})(1-Z_{t-1})W^{(d)}_{t-1}\right]}{\E\left[(1-Y_{t-1})(1-Z_{t-1})W^{(d)}_{t-1}\right]}\right].
\end{align*}
Next, we simplify the terms within the brackets by finding a common denominator for each subtraction:
\allowdisplaybreaks
\begin{align*}
    &= \frac{\E\left[Y_{k}(1-Z_{k})(1-Y_{k-1})W^{(d)}_{k-1}\right]}{\E\left[(1-Z_{k})(1-Y_{k-1})W^{(d)}_{k-1}\right]}\\
&\quad \times \left[\frac{\E\left[(1-Y_{k-1})(1-Z_{k-1})W^{(d)}_{k-1}\right] - \E\left[Z_{k}(1-Y_{k-1})(1-Z_{k-1})W^{(d)}_{k-1}\right]}{\E\left[(1-Y_{k-1})(1-Z_{k-1})W^{(d)}_{k-1}\right]}\right]\\
&\quad \times \left[\frac{\E\left[(1-Z_{k-1})(1-Y_{k-2})W^{(d)}_{k-2}\right] - \E\left[Y_{k-1}(1-Z_{k-1})(1-Y_{k-2})W^{(d)}_{k-2}\right]}{\E\left[(1-Z_{k-1})(1-Y_{k-2})W^{(d)}_{k-2}\right]}\right]\\
&\quad \times \left[\frac{\E\left[(1-Y_{k-2})(1-Z_{k-2})W^{(d)}_{k-2}\right] - \E\left[Z_{k-1}(1-Y_{k-2})(1-Z_{k-2})W^{(d)}_{k-2}\right]}{\E\left[(1-Y_{k-2})(1-Z_{k-2})W^{(d)}_{k-2}\right]}\right]\\
&\quad \times\cdots\times \left[\frac{\E\left[(1-Z_{t})(1-Y_{t-1})W^{(d)}_{t-1}\right] - \E\left[Y_{t}(1-Z_{t})(1-Y_{t-1})W^{(d)}_{t-1}\right]}{\E\left[(1-Z_{t})(1-Y_{t-1})W^{(d)}_{t-1}\right]}\right]\\
&\quad \times \left[\frac{\E\left[(1-Y_{t-1})(1-Z_{t-1})W^{(d)}_{t-1}\right] - \E\left[Z_{t}(1-Y_{t-1})(1-Z_{t-1})W^{(d)}_{t-1}\right]}{\E\left[(1-Y_{t-1})(1-Z_{t-1})W^{(d)}_{t-1}\right]}\right].
\end{align*}

{The expression simplifies through a telescoping rearrangement of the denominators. Then list the bracketed factors in the order displayed so that the leading ratio is factor $0$ and the brackets alternate between competing-event terms (indices $k, k-1, \dots, t$) and event-of-interest terms (indices $k-1, k-2, \dots, t$). Because the overall expression is a single product, we may reassign denominators freely. We assign the denominator of the factor immediately preceding it to each bracketed factor, and have the final unassigned denominator, $\E[(1-Y_{t-1})(1-Z_{t-1})W^{(d)}_{t-1}]$, as a standalone factor at the end. This yields:}
\begin{align*}
&= \E\left[Y_{k}(1-Z_{k})(1-Y_{k-1})W^{(d)}_{k-1}\right]\\
&\quad \times \left[\frac{\E\left[(1-Y_{k-1})(1-Z_{k-1})W^{(d)}_{k-1}\right] - \E\left[Z_{k}(1-Y_{k-1})(1-Z_{k-1})W^{(d)}_{k-1}\right]}{\E\left[(1-Z_{k})(1-Y_{k-1})W^{(d)}_{k-1}\right]}\right]\\
&\quad \times \left[\frac{\E\left[(1-Z_{k-1})(1-Y_{k-2})W^{(d)}_{k-2}\right] - \E\left[Y_{k-1}(1-Z_{k-1})(1-Y_{k-2})W^{(d)}_{k-2}\right]}{\E\left[(1-Y_{k-1})(1-Z_{k-1})W^{(d)}_{k-1}\right]}\right]\\
&\quad \times \left[\frac{\E\left[(1-Y_{k-2})(1-Z_{k-2})W^{(d)}_{k-2}\right] - \E\left[Z_{k-1}(1-Y_{k-2})(1-Z_{k-2})W^{(d)}_{k-2}\right]}{\E\left[(1-Z_{k-1})(1-Y_{k-2})W^{(d)}_{k-2}\right]}\right]\\
&\quad \times\cdots\times \left[\frac{\E\left[(1-Z_{t})(1-Y_{t-1})W^{(d)}_{t-1}\right] - \E\left[Y_{t}(1-Z_{t})(1-Y_{t-1})W^{(d)}_{t-1}\right]}{{\E\left[(1-Y_{t})(1-Z_{t})W^{(d)}_{t}\right]}}\right]\\
&\quad \times \left[\frac{\E\left[(1-Y_{t-1})(1-Z_{t-1})W^{(d)}_{t-1}\right] - \E\left[Z_{t}(1-Y_{t-1})(1-Z_{t-1})W^{(d)}_{t-1}\right]}{\E\left[(1-Z_{t})(1-Y_{t-1})W^{(d)}_{t-1}\right]}\right]\\
&\quad \times \left[\frac{1}{\E\left[(1-Y_{t-1})(1-Z_{t-1})W^{(d)}_{t-1}\right]}\right].
\end{align*}
To obtain the final result, we set the initial condition at $t=1$. {Since $Y_0 = Z_0 = 0$, the base term is $\E\left[(1-Y_{0})(1-Z_{0})W^{(d)}_{0}\right] = \E\left[W^{(d)}_0\right] = \E\left[C_0 U_0 / (\pi^{C}_0 \pi^U_0)\right]$. Conditioning first on $(X_0, A_0)$ and applying the first display of (P3) with $j = 0$, then conditioning on $X_0$ and applying the second display of (P3), gives $\E\left[W^{(d)}_0\right] = 1$; note that this calculation is why the compatibility product in \eqref{eq:ipcw} must begin at $j = 0$. The trailing standalone factor in the previous display therefore equals one. Finally, the $j$-indexed factors of the two products in the lemma statement that involve indices $0$, $-1$, and the conventions $Y_{-1} = Z_{-1} = 0$, $W_{-1} = 1$ each equal one as well: for instance, the $j = 1$ term of the second product is $\left\{\E[(1-Z_0)(1-Y_{-1})W_{-1}] - \E[Y_0(1-Z_0)(1-Y_{-1})W_{-1}]\right\}/\E[(1-Y_0)(1-Z_0)W_0] = (1-0)/1 = 1$. Under these conditions, the products consolidate and the expression collapses to the following form:}
\allowdisplaybreaks
\begin{align*}
{h}_{k}^{(Y)}(d)&\left\{1-{h}_{k}^{(Z)}(d)\right\}\prod_{j=1}^{k-1}\left\{1-{h}_{j}^{(Y)}(d)\right\}\left\{1-{h}_{j}^{(Z)}(d)\right\}\\
&= \E\left[Y_{k}(1-Z_{k})(1-Y_{k-1})W^{(d)}_{k-1}\right]\\
&\quad \times \prod_{j=1}^{k} \left[\frac{\E\left[(1-Y_{j-1})(1-Z_{j-1})W^{(d)}_{j-1}\right] - \E\left[Z_{j}(1-Y_{j-1})(1-Z_{j-1})W^{(d)}_{j-1}\right]}{\E\left[(1-Z_{j})(1-Y_{j-1})W^{(d)}_{j-1}\right]}\right]\\
&\quad \times \prod_{j=1}^{k} \left[\frac{\E\left[(1-Z_{j-1})(1-Y_{j-2})W^{(d)}_{j-2}\right] - \E\left[Y_{j-1}(1-Z_{j-1})(1-Y_{j-2})W^{(d)}_{j-2}\right]}{\E\left[(1-Y_{j-1})(1-Z_{j-1})W^{(d)}_{j-1}\right]}\right].
\end{align*}
This matches the form of the lemma.
\end{proof}

The following lemma establishes the relationship between weighted expectations at successive time points, allowing for the reduction of the weights in the expression derived in Lemma \ref{lemma:young_lemma4}.

\begin{lemma}\label{lemma:young_lemma5}
{For any $j \in \{1, \dots, K\}$, the following equalities hold for expectations involving the weights $W^{(d)}_{j-1}$. These are properties of the observed-data distribution and the definition of the true propensities \eqref{eq:pi_C}--\eqref{eq:pi_U} alone, and use no causal assumptions:}
\begin{align*}
    &\E\left[(1-Z_{j})(1-Y_{j-1})W^{(d)}_{j-1}\right]\\ 
    &\quad = \E\left[(1-Y_{j-1})(1-Z_{j-1})W^{(d)}_{j-1}\right] - \E\left[Z_{j}(1-Y_{j-1})(1-Z_{j-1})W^{(d)}_{j-1}\right],
    \shortintertext{and}
    &\E\left[(1-Y_{j-1})(1-Z_{j-1})W^{(d)}_{j-1}\right]\\ 
    &\quad = \E\left[(1-Z_{j-1})(1-Y_{j-2})W^{(d)}_{j-2}\right] - \E\left[Y_{j-1}(1-Z_{j-1})(1-Y_{j-2})W^{(d)}_{j-2}\right].
\end{align*}
\end{lemma}

\begin{proof}[Proof of Lemma \ref{lemma:young_lemma5}]
Following the approach of \cite{young_causal_2020}, we first observe that if an individual has not experienced the competing event by time $j$ ($Z_j=0$) and has not experienced the event of interest by time $j-1$ ($Y_{j-1}=0$), they necessarily must not have experienced the competing event at the previous time ($Z_{j-1}=0$). Thus, $(Z_{j} = 0, Y_{j-1} = 0)$ implies $(1-Z_{j-1})=1$. Using this property to insert $(1-Z_{j-1})$ and then distributing the terms:
\begin{align*}
    &\E\left[(1-Z_{j})(1-Y_{j-1})W^{(d)}_{j-1}\right]\\ 
    &\quad = \E\left[(1-Z_{j})(1-Y_{j-1})(1-Z_{j-1})W^{(d)}_{j-1}\right]\\
    &\quad = \E\left[(1-Y_{j-1})(1-Z_{j-1})W^{(d)}_{j-1} - Z_{j}(1-Y_{j-1})(1-Z_{j-1})W^{(d)}_{j-1}\right]\\
    &\quad = \E\left[(1-Y_{j-1})(1-Z_{j-1})W^{(d)}_{j-1}\right] - \E\left[Z_{j}(1-Y_{j-1})(1-Z_{j-1})W^{(d)}_{j-1}\right].
\end{align*}
{For the second equality, we proceed in two steps: we first show that the weight $W^{(d)}_{j-1}$ can be reduced to $W^{(d)}_{j-2}$ under the expectation, and we then insert the redundant survival indicator and distribute, exactly as above but at the reduced weight. (Performing the steps in this order matters: the insertion argument is purely algebraic and is valid at either weight, whereas the weight-reduction argument below applies to a factor of the form $f(\overline{X}_{j-1}, \overline{Y}_{j-1}, \overline{Z}_{j-1})\,W^{(d)}_{j-1}$, and $(1-Y_{j-1})(1-Z_{j-1})$ is such a factor.)}

{For the weight-reduction step, write}
\begin{equation*}
    W^{(d)}_{j-1} = W^{(d)}_{j-2} \times \frac{\mathbbm{1}(C_{j-1}=1)}{\pi^C_{j-1}} \times \frac{\mathbbm{1}(U_{j-1}=1)}{\pi^U_{j-1}},
\end{equation*}
which follows from the definition \eqref{eq:ipcw}. In the chain of equalities below, the successive steps are justified as follows. The first equality substitutes this decomposition; there, the two conditional probabilities in the denominators are the true propensities $\pi^C_{j-1}$ and $\pi^U_{j-1}$ rewritten via (P2): on the event $\{C_{j-1} = 1, \overline{U}_{j-2} = 1\}$, which is enforced by the indicators in the numerator (and off which the entire integrand vanishes), conditioning on $(\overline{X}_{j-1},\allowbreak \overline{Z}_{j-1},\allowbreak \overline{Y}_{j-1},\allowbreak \overline{A}_{j-1},\allowbreak \overline{U}_{j-2})$ carries the same information as the conditioning sets in \eqref{eq:pi_C}--\eqref{eq:pi_U}, because on that event the treatment history is the deterministic regime path. The second equality is the law of iterated expectations; the third pulls out of the inner conditional expectation every factor that is a function of the conditioning variables (all factors except $U_{j-1}$; note $W^{(d)}_{j-2}$ is such a function by (P4)); the fourth applies (P3) to replace $\E\left[U_{j-1} \mid \cdot\right]$ by the displayed probability on the event enforced by $\mathbbm{1}(C_{j-1}=1)$ and the indicators inside $W^{(d)}_{j-2}$; the fifth cancels the ratio; the sixth rewrites $\mathbbm{1}(C_{j-1} = 1) = \mathbbm{1}(A_{j-1} = d_{j-1})$ on the event $\{C_{j-2} = 1\}$ enforced by $W^{(d)}_{j-2}$; and the remaining equalities repeat the same conditioning argument for the treatment indicator, using the second display of (P3):
\allowdisplaybreaks
\begin{align*}
&\E\left[(1-Y_{j-1})(1-Z_{j-1})W^{(d)}_{j-1}\right]\\
    &= \E\bigg[(1-Y_{j-1})(1-Z_{j-1})W^{(d)}_{j-2} \frac{\scriptstyle{\mathbbm{1}(C_{j-1} = 1)}}{\scriptstyle{\Pr(C_{j-1} = 1 \; | \; \overline{X}_{j-1}, \overline{Z}_{j-1}, \overline{Y}_{j-1}, \overline{C}_{j-2} = 1, \overline{U}_{j-2} = 1)}}\\
&\qquad\quad \times \frac{\scriptstyle{\mathbbm{1}(U_{j-1} = 1)}}{\scriptstyle{\Pr(U_{j-1} = 1 \; | \; \overline{X}_{j-1}, \overline{Z}_{j-1}, \overline{Y}_{j-1}, \overline{A}_{j-1}, \overline{U}_{j-2})}}\bigg]\\
    &= \E\Bigg[\E\bigg[ (1-Y_{j-1})(1-Z_{j-1})W^{(d)}_{j-2} \frac{\scriptstyle{\mathbbm{1}(C_{j-1} = 1)}}{\scriptstyle{\Pr(C_{j-1} = 1 \; | \; \overline{X}_{j-1}, \overline{Z}_{j-1}, \overline{Y}_{j-1}, \overline{C}_{j-2} = 1, \overline{U}_{j-2} = 1)}}\\
&\qquad\qquad \times \frac{\scriptstyle{\mathbbm{1}(U_{j-1} = 1)}}{\scriptstyle{\Pr(U_{j-1} = 1 \; | \; \overline{X}_{j-1}, \overline{Z}_{j-1}, \overline{Y}_{j-1}, \overline{A}_{j-1}, \overline{U}_{j-2})}} \bigg| \; \overline{X}_{j-1}, \overline{Z}_{j-1}, \overline{Y}_{j-1}, \overline{C}_{j-1}, \overline{U}_{j-2} \bigg]\Bigg]\\
    &= \E\Bigg[ (1-Y_{j-1})(1-Z_{j-1})W^{(d)}_{j-2}\; \frac{\scriptstyle{\mathbbm{1}(C_{j-1} = 1)}}{\scriptstyle{\Pr(C_{j-1} = 1 \; | \; \overline{X}_{j-1}, \overline{Z}_{j-1}, \overline{Y}_{j-1}, \overline{C}_{j-2} = 1, \overline{U}_{j-2} = 1)}}\\
&\qquad\quad \times \frac{\scriptstyle{1}}{\scriptstyle{\Pr(U_{j-1} = 1 \; | \; \overline{X}_{j-1}, \overline{Z}_{j-1}, \overline{Y}_{j-1}, \overline{A}_{j-1}, \overline{U}_{j-2})}} \E\left[ U_{j-1} \big| \; \overline{X}_{j-1}, \overline{Z}_{j-1}, \overline{Y}_{j-1}, \overline{C}_{j-1}, \overline{U}_{j-2} \right]\Bigg]\\
    &= \E\Bigg[ (1-Y_{j-1})(1-Z_{j-1})W^{(d)}_{j-2}\; \frac{\scriptstyle{\mathbbm{1}(C_{j-1} = 1)}}{\scriptstyle{\Pr(C_{j-1} = 1 \; | \; \overline{X}_{j-1}, \overline{Z}_{j-1}, \overline{Y}_{j-1}, \overline{C}_{j-2} = 1, \overline{U}_{j-2} = 1)}}\\
&\qquad\quad \times \frac{\scriptstyle{\Pr(U_{j-1} = 1 \; | \; \overline{X}_{j-1}, \overline{Z}_{j-1}, \overline{Y}_{j-1}, \overline{A}_{j-1}, \overline{U}_{j-2})}}{\scriptstyle{\Pr(U_{j-1} = 1 \; | \; \overline{X}_{j-1}, \overline{Z}_{j-1}, \overline{Y}_{j-1}, \overline{A}_{j-1}, \overline{U}_{j-2})}} \Bigg]\\
    &= \E\left[ (1-Y_{j-1})(1-Z_{j-1})W^{(d)}_{j-2}\; \frac{\scriptstyle{\mathbbm{1}(C_{j-1} = 1)}}{\scriptstyle{\Pr(C_{j-1} = 1 \; | \; \overline{X}_{j-1}, \overline{Z}_{j-1}, \overline{Y}_{j-1}, \overline{C}_{j-2} = 1, \overline{U}_{j-2} = 1)}}\right]\\
    &= \E\left[ (1-Y_{j-1})(1-Z_{j-1})W^{(d)}_{j-2}\; \frac{\scriptstyle{\mathbbm{1}(A_{j-1} = d_{j-1})}}{\scriptstyle{\Pr(A_{j-1} = d_{j-1} \; | \; \overline{X}_{j-1}, \overline{Z}_{j-1}, \overline{Y}_{j-1}, \overline{C}_{j-2} = 1, \overline{U}_{j-2} = 1)}}\right]\\
    &= \E\Bigg[\E\bigg[ (1-Y_{j-1})(1-Z_{j-1})W^{(d)}_{j-2}\\
&\qquad\qquad \times \frac{\scriptstyle{\mathbbm{1}(A_{j-1} = d_{j-1})}}{\scriptstyle{\Pr(A_{j-1} = d_{j-1} \; | \; \overline{X}_{j-1}, \overline{Z}_{j-1}, \overline{Y}_{j-1}, \overline{C}_{j-2} = 1, \overline{U}_{j-2} = 1)}} \bigg| \; \overline{X}_{j-1}, \overline{Z}_{j-1}, \overline{Y}_{j-1}, \overline{C}_{j-2} = 1, \overline{U}_{j-2} = 1 \bigg]\Bigg]\\
    &= \E\Bigg[ (1-Y_{j-1})(1-Z_{j-1})W^{(d)}_{j-2}\; \frac{\scriptstyle{1}}{\scriptstyle{\Pr(A_{j-1} = d_{j-1} \; | \; \overline{X}_{j-1}, \overline{Z}_{j-1}, \overline{Y}_{j-1}, \overline{C}_{j-2} = 1, \overline{U}_{j-2} = 1)}}\\
&\qquad\quad \times \E\left[ \mathbbm{1}(A_{j-1} = d_{j-1}) \big| \; \overline{X}_{j-1}, \overline{Z}_{j-1}, \overline{Y}_{j-1}, \overline{C}_{j-2} = 1, \overline{U}_{j-2} = 1 \right]\Bigg]\\
    &= \E\Bigg[ (1-Y_{j-1})(1-Z_{j-1})W^{(d)}_{j-2}\; \frac{\scriptstyle{\Pr(A_{j-1} = d_{j-1} \; | \; \overline{X}_{j-1}, \overline{Z}_{j-1}, \overline{Y}_{j-1}, \overline{C}_{j-2} = 1, \overline{U}_{j-2} = 1)}}{\scriptstyle{\Pr(A_{j-1} = d_{j-1} \; | \; \overline{X}_{j-1}, \overline{Z}_{j-1}, \overline{Y}_{j-1}, \overline{C}_{j-2} = 1, \overline{U}_{j-2} = 1)}} \Bigg]\\
    &= \E\left[(1-Y_{j-1})(1-Z_{j-1})W^{(d)}_{j-2}\right].
\end{align*}
{With the weight reduced, the insertion step proceeds exactly as for the first equality: $(Y_{j-1} = 0, Z_{j-1} = 0)$ implies $Y_{j-2} = 0$ by monotonicity of the outcome process, so the indicator $(1-Y_{j-2})$ may be inserted without changing the integrand, and distributing the product gives}
\begin{align*}
    &\E\left[(1-Y_{j-1})(1-Z_{j-1})W^{(d)}_{j-2}\right]\\
    &\quad = \E\left[(1-Y_{j-1})(1-Z_{j-1})(1-Y_{j-2})W^{(d)}_{j-2}\right]\\
    &\quad = \E\left[(1-Z_{j-1})(1-Y_{j-2})W^{(d)}_{j-2} - Y_{j-1}(1-Z_{j-1})(1-Y_{j-2})W^{(d)}_{j-2}\right]\\
    &\quad = \E\left[(1-Z_{j-1})(1-Y_{j-2})W^{(d)}_{j-2}\right] - \E\left[Y_{j-1}(1-Z_{j-1})(1-Y_{j-2})W^{(d)}_{j-2}\right],
\end{align*}
which, combined with the weight-reduction identity, establishes the second equality.
\end{proof}

\begin{proof}[Proof of Theorem \ref{theorem:g-formula_identification}]
{Fix $k \in \{1, \dots, K\}$. Under Assumptions \ref{sequential-ignorability}, \ref{positivity}, and \ref{consistency}, which justify the weighting identities in Theorem \ref{theorem:weighting_identification}, Lemma \ref{lemma:young_lemma4} expands the incidence at time $k$ as
\begin{align*}
{h}_{k}^{(Y)}(d)&\left\{1-{h}_{k}^{(Z)}(d)\right\}\prod_{j=1}^{k-1}\left\{1-{h}_{j}^{(Y)}(d)\right\}\left\{1-{h}_{j}^{(Z)}(d)\right\}\\
&= \E\left[Y_{k}(1-Z_{k})(1-Y_{k-1})W^{(d)}_{k-1}\right]\\
&\quad \times \prod_{j=1}^{k} \left[\frac{\E\left[(1-Y_{j-1})(1-Z_{j-1})W^{(d)}_{j-1}\right] - \E\left[Z_{j}(1-Y_{j-1})(1-Z_{j-1})W^{(d)}_{j-1}\right]}{\E\left[(1-Z_{j})(1-Y_{j-1})W^{(d)}_{j-1}\right]}\right]\\
&\quad \times \prod_{j=1}^{k} \left[\frac{\E\left[(1-Z_{j-1})(1-Y_{j-2})W^{(d)}_{j-2}\right] - \E\left[Y_{j-1}(1-Z_{j-1})(1-Y_{j-2})W^{(d)}_{j-2}\right]}{\E\left[(1-Y_{j-1})(1-Z_{j-1})W^{(d)}_{j-1}\right]}\right].
\end{align*}
Every bracketed factor equals one, by Lemma \ref{lemma:young_lemma5} applied factor by factor. For the $j$-th factor of the first product, the first equality of Lemma \ref{lemma:young_lemma5} states that the numerator, $\E\left[(1-Y_{j-1})(1-Z_{j-1})W^{(d)}_{j-1}\right] - \E\left[Z_{j}(1-Y_{j-1})(1-Z_{j-1})W^{(d)}_{j-1}\right]$, equals the denominator $\E\left[(1-Z_{j})(1-Y_{j-1})W^{(d)}_{j-1}\right]$. For the $j$-th factor of the second product, the second equality of Lemma \ref{lemma:young_lemma5} states that the numerator, $\E\left[(1-Z_{j-1})(1-Y_{j-2})W^{(d)}_{j-2}\right] - \E\left[Y_{j-1}(1-Z_{j-1})(1-Y_{j-2})W^{(d)}_{j-2}\right]$, equals the denominator $\E\left[(1-Y_{j-1})(1-Z_{j-1})W^{(d)}_{j-1}\right]$. Both denominators are strictly positive by Assumption \ref{positivity} whenever the at-risk population is nonempty, so the cancellations are well defined. Therefore,
\begin{equation*}
{h}_{k}^{(Y)}(d)\left\{1-{h}_{k}^{(Z)}(d)\right\}\prod_{j=1}^{k-1}\left\{1-{h}_{j}^{(Y)}(d)\right\}\left\{1-{h}_{j}^{(Z)}(d)\right\} = \E\left[Y_{k}(1-Z_{k})(1-Y_{k-1})W^{(d)}_{k-1}\right].
\end{equation*}
Then, for the fixed index $k$,}
\begin{align*}
    \E\left[Y_{k}(1-Z_{k})(1-Y_{k-1})W^{(d)}_{k-1}\right] & = \E\left[Y_{k}(1-Z_{k}) \prod_{j=1}^{k-1}\bigg\{(1-Y_{j})(1-Z_{j}) \bigg\}W^{(d)}_{k-1}\right],
\end{align*}
because $Y$ and $Z$ are binary and, for all $j$, $Y_j = 1$ and $Z_j = 1$ imply that $\underline{Y}_{j+1} = 1$ and $\underline{Z}_{j+1} = 1$, respectively. Then, for observed data $\overline{V}_k = \left\{\overline{Y}_k, \overline{Z}_k, \overline{U}_{k-1}, \overline{C}_{k-1}, \overline{X}_{k-1}\right\}$, to express this expectation in terms of the observed data distribution, we expand it into its integral form:
\allowdisplaybreaks
\begin{align*}
    \E&\left[Y_{k}(1-Z_{k})\prod_{j=1}^{k-1}\bigg\{(1-Y_{j})(1-Z_{j}) \bigg\}W^{(d)}_{k-1}\right]\\ 
    &= \int_{\overline{x}_{k-1}}\sum_{\overline{c}_{k-1}}\sum_{\overline{u}_{k-1}}\sum_{\overline{z}_k}\sum_{\overline{y}_k} y_{k}(1-z_{k}) \prod_{j=1}^{k-1}\bigg\{(1-y_{j})(1-z_{j}) \bigg\}\\
    &\qquad \times \prod_{j=0}^{k-1} \frac{\scriptstyle{u_{j}}}{\scriptstyle{\Pr\left(U_{j} = 1 \; | \; \overline{Z}_{j}, \overline{Y}_{j}, \overline{U}_{j-1} = 1, \overline{C}_{j} = 1, \overline{X}_{j}\right)}} \;\frac{\scriptstyle{c_{j}}}{\scriptstyle{\Pr\left(C_{j} = 1 \; | \; \overline{Z}_{j}, \overline{Y}_{j}, \overline{U}_{j-1} = 1, \overline{C}_{j-1} = 1, \overline{X}_{j}\right)}}\\
    &\qquad \times \prod_{j=1}^{k} \bigg\{ p\left( y_j \mid \overline{z}_j, \overline{y}_{j-1}, \overline{x}_{j-1}, \overline{c}_{j-1}, \overline{u}_{j-1} \right)  p\left( z_j \mid \overline{z}_{j-1}, \overline{y}_{j-1}, \overline{x}_{j-1}, \overline{c}_{j-1}, \overline{u}_{j-1} \right)\\
    &\qquad\qquad\qquad \times p\left( u_{j-1} \mid \overline{z}_{j-1}, \overline{y}_{j-1}, \overline{x}_{j-1}, \overline{c}_{j-1}, \overline{u}_{j-2} \right)  p\left( c_{j-1} \mid \overline{z}_{j-1}, \overline{y}_{j-1}, \overline{x}_{j-1}, \overline{c}_{j-2}, \overline{u}_{j-2} \right)\\
    &\qquad\qquad\qquad \times \text{d}P \left( x_{j-1} \mid \overline{z}_{j-1}, \overline{y}_{j-1}, \overline{x}_{j-2}, \overline{c}_{j-2}, \overline{u}_{j-2} \right) \bigg\},
    \shortintertext{which, summing over binary $\overline{U}_{k-1}$ and $\overline{C}_{k-1}$,}
    &= \int_{\overline{x}_{k-1}}\sum_{\overline{z}_k}\sum_{\overline{y}_k} y_{k}(1-z_{k}) \prod_{j=1}^{k-1}\bigg\{(1-y_{j})(1-z_{j}) \bigg\}\\
    &\qquad \times \prod_{j=0}^{k-1}\frac{1}{\scriptstyle{\Pr\left(U_{j} = 1 \; | \; \overline{Z}_{j}, \overline{Y}_{j}, \overline{U}_{j-1} = 1, \overline{C}_{j} = 1, \overline{X}_{j}\right)}} \;\frac{1}{\scriptstyle{\Pr\left(C_{j} = 1 \; | \; \overline{Z}_{j}, \overline{Y}_{j}, \overline{U}_{j-1} = 1, \overline{C}_{j-1} = 1, \overline{X}_{j}\right)}}\\
    &\qquad \times \prod_{j=1}^{k} \bigg\{ p\left( y_j \mid \overline{z}_j, \overline{y}_{j-1}, \overline{x}_{j-1}, \overline{C}_{j-1} = 1, \overline{U}_{j-1} = 1 \right) \\
    &\qquad\qquad\qquad \times p\left( z_j \mid \overline{z}_{j-1}, \overline{y}_{j-1}, \overline{x}_{j-1}, \overline{C}_{j-1} = 1, \overline{U}_{j-1} = 1 \right)\\
    &\qquad\qquad\qquad \times \Pr\left(U_{j-1} = 1 \; | \; \overline{Z}_{j-1}, \overline{Y}_{j-1}, \overline{U}_{j-2} = 1, \overline{C}_{j-1} = 1, \overline{X}_{j-1}\right) \\
     &\qquad\qquad\qquad \times \Pr\left(C_{j-1} = 1 \; | \; \overline{Z}_{j-1}, \overline{Y}_{j-1}, \overline{U}_{j-2} = 1, \overline{C}_{j-2} = 1, \overline{X}_{j-1}\right)\\
    &\qquad\qquad\qquad \times \text{d}P \left( x_{j-1} \mid \overline{z}_{j-1}, \overline{y}_{j-1}, \overline{x}_{j-2}, \overline{C}_{j-2} = 1, \overline{U}_{j-2} = 1 \right) \bigg\}\\
    &= \int_{\overline{x}_{k-1}}\sum_{\overline{z}_k}\sum_{\overline{y}_k} y_{k}(1-z_{k}) \prod_{j=1}^{k-1}\bigg\{(1-y_{j})(1-z_{j}) \bigg\}\\
    &\qquad \times \prod_{j=1}^{k} \bigg\{ p\left( y_j \mid \overline{z}_j, \overline{y}_{j-1}, \overline{x}_{j-1}, \overline{C}_{j-1} = 1, \overline{U}_{j-1} = 1 \right) \\
    &\qquad\qquad\qquad \times p\left( z_j \mid \overline{z}_{j-1}, \overline{y}_{j-1}, \overline{x}_{j-1}, \overline{C}_{j-1} = 1, \overline{U}_{j-1} = 1 \right)\\
    &\qquad\qquad\qquad \times \text{d}P \left( x_{j-1} \mid \overline{z}_{j-1}, \overline{y}_{j-1}, \overline{x}_{j-2}, \overline{C}_{j-2} = 1, \overline{U}_{j-2} = 1 \right) \bigg\}.
    \shortintertext{Summing over binary $\overline{Y}_{k}$ and $\overline{Z}_{k}$,}
    &= \int_{\overline{x}_{k-1}} \Pr\left(Y_{k} = 1 \:|\: \overline{Z}_k = 0, \overline{Y}_{k-1} = 0, \overline{U}_{k-1} = 1, \overline{C}_{k-1} = 1, \overline{X}_{k-1} = \overline{x}_{k-1}\right)\\
    &\qquad \times \prod_{j=1}^{k} \bigg
    \{\Pr\left(Z_j = 0 \; | \; \overline{Y}_{j-1} = 0, \overline{Z}_{j-1} = 0, \overline{U}_{j-1} = 1, \overline{C}_{j-1} = 1, \overline{X}_{j-1} = \overline{x}_{j-1} \right)\\
    &\qquad\qquad\qquad \times \Pr\left({Y}_{j-1} = 0 \; | \; \overline{Z}_{j-1} = 0, \overline{Y}_{j-2} = 0, \overline{U}_{j-2} = 1, \overline{C}_{j-2} = 1, \overline{X}_{j-2} = \overline{x}_{j-2} \right) \\
    &\qquad\qquad\qquad \times \text{d}P \left( x_{j-1} \mid \overline{Z}_{j-1} = 0, \overline{Y}_{j-1} = 0, \overline{X}_{j-2} = \overline{x}_{j-2}, \overline{C}_{j-2} = 1, \overline{U}_{j-2} = 1 \right) \bigg\},
    \shortintertext{which, by definition, is:}
    &= \int_{\overline{x}_{k-1}} \mu_{Y_k}(\overline{x}_{k-1})\; \nu_{1 - Z_k}(\overline{x}_{k-1}) \prod_{j=1}^{k-1}\bigg\{\mu_{1 - Y_j}\left(\overline{x}_{j-1}\right)\; \nu_{1 - Z_j}\left(\overline{x}_{j-1}\right) \bigg\} \\\tag*{\textcircled{1}}
    &\qquad \times \prod_{j=1}^{k} \text{d}P \left( x_{j-1} \mid \overline{Z}_{j-1} = 0, \overline{Y}_{j-1} = 0, \overline{X}_{j-2} = \overline{x}_{j-2}, \overline{C}_{j-2} = 1, \overline{U}_{j-2} = 1 \right).
\end{align*}
Isolating the terms in \textcircled{1} that are dependent on $X_{k-1}$,
\allowdisplaybreaks
\begin{align*}
    \int_{{x}_{k-1}}& \mu_{Y_k}(\overline{x}_{k-1})\; \nu_{1 - Z_k}(\overline{x}_{k-1})\; \text{d}P \left( x_{k-1} \mid \overline{Z}_{k-1} = 0, \overline{Y}_{k-1} = 0, \overline{X}_{k-2} = \overline{x}_{k-2}, \overline{C}_{k-2} = 1, \overline{U}_{k-2} = 1 \right),\\
    \shortintertext{by definition results in:}
    &= \E \left[ \mu_{Y_k}(\overline{X}_{k-1})\; \nu_{1 - Z_k}(\overline{X}_{k-1})\; \big| \; \overline{Z}_{k-1} = 0, \overline{Y}_{k-1} = 0, \overline{C}_{k-2} = 1, \overline{U}_{k-2} = 1, \overline{X}_{k-2} \right]\\ \tag*{\textcircled{2}}
    &= \E \left[ \mu_{Y_k}(\overline{X}_{k-1})\; \nu_{1 - Z_k}(\overline{X}_{k-1})\; \big| \; \overline{\Ocal}_{k-1} \right].
\end{align*}
Substituting {\textcircled{2}} into {\textcircled{1}}, 
\begin{align*}
    \text{\textcircled{1}} &= \int_{\overline{x}_{k-2}} \E \left[ \mu_{Y_k}(\overline{X}_{k-1})\; \nu_{1 - Z_k}(\overline{X}_{k-1})\; \big| \; \overline{\Ocal}_{k-1} \right] \prod_{j=1}^{k-1}\bigg\{\mu_{1 - Y_j}\left(\overline{x}_{j-1}\right)\; \nu_{1 - Z_j}\left(\overline{x}_{j-1}\right) \bigg\}\\
    &\qquad \times \prod_{j=1}^{k-1} \bigg\{\text{d}P \left( x_{j-1} \mid \overline{Z}_{j-1} = 0, \overline{Y}_{j-1} = 0, \overline{X}_{j-2} = \overline{x}_{j-2}, \overline{C}_{j-2} = 1, \overline{U}_{j-2} = 1 \right) \bigg\}.
\end{align*}
Then, isolating the terms dependent on ${X_{k-2}}$ in \textcircled{1},
\begin{align*}
    \int_{{x}_{k-2}}& \E \left[ \mu_{Y_k}(\overline{X}_{k-1})\; \nu_{1 - Z_k}(\overline{X}_{k-1})\; \big| \; \overline{\Ocal}_{k-1} \right] \mu_{1 - Y_{k-1}}(\overline{x}_{k-2})\; \nu_{1 - Z_{k-1}}(\overline{x}_{k-2})\\
    &\quad \times \text{d}P \left( x_{k-2} \mid \overline{Z}_{k-2} = 0, \overline{Y}_{k-2} = 0, \overline{X}_{k-3} = \overline{x}_{k-3}, \overline{C}_{k-3} = 1, \overline{U}_{k-3} = 1 \right)\\
    \shortintertext{by definition results in:}
    &= \E \left[ \mu_{1 - Y_{k-1}}\left(\overline{X}_{k-2}\right)\; \nu_{1 - Z_{k-1}} \left(\overline{X}_{k-2}\right) \; \E \left[ \mu_{Y_k}(\overline{X}_{k-1})\; \nu_{1 - Z_k}(\overline{X}_{k-1}) \; \big| \; \overline{\Ocal}_{k-1} \right] \; \bigg| \; \overline{\Ocal}_{k-2} \right]\\ \tag*{\textcircled{3}}
    &= \E \left[ \E \left[ \mu_{Y_k}(\overline{X}_{k-1})\; \nu_{1 - Z_k}(\overline{X}_{k-1})\; \mu_{1 - Y_{k-1}}(\overline{X}_{k-2})\; \nu_{1 - Z_{k-1}} (\overline{X}_{k-2})\; \big| \; \overline{\Ocal}_{k-1} \right]\; \bigg| \; \overline{\Ocal}_{k-2} \right].
\end{align*}
Substituting {\textcircled{3}},
\begin{align*}
    \text{\textcircled{1}} &= \int_{\overline{x}_{k-3}} \E \left[ \E \left[ \mu_{Y_k}(\overline{X}_{k-1})\; \nu_{1 - Z_k}(\overline{X}_{k-1})\; \mu_{1 - Y_{k-1}}(\overline{X}_{k-2})\; \nu_{1 - Z_{k-1}} (\overline{X}_{k-2})\; \big| \; \overline{\Ocal}_{k-1} \right]\; \bigg| \; \overline{\Ocal}_{k-2} \right]\\
    &\qquad \times \prod_{j=1}^{k-2} \bigg\{\mu_{1 - Y_j}(\overline{x}_{j-1})\; \nu_{1 - Z_j}(\overline{x}_{j-1}) \bigg\}\\
     &\qquad \times \prod_{j=1}^{k-2} \bigg\{ \text{d}P \left( x_{j-1} \mid \overline{Z}_{j-1} = 0, \overline{Y}_{j-1} = 0, \overline{X}_{j-2} = \overline{x}_{j-2}, \overline{C}_{j-2} = 1, \overline{U}_{j-2} = 1 \right) \bigg\}.
\end{align*}
Repeating this recursively for ${X_{k-3},...,X_{0}}$, \textcircled{1} is equal to:
\begin{align*}
   &\int_{{x}_{0}} \E \Bigg[ \E \Bigg[ \cdots \E \Bigg[ \mu_{Y_k}(\overline{X}_{k-1}) \nu_{1 - Z_k}(\overline{X}_{k-1})\\
   &\qquad\qquad\qquad\quad \times \prod_{j=1}^{k-1} \mu_{1 - Y_{j}}(\overline{X}_{j-1}) \nu_{1 - Z_{j}}(\overline{X}_{j-1}) \big|  \overline{\Ocal}_{k-1} \Bigg] \cdots \big| \overline{\Ocal}_{2}\Bigg] \Big| \overline{\Ocal}_{1} \Bigg] \text{d}P\left({x}_{0}\right)\\
    &= \E \Bigg[ \E \Bigg[ \E \Bigg[ \cdots \E \Bigg[ \mu_{Y_k}(\overline{X}_{k-1}) \nu_{1 - Z_k}(\overline{X}_{k-1}) \\
    &\qquad\qquad\qquad\qquad \times \prod_{j=1}^{k-1} \mu_{1 - Y_{j}}(\overline{X}_{j-1}) \nu_{1 - Z_{j}}(\overline{X}_{j-1}) \big|  \overline{\Ocal}_{k-1} \Bigg] \cdots \big| \overline{\Ocal}_{2}\Bigg] \Big| \overline{\Ocal}_{1} \Bigg]\Bigg]\\
    &= \E \left[ G_1 ^{(k)}\right].
\end{align*}
This constitutes the g-formula representation of the cumulative incidence.
\end{proof}

\section{Efficient Influence Function}\label{sec:supp-EIF}

Let $O_k$ represent all observed data at time $k$, i.e., {$O_k \equiv \{\overline{Y}_k, \overline{Z}_k, \overline{U}_{k-1}, \overline{A}_{k-1}, \overline{X}_{k-1}\}$}. {With mild abuse of notation, in this and the following section $O_k$ denotes this cumulative collection, whereas in the main text $O_j = \{Z_j, Y_j, X_j, A_j, U_j\}$ denotes the single-time-point block; which of the two is meant is always clear from the context.} 

Let $p_\varepsilon(o) = p(o; \varepsilon)$ represent a parametric submodel parametrized by $\varepsilon \in \R$. Furthermore, let $\ell(a\:|\:b;\varepsilon) = \text{log}\: p(a\:|\:b; \varepsilon)$ for any partition $(A, B) \subseteq O$, with score functions represented as $S_\varepsilon(o) = \ell_\varepsilon'(o;0)$. {Because the joint density of $O_k$ factorizes along the temporal ordering into the conditional factors of $Y_j$ and $Z_j$, $j = 1, \dots, k$, and of $X_{j-1}$, $A_{j-1}$, and $U_{j-1}$, $j = 1, \dots, k$, taking logarithms and differentiating at $\varepsilon = 0$ decomposes the score of $O_k$ as}
\begin{align*}
    \ell'_\varepsilon(O_k; 0) = \sum_{j=1}^{k} \Big\{ &\ell'_{\epsilon}\left(Y_{j} \mid \overline{Z}_{j}, \overline{Y}_{j-1}, \overline{U}_{j-1}, \overline{A}_{j-1}, \overline{X}_{j-1}; 0\right) \\
    &+\ell'_{\epsilon}\left(Z_{j} \mid \overline{Y}_{j-1}, \overline{Z}_{j-1}, \overline{U}_{j-1}, \overline{A}_{j-1}, \overline{X}_{j-1}; 0\right) \\
    &+\ell'_{\epsilon}\left(U_{j-1} \mid \overline{Y}_{j-1}, \overline{Z}_{j-1}, \overline{U}_{j-2}, \overline{A}_{j-1}, \overline{X}_{j-1}; 0\right) \\
    &+\ell'_{\epsilon}\left(A_{j-1} \mid \overline{Y}_{j-1}, \overline{Z}_{j-1}, \overline{U}_{j-2}, \overline{A}_{j-2}, \overline{X}_{j-1}; 0\right) \\
    &+\ell'_{\epsilon}\left(X_{j-1} \mid \overline{Y}_{j-1}, \overline{Z}_{j-1}, \overline{U}_{j-2}, \overline{A}_{j-2}, \overline{X}_{j-2}; 0\right)\Big\},
\end{align*}
with the conventions of Section \ref{sec:supp-prelims} for indices below zero; this decomposition is applied to $\ell'_\varepsilon(O_k;0)$ in the first display of the proof of Theorem \ref{theorem:EIF}. Scores of factors realized after $(Z_k, Y_k)$ pair to zero with any function of $O_k$ by the score-centering property (P1), so no generality is lost by restricting to the score of $O_k$.

In order to prove Theorem \ref{theorem:EIF}, we show that $\left.\frac{\partial}{\partial \varepsilon} \psi_k(P_\varepsilon)\right|_{\varepsilon=0} = \E\left[\varphi_k \ell_{\varepsilon}'(O_k;0)\right]$. We begin by introducing several lemmas. For simplicity, we assume some arbitrary regime and drop the superscript $(d)$ notation for the weights and compatibility indicators (i.e., $W_{t}$ and $C_{t}$ for each time $t$).

\begin{lemma}\label{lemma:wt_reduction}
    For $j<0$, let $W_j = 1$. Let $\ell'_\varepsilon(Y_0 \mid Z_0;0) = \ell'_\varepsilon(Z_0;0) = 0$. {For $t = 0, \dots, k-1$, sum the outcome, competing-event, and covariate scores through time $t$ as}
    \begin{align*}
        \Scal_t \equiv \ell'_{\varepsilon}\left(X_0; 0\right) + \sum_{j=1}^{t} \Big\{ &\ell'_{\epsilon}\left(Y_{j} \mid \overline{Z}_{j}, \overline{Y}_{j-1}, \overline{U}_{j-1}, \overline{A}_{j-1}, \overline{X}_{j-1}; 0\right) \\
    &+\ell'_{\epsilon}\left(Z_{j} \mid \overline{Y}_{j-1}, \overline{Z}_{j-1}, \overline{U}_{j-1}, \overline{A}_{j-1}, \overline{X}_{j-1}; 0\right) \\
    &+\ell'_{\epsilon}\left(X_{j} \mid \overline{Y}_{j}, \overline{Z}_{j}, \overline{U}_{j-1}, \overline{A}_{j-1}, \overline{X}_{j-1}; 0\right)\Big\}.
    \end{align*}
    Then, by using the definition of compatibility and Assumption \ref{sequential-ignorability}, for $t=0,\dots,k-1$,
    \begin{equation*}
        \E \left[W_t \left\{\prod_{j=0}^{t}(1-Y_j)(1-Z_j)\right\}G_{t+1}^{(k)}\, \Scal_t\right]
        = \E \left[W_{t-1} \left\{\prod_{j=0}^{t}(1-Y_j)(1-Z_j)\right\}G_{t+1}^{(k)}\, \Scal_t\right].
    \end{equation*}
    In particular, for $t = 0$ we have $\Scal_0 = \ell'_\varepsilon(X_0;0)$, and the identity reads $\E\left[W_0\, G_1^{(k)}\, \ell'_\varepsilon(X_0;0)\right] = \E\left[G_1^{(k)}\, \ell'_\varepsilon(X_0;0)\right]$.
\end{lemma}

\begin{proof}[Proof of Lemma \ref{lemma:wt_reduction}]
{Throughout, note that $\Scal_t$ is a function of $\left(\overline{X}_t, \overline{Y}_t, \overline{Z}_t, \overline{A}_{t-1}, \overline{U}_{t-1}\right)$: each of its scores is a function of the variable indicated and of that score's conditioning variables, all realized no later than $X_t$. In particular, $\Scal_t$ involves neither $A_t$ nor $U_t$, which is what permits it to be carried through each conditioning step below. The case $t = 0$ is included, with $\Scal_0 = \ell'_\varepsilon(X_0;0)$, $W_{-1} = 1$, and $\prod_{j=0}^{0}(1-Y_j)(1-Z_j) = 1$.}
We begin by substituting the definition of the weight term $W_t$. The first three equalities follow by definition and the properties of binary variables:
\allowdisplaybreaks
\begin{align*}
    &\E \Biggl[W_t \left\{\prod_{j=0}^{t}(1-Y_j)(1-Z_j)\right\}G_{t+1}^{(k)} {\Scal_t}\Biggr]\\
        &= \E \Biggl[W_{t-1} \left\{\prod_{j=0}^{t}(1-Y_j)(1-Z_j)\right\}G_{t+1}^{(k)} \times U_t C_t\\
    &\qquad\quad\times \left\{\Pr(U_t=1\mid \overline{Y}_t=0,\overline{Z}_t=0,\overline{U}_{t-1}=1,\overline{C}_t=1,\overline{X}_t)\right\}^{-1}
    \\
    &\qquad\quad\times \left\{\Pr(C_t = 1 \mid \overline{Y}_t=0,\overline{Z}_t=0,\overline{U}_{t-1}=1,\overline{C}_{t-1}=1,\overline{X}_t)\right\}^{-1}
    {\Scal_t}\Biggr]\\
        &= \E \Biggl[W_{t-1} \left\{\prod_{j=0}^{t}(1-Y_j)(1-Z_j)\right\}G_{t+1}^{(k)} \times U_t C_t\\
    &\qquad\quad\times \left\{\Pr(U_t=1\mid \overline{Y}_t,\overline{Z}_t,\overline{U}_{t-1},\overline{C}_t,\overline{X}_t)\right\}^{-1}\\
    &\qquad\quad\times \left\{\Pr(C_t = 1 \mid \overline{Y}_t,\overline{Z}_t,\overline{U}_{t-1},\overline{C}_{t-1},\overline{X}_t)\right\}^{-1}{\Scal_t}\Biggr]\\
        &= \E \Biggl[W_{t-1} \left\{\prod_{j=0}^{t}(1-Y_j)(1-Z_j)\right\}G_{t+1}^{(k)} \times U_t C_t\\
    &\qquad\quad\times \left\{\Pr(U_t=1\mid \overline{Y}_t,\overline{Z}_t,\overline{U}_{t-1},\overline{A}_t,\overline{X}_t)\right\}^{-1}\\
    &\qquad\quad\times \left\{\Pr(C_t = 1 \mid \overline{Y}_t,\overline{Z}_t,\overline{U}_{t-1},\overline{C}_{t-1},\overline{X}_t)\right\}^{-1} {\Scal_t}\Biggr].\\
\shortintertext{Conditioning on $\overline{Y}_t,\overline{Z}_t,\overline{U}_{t-1},\overline{A}_t, \text{and } \overline{X}_t$:}
        &= \E \Biggl\{ \E \Biggl[
        W_{t-1} \left\{\prod_{j=0}^{t}(1-Y_j)(1-Z_j)\right\}G_{t+1}^{(k)} \times U_t C_t\\
    &\qquad\qquad\times \left\{\Pr(U_t=1\mid \overline{Y}_t,\overline{Z}_t,\overline{U}_{t-1},\overline{A}_t,\overline{X}_t)\right\}^{-1}\\
    &\qquad\qquad\times \left\{\Pr(C_t = 1 \mid \overline{Y}_t,\overline{Z}_t,\overline{U}_{t-1},\overline{C}_{t-1},\overline{X}_t)\right\}^{-1}\\
    &\qquad\qquad\times {\Scal_t} \bigg|\; \overline{Y}_t,\overline{Z}_t,\overline{U}_{t-1},\overline{A}_t,\overline{X}_t 
    \Biggr]\Biggr\}\\
        &= \E \Biggl\{ \E \Biggl[ U_t \bigg|\; \overline{Y}_t,\overline{Z}_t,\overline{U}_{t-1},\overline{A}_t,\overline{X}_t 
    \Biggr]\\
    &\qquad\qquad\times W_{t-1} \left\{\prod_{j=0}^{t}(1-Y_j)(1-Z_j)\right\}G_{t+1}^{(k)} \times C_t\\
    &\qquad\qquad\times \left\{\Pr(U_t=1\mid \overline{Y}_t,\overline{Z}_t,\overline{U}_{t-1},\overline{A}_t,\overline{X}_t)\right\}^{-1}\\
    &\qquad\qquad\times \left\{\Pr(C_t = 1 \mid \overline{Y}_t,\overline{Z}_t,\overline{U}_{t-1},\overline{C}_{t-1},\overline{X}_t)\right\}^{-1} {\Scal_t}
    \Biggr\}.\\
\shortintertext{Because $U_t$ is binary:}
            &= \E \Biggl\{ \Pr(U_t=1\mid \overline{Y}_t,\overline{Z}_t,\overline{U}_{t-1},\overline{A}_t,\overline{X}_t) \\
    &\qquad\qquad\times W_{t-1} \left\{\prod_{j=0}^{t}(1-Y_j)(1-Z_j)\right\}G_{t+1}^{(k)} \times C_t\\
    &\qquad\qquad\times \left\{\Pr(U_t=1\mid \overline{Y}_t,\overline{Z}_t,\overline{U}_{t-1},\overline{A}_t,\overline{X}_t)\right\}^{-1}\\
    &\qquad\qquad\times \left\{\Pr(C_t = 1 \mid \overline{Y}_t,\overline{Z}_t,\overline{U}_{t-1},\overline{C}_{t-1},\overline{X}_t)\right\}^{-1} {\Scal_t}
    \Biggr\}\\
            &= \E \Biggl\{ W_{t-1} \left\{\prod_{j=0}^{t}(1-Y_j)(1-Z_j)\right\}G_{t+1}^{(k)} \times C_t\\
    &\qquad\qquad\times \left\{\Pr(C_t = 1 \mid \overline{Y}_t,\overline{Z}_t,\overline{U}_{t-1},\overline{C}_{t-1},\overline{X}_t)\right\}^{-1} {\Scal_t}
    \Biggr\}.\\
\shortintertext{By definition of compatibility:}
            &= \E \Biggl\{ W_{t-1} \left\{\prod_{j=0}^{t}(1-Y_j)(1-Z_j)\right\}G_{t+1}^{(k)} \times \mathbbm{1}(A_t = d_t)\\
    &\qquad\qquad\times \left\{\Pr(A_t = d_t \mid \overline{Y}_t,\overline{Z}_t,\overline{U}_{t-1},\overline{A}_{t-1},\overline{X}_t)\right\}^{-1}{\Scal_t}
    \Biggr\}.\\
\shortintertext{Conditioning on $\overline{Y}_t,\overline{Z}_t,\overline{U}_{t-1},\overline{A}_{t-1}, \text{and } \overline{X}_t$:}
            &= \E \Biggl\{ \E \Biggl[
            W_{t-1} \left\{\prod_{j=0}^{t}(1-Y_j)(1-Z_j)\right\}G_{t+1}^{(k)} \times \mathbbm{1}(A_t = d_t)\\
    &\qquad\qquad\times \left\{\Pr(A_t = d_t \mid \overline{Y}_t,\overline{Z}_t,\overline{U}_{t-1},\overline{A}_{t-1},\overline{X}_t)\right\}^{-1}\\
    &\qquad\qquad\times {\Scal_t}
    \bigg|\; \overline{Y}_t,\overline{Z}_t,\overline{U}_{t-1},\overline{A}_{t-1},\overline{X}_t \Biggr]
    \Biggr\}\\
            &= \E \Biggl\{ 
            \E \Biggl[\mathbbm{1}(A_t = d_t) \bigg|\; \overline{Y}_t,\overline{Z}_t,\overline{U}_{t-1},\overline{A}_{t-1},\overline{X}_t \Biggr]\\
        &\qquad\qquad\times W_{t-1} \left\{\prod_{j=0}^{t}(1-Y_j)(1-Z_j)\right\}G_{t+1}^{(k)}\\
    &\qquad\qquad\times \left\{\Pr(A_t = d_t \mid \overline{Y}_t,\overline{Z}_t,\overline{U}_{t-1},\overline{A}_{t-1},\overline{X}_t)\right\}^{-1}  {\Scal_t}
    \Biggr\}.\\
\shortintertext{Because $\mathbbm{1}(A_t = d_t)$ is binary:}
            &= \E \Biggl\{ 
            \Pr(A_t = d_t \mid \overline{Y}_t,\overline{Z}_t,\overline{U}_{t-1},\overline{A}_{t-1},\overline{X}_t)\\
    &\qquad\qquad\times W_{t-1} \left\{\prod_{j=0}^{t}(1-Y_j)(1-Z_j)\right\}G_{t+1}^{(k)}\\
    &\qquad\qquad\times \left\{\Pr(A_t = d_t \mid \overline{Y}_t,\overline{Z}_t,\overline{U}_{t-1},\overline{A}_{t-1},\overline{X}_t)\right\}^{-1}  {\Scal_t}
    \Biggr\}\\
            &= \E \Biggl[ W_{t-1} \left\{\prod_{j=0}^{t}(1-Y_j)(1-Z_j)\right\}G_{t+1}^{(k)} {\Scal_t}
    \Biggr].
\end{align*}
This matches the form of the lemma exactly.
\end{proof}

\begin{lemma}\label{lemma:wt_collapse}
Let $g$ be any function of $O_k = \left(\overline{Y}_k, \overline{Z}_k, \overline{X}_{k-1}, \overline{A}_{k-1}, \overline{U}_{k-1}\right)$ for which the expectations below are finite. Under Assumption \ref{positivity}, and with $\overline{d}_{j}(\overline{x}_{j})$ denoting the deterministic treatment path assigned by the regime along a history with all outcome indicators zero (Section \ref{sec:supp-prelims}),
\allowdisplaybreaks
\begin{align*}
    &\E\left[W_{k-1} \left\{\prod_{j=0}^{k-1}(1-Y_j)(1-Z_j)\right\} g\left(\overline{Y}_k, \overline{Z}_k, \overline{X}_{k-1}, \overline{A}_{k-1}, \overline{U}_{k-1}\right)\right]\\
    &\qquad = \int_{\overline{x}_{k-1}} \sum_{{y}_{k}} \sum_{{z}_{k}}\; g\left((0, \dots, 0, y_k), (0, \dots, 0, z_k), \overline{x}_{k-1}, \overline{d}_{k-1}(\overline{x}_{k-1}), \overline{1}\right)\\
    &\qquad\qquad \times p\left( y_k \mid \overline{z}_k = (\overline{0}, z_k), \overline{y}_{k-1} = 0, \overline{u}_{k-1} = 1, \overline{a}_{k-1} = \overline{d}_{k-1}, \overline{x}_{k-1}\right)\\
    &\qquad\qquad \times p\left( z_k \mid \overline{y}_{k-1} = 0, \overline{z}_{k-1} = 0, \overline{u}_{k-1} = 1, \overline{a}_{k-1} = \overline{d}_{k-1}, \overline{x}_{k-1}\right)\\
    &\qquad\qquad \times \prod_{j=1}^{k-1} \Big\{ p\left( y_j = 0 \mid \overline{z}_j = 0, \overline{y}_{j-1} = 0, \overline{u}_{j-1} = 1, \overline{a}_{j-1} = \overline{d}_{j-1}, \overline{x}_{j-1}\right)\\
    &\qquad\qquad\qquad\qquad \times p\left( z_j = 0 \mid \overline{y}_{j-1} = 0, \overline{z}_{j-1} = 0, \overline{u}_{j-1} = 1, \overline{a}_{j-1} = \overline{d}_{j-1}, \overline{x}_{j-1}\right) \Big\}\\
    &\qquad\qquad \times \prod_{j=1}^{k} \text{d}P\left( x_{j-1} \mid \overline{y}_{j-1} = 0, \overline{z}_{j-1} = 0, \overline{u}_{j-2} = 1, \overline{a}_{j-2} = \overline{d}_{j-2}, \overline{x}_{j-2}\right).
\end{align*}
That is, taking the expectation of $g$ against the weighted, at-risk observed-data measure is the same as integrating $g$, evaluated along the uncensored regime-compatible sequence, against the product of the corresponding conditional laws with the treatment sequence fixed to $\overline{d}_{k-1}$ and the censoring sequence fixed to $\overline{1}$. 
\end{lemma}

\begin{proof}[Proof of Lemma \ref{lemma:wt_collapse}]
Write the left-hand side as an integral against the joint law of $O_k$, factorized along the temporal ordering $O_j = \{Z_j, Y_j, X_j, A_j, U_j\}$:
\begin{align*}
    &\E\left[W_{k-1} \Acal_k\, g\right]
    = \int_{\overline{x}_{k-1}} \sum_{\overline{a}_{k-1}} \sum_{\overline{u}_{k-1}} \sum_{\overline{z}_{k}} \sum_{\overline{y}_{k}}
    \left\{\prod_{j=1}^{k-1}(1-y_j)(1-z_j)\right\} g\left(\overline{y}_k, \overline{z}_k, \overline{x}_{k-1}, \overline{a}_{k-1}, \overline{u}_{k-1}\right)\\
    &\qquad\qquad\qquad\quad \times \prod_{j=0}^{k-1} \frac{u_j\, \mathbbm{1}(\overline{a}_j = \overline{d}_j)}{\pi^U_j\, \pi^{C}_j}\\
    &\qquad\qquad\qquad\quad \times \prod_{j=1}^{k} \Big\{ p\left( y_j \mid \overline{z}_j, \overline{y}_{j-1}, \overline{u}_{j-1}, \overline{a}_{j-1}, \overline{x}_{j-1}\right) p\left( z_j \mid \overline{y}_{j-1}, \overline{z}_{j-1}, \overline{u}_{j-1}, \overline{a}_{j-1}, \overline{x}_{j-1}\right)\Big\}\\
    &\qquad\qquad\qquad\quad \times \prod_{j=0}^{k-1} \Big\{ p\left( a_j \mid \overline{y}_j, \overline{z}_j, \overline{u}_{j-1}, \overline{a}_{j-1}, \overline{x}_j\right) p\left( u_j \mid \overline{y}_j, \overline{z}_j, \overline{u}_{j-1}, \overline{a}_{j}, \overline{x}_j\right) \Big\}\\
    &\qquad\qquad\qquad\quad \times \prod_{j=1}^{k} \text{d}P\left( x_{j-1} \mid \overline{y}_{j-1}, \overline{z}_{j-1}, \overline{u}_{j-2}, \overline{a}_{j-2}, \overline{x}_{j-2}\right),
\end{align*}
where we have substituted the definition \eqref{eq:ipcw} of $W_{k-1}$ (with the indicator $C_j$ written as $\mathbbm{1}(\overline{a}_j = \overline{d}_j)$) and expanded $\Acal_k$; here $\overline{d}_j = \overline{d}_j(\overline{\Hcal}_j, \overline{a}_{j-1})$ is determined by the integration variables. The remaining steps aim to get rid of the sums over $\overline{u}_{k-1}$, $\overline{a}_{k-1}$, and $(\overline{y}_{k-1}, \overline{z}_{k-1})$.

First, the factor $\prod_{j=0}^{k-1} u_j$ is zero unless $\overline{u}_{k-1} = \overline{1}$, so the sum over $\overline{u}_{k-1}$ retains the single term $\overline{u}_{k-1} = \overline{1}$, in which each censoring density factor evaluates to $p\left(u_j = 1 \mid \overline{y}_j, \overline{z}_j, \overline{u}_{j-1} = 1, \overline{a}_j, \overline{x}_j\right)$. On the possible values remaining, we have $\overline{a}_j = \overline{d}_j$ and, for $j \leq k-1$, $\overline{y}_j = 0, \overline{z}_j = 0$. For these terms $p\left(u_j = 1 \mid \cdot\right) = \pi^U_j$ by definition \eqref{eq:pi_U} and (P2), since conditioning on $\overline{a}_j = \overline{d}_j$ is the same as conditioning on $\overline{C}_j = 1$ given the history. Each such term in the factor cancels with $1/\pi^U_j$ in the weight.

Next, the factor $\prod_{j=0}^{k-1} \mathbbm{1}(\overline{a}_j = \overline{d}_j)$ is zero unless $\overline{a}_{k-1}$ equals the  regime sequence. The sum over $\overline{a}_{k-1}$ thus yields a single term in which each treatment density evaluates to $p\left(a_j = d_j \mid \overline{y}_j = 0, \overline{z}_j = 0, \overline{u}_{j-1} = 1, \overline{a}_{j-1} = \overline{d}_{j-1}, \overline{x}_j\right) = \pi^{C}_j$, again by \eqref{eq:pi_C} and (P2). Each such factor cancels with $1/\pi^{C}_j$ in the weight.

Finally, the factor $\prod_{j=1}^{k-1}(1-y_j)(1-z_j)$ restricts the sums over $\left(\overline{y}_{k-1}, \overline{z}_{k-1}\right)$ to the all-zero configuration, leaving free only $(y_k, z_k)$. The outcome density factors for $j \leq k-1$ evaluate at $y_j = 0$, $z_j = 0$, and every conditional law in the display then carries the conditioning values $\overline{y}_{j-1} = 0$, $\overline{z}_{j-1} = 0$, $\overline{u}_{j-1} = 1$, $\overline{a}_{j-1} = \overline{d}_{j-1}$.

After the three steps, all weight factors have canceled, the argument of $g$ is fixed to $\left((\overline{0}, y_k), (\overline{0}, z_k), \overline{x}_{k-1}, \overline{d}_{k-1}(\overline{x}_{k-1}), \overline{1}\right)$, and the remaining integrand is exactly the right-hand side of the lemma.
\end{proof}

\begin{proof}[Proof of Theorem \ref{theorem:EIF}]
    We verify this result by checking that the proposed EIF $\varphi_k$ is a pathwise derivative in that $\left.\frac{\partial}{\partial \varepsilon} \psi_k(P_\varepsilon)\right|_{\varepsilon=0} = \E\left[\varphi_k \ell_{\varepsilon}'(O_k;0)\right]$.
Allowing for continuous covariates, 
\allowdisplaybreaks
\begin{align*}
    \psi_k &= \int_{\overline{x}_{k-1}} \Pr\left(Y_{k} = 1 \:|\: \overline{Z}_k = 0, \overline{Y}_{k-1} = 0, \overline{U}_{k-1} = 1, \overline{C}_{k-1} = 1, \overline{X}_{k-1} = \overline{x}_{k-1}\right)\\
    &\qquad \times \prod_{j=1}^{k} \bigg
    \{\Pr\left(Z_j = 0 \; | \; \overline{Y}_{j-1} = 0, \overline{Z}_{j-1} = 0, \overline{U}_{j-1} = 1, \overline{C}_{j-1} = 1, \overline{X}_{j-1} = \overline{x}_{j-1} \right)\\
    &\qquad\qquad\qquad \times \Pr\left({Y}_{j-1} = 0 \; | \; \overline{Z}_{j-1} = 0, \overline{Y}_{j-2} = 0, \overline{U}_{j-2} = 1, \overline{C}_{j-2} = 1, \overline{X}_{j-2} = \overline{x}_{j-2} \right) \\
    &\qquad\qquad\qquad \times \text{d}P \left( x_{j-1} \mid \overline{Z}_{j-1} = 0, \overline{Y}_{j-1} = 0, \overline{X}_{j-2} = \overline{x}_{j-2}, \overline{C}_{j-2} = 1, \overline{U}_{j-2} = 1 \right) \bigg\}\\
    &= \int_{\overline{x}_{k-1}} \Pr\left(Y_{k} = 1 \:|\: \overline{Z}_k = 0, \overline{Y}_{k-1} = 0, \overline{U}_{k-1} = 1, \overline{A}_{j-2} = \overline{d}_{j-2}, \overline{X}_{k-1} = \overline{x}_{k-1}\right)\\
    &\qquad \times \prod_{j=1}^{k} \bigg
    \{\Pr\left(Z_j = 0 \; | \; \overline{Y}_{j-1} = 0, \overline{Z}_{j-1} = 0, \overline{U}_{j-1} = 1, \overline{A}_{j-2} = \overline{d}_{j-2}, \overline{X}_{j-1} = \overline{x}_{j-1} \right)\\
    &\qquad\qquad\qquad \times \Pr\left({Y}_{j-1} = 0 \; | \; \overline{Z}_{j-1} = 0, \overline{Y}_{j-2} = 0, \overline{U}_{j-2} = 1, \overline{A}_{j-2} = \overline{d}_{j-2}, \overline{X}_{j-2} = \overline{x}_{j-2} \right) \\
    &\qquad\qquad\qquad \times \text{d}P \left( x_{j-1} \mid \overline{Z}_{j-1} = 0, \overline{Y}_{j-1} = 0, \overline{X}_{j-2} = \overline{x}_{j-2}, \overline{A}_{j-2} = \overline{d}_{j-2}, \overline{U}_{j-2} = 1 \right) \bigg\},
\end{align*}
where $\overline{d}_{k-1} = \overline{d}_{k-1}\left( \overline{X}_{k-1}, \overline{Y}_{k-1}, \overline{Z}_{k-1}, \overline{A}_{k-2} \right)$ is a function of previously observed history. We can write the derivative on the LHS as:
\begin{align*}
    \left.\frac{\partial}{\partial \varepsilon} \psi_k(P_\varepsilon)\right|_{\varepsilon=0} &= \psi_k'(0)\\
    &= \int_{\overline{x}_{k-1}} \sum_{\overline{y}_{k}} \sum_{\overline{x}_{k}}\: y_k \left(1 - z_{k}\right) \prod_{j=1}^{k-1} \big\{ \left(1 - y_{j}\right) \left(1 - z_{j}\right) \big\}\\ 
    &\qquad \times \sum_{j=1}^{k} \Biggl\{ \ell'_{\varepsilon} \left( y_j \: \big| \: \overline{z}_j = 0, \overline{y}_{j-1} = 0, \overline{u}_{j-1} = 1, \overline{a}_{j-1} = \overline{d}_{j-1}, \overline{x}_{j-1}; 0\right)\\
    &\qquad\qquad\qquad + \ell'_{\varepsilon} \left( z_j \: \big| \: \overline{y}_{j-1} = 0, \overline{z}_{j-1} = 0, \overline{u}_{j-1} = 1, \overline{a}_{j-1} = \overline{d}_{j-1}, \overline{x}_{j-1}; 0\right)\\
    &\qquad\qquad\qquad + \ell'_{\varepsilon} \left( x_{j-1} \: \big| \: \overline{y}_{j-1} = 0, \overline{z}_{j-1} = 0, \overline{u}_{j-2} = 1, \overline{a}_{j-2} = \overline{d}_{j-2}, \overline{x}_{j-2}; 0\right)\Biggr\}\\
    &\qquad \times \prod_{j=1}^{k} \Biggl\{p \left( y_j \: \big| \: \overline{z}_j = 0, \overline{y}_{j-1} = 0, \overline{u}_{j-1} = 1, \overline{a}_{j-1} = \overline{d}_{j-1}, \overline{x}_{j-1}; 0\right)\\
    &\qquad\qquad\qquad \times p \left( z_j \: \big| \: \overline{y}_{j-1} = 0, \overline{z}_{j-1} = 0, \overline{u}_{j-1} = 1, \overline{a}_{j-1} = \overline{d}_{j-1}, \overline{x}_{j-1}; 0\right)\\
    &\qquad\qquad\qquad \times p \left( x_{j-1} \: \big| \: \overline{y}_{j-1} = 0, \overline{z}_{j-1} = 0, \overline{u}_{j-2} = 1, \overline{a}_{j-2} = \overline{d}_{j-2}, \overline{x}_{j-2}; 0\right)\Biggr\}.
\end{align*}

The expectation on the RHS can be written as:
\begin{align*}
\mathbb{E}\left[\varphi_k \ell'_{\varepsilon}(O_k;0)\right] &= \mathbb{E}\Bigg[ Y_{k}(1-Z_{k}) W_{k-1} \left\{ \prod_{j=0}^{k-1}(1-Y_{j})(1-Z_{j}) \right\} \\
&\qquad\quad \times \sum_{j=1}^{k} \Big\{ \ell'_{\epsilon}\left(Y_{j} \mid \overline{Z}_{j}, \overline{Y}_{j-1}, \overline{U}_{j-1}, \overline{A}_{j-1}, \overline{X}_{j-1}; 0\right) \\
&\qquad\qquad +\ell'_{\epsilon}\left(Z_{j} \mid \overline{Y}_{j-1}, \overline{Z}_{j-1}, \overline{U}_{j-1}, \overline{A}_{j-1}, \overline{X}_{j-1}; 0\right) \\
&\qquad\qquad +\ell'_{\epsilon}\left(U_{j-1} \mid \overline{Y}_{j-1}, \overline{Z}_{j-1}, \overline{U}_{j-2}, \overline{A}_{j-1}, \overline{X}_{j-1}; 0\right) \\
&\qquad\qquad +\ell'_{\epsilon}\left(A_{j-1} \mid \overline{Y}_{j-1}, \overline{Z}_{j-1}, \overline{U}_{j-2}, \overline{A}_{j-2}, \overline{X}_{j-1}; 0\right) \\
&\qquad\qquad +\ell'_{\epsilon}\left(X_{j-1} \mid \overline{Y}_{j-1}, \overline{Z}_{j-1}, \overline{U}_{j-2}, \overline{A}_{j-2}, \overline{X}_{j-2}; 0\right)\Big\}\Bigg] \\
\vspace{1em}
&\quad - \mathbb{E}\Bigg[ W_{k-1} \left\{ \prod_{j=0}^{k-1}(1-Y_{j})(1-Z_{j}) \right\} G_{k}^{(k)} \\
&\qquad\qquad \times \sum_{j=1}^{k} \Big\{ \ell'_{\epsilon}\left(Y_{j} \mid \overline{Z}_{j}, \overline{Y}_{j-1}, \overline{U}_{j-1}, \overline{A}_{j-1}, \overline{X}_{j-1}; 0\right) \\
&\qquad\qquad\quad +\ell'_{\epsilon}\left(Z_{j} \mid \overline{Y}_{j-1}, \overline{Z}_{j-1}, \overline{U}_{j-1}, \overline{A}_{j-1}, \overline{X}_{j-1}; 0\right) \\
&\qquad\qquad\quad +\ell'_{\epsilon}\left(U_{j-1} \mid \overline{Y}_{j-1}, \overline{Z}_{j-1}, \overline{U}_{j-2}, \overline{A}_{j-1}, \overline{X}_{j-1}; 0\right) \\
&\qquad\qquad\quad +\ell'_{\epsilon}\left(A_{j-1} \mid \overline{Y}_{j-1}, \overline{Z}_{j-1}, \overline{U}_{j-2}, \overline{A}_{j-2}, \overline{X}_{j-1}; 0\right) \\
&\qquad\qquad\quad +\ell'_{\epsilon}\left(X_{j-1} \mid \overline{Y}_{j-1}, \overline{Z}_{j-1}, \overline{U}_{j-2}, \overline{A}_{j-2}, \overline{X}_{j-2}; 0\right)\Big\}\Bigg] \\
\vspace{1em}
&\quad +\sum_{t=1}^{k-1} \mathbb{E}\Bigg[ W_{t-1}(1-Y_{t})(1-Z_{t}) \left\{ \prod_{j=0}^{t-1}(1-Y_{j})(1-Z_{j}) \right\} G_{t+1}^{(k)} \\
&\qquad\qquad\quad \times \sum_{j=1}^{k} \Big\{\ell'_{\epsilon}\left(Y_{j} \mid \overline{Z}_{j}, \overline{Y}_{j-1}, \overline{U}_{j-1}, \overline{A}_{j-1}, \overline{X}_{j-1}; 0\right) \\
&\qquad\qquad\qquad +\ell'_{\epsilon}\left(Z_{j} \mid \overline{Y}_{j-1}, \overline{Z}_{j-1}, \overline{U}_{j-1}, \overline{A}_{j-1}, \overline{X}_{j-1}; 0\right) \\
&\qquad\qquad\qquad +\ell'_{\epsilon}\left(U_{j-1} \mid \overline{Y}_{j-1}, \overline{Z}_{j-1}, \overline{U}_{j-2}, \overline{A}_{j-1}, \overline{X}_{j-1}; 0\right) \\
&\qquad\qquad\qquad +\ell'_{\epsilon}\left(A_{j-1} \mid \overline{Y}_{j-1}, \overline{Z}_{j-1}, \overline{U}_{j-2}, \overline{A}_{j-2}, \overline{X}_{j-1}; 0\right) \\
&\qquad\qquad\qquad +\ell'_{\epsilon}\left(X_{j-1} \mid \overline{Y}_{j-1}, \overline{Z}_{j-1}, \overline{U}_{j-2}, \overline{A}_{j-2}, \overline{X}_{j-2}; 0\right)\Big\}\Bigg] \\
\vspace{1em}
&\quad -\sum_{t=1}^{k-1} \mathbb{E}\Bigg[ W_{t-1} \left\{ \prod_{j=0}^{t-1}(1-Y_{j})(1-Z_{j}) \right\} G_{t}^{(k)} \\
&\qquad\qquad\quad \times \sum_{j=1}^{k} \Big\{\ell'_{\epsilon}\left(Y_{j} \mid \overline{Z}_{j}, \overline{Y}_{j-1}, \overline{U}_{j-1}, \overline{A}_{j-1}, \overline{X}_{j-1}; 0\right) \\
&\qquad\qquad\qquad +\ell'_{\epsilon}\left(Z_{j} \mid \overline{Y}_{j-1}, \overline{Z}_{j-1}, \overline{U}_{j-1}, \overline{A}_{j-1}, \overline{X}_{j-1}; 0\right) \\
&\qquad\qquad\qquad +\ell'_{\epsilon}\left(U_{j-1} \mid \overline{Y}_{j-1}, \overline{Z}_{j-1}, \overline{U}_{j-2}, \overline{A}_{j-1}, \overline{X}_{j-1}; 0\right) \\
&\qquad\qquad\qquad +\ell'_{\epsilon}\left(A_{j-1} \mid \overline{Y}_{j-1}, \overline{Z}_{j-1}, \overline{U}_{j-2}, \overline{A}_{j-2}, \overline{X}_{j-1}; 0\right) \\
&\qquad\qquad\qquad +\ell'_{\epsilon}\left(X_{j-1} \mid \overline{Y}_{j-1}, \overline{Z}_{j-1}, \overline{U}_{j-2}, \overline{A}_{j-2}, \overline{X}_{j-2}; 0\right)\Big\}\Bigg] \\
\vspace{1em}
&\quad +\mathbb{E}\Bigg[ G_{1}^{(k)} \times \sum_{j=1}^{k} \Big\{\ell'_{\epsilon}\left(Y_{j} \mid \overline{Z}_{j}, \overline{Y}_{j-1}, \overline{U}_{j-1}, \overline{A}_{j-1}, \overline{X}_{j-1}; 0\right) \\
&\qquad\qquad\qquad +\ell'_{\epsilon}\left(Z_{j} \mid \overline{Y}_{j-1}, \overline{Z}_{j-1}, \overline{U}_{j-1}, \overline{A}_{j-1}, \overline{X}_{j-1}; 0\right) \\
&\qquad\qquad\qquad +\ell'_{\epsilon}\left(U_{j-1} \mid \overline{Y}_{j-1}, \overline{Z}_{j-1}, \overline{U}_{j-2}, \overline{A}_{j-1}, \overline{X}_{j-1}; 0\right) \\
&\qquad\qquad\qquad +\ell'_{\epsilon}\left(A_{j-1} \mid \overline{Y}_{j-1}, \overline{Z}_{j-1}, \overline{U}_{j-2}, \overline{A}_{j-2}, \overline{X}_{j-1}; 0\right) \\
&\qquad\qquad\qquad +\ell'_{\epsilon}\left(X_{j-1} \mid \overline{Y}_{j-1}, \overline{Z}_{j-1}, \overline{U}_{j-2}, \overline{A}_{j-2}, \overline{X}_{j-2}; 0\right)\Big\}\Bigg] \\
\vspace{1em}
&\quad-\mathbb{E}\Bigg[\psi_{k} \times \sum_{j=1}^{k}\Big\{\ell'_{\epsilon}\left(Y_{j} \mid \overline{Z}_{j}, \overline{Y}_{j-1}, \overline{U}_{j-1}, \overline{A}_{j-1}, \overline{X}_{j-1}; 0\right) \\
&\qquad\qquad\qquad +\ell'_{\epsilon}\left(Z_{j} \mid \overline{Y}_{j-1}, \overline{Z}_{j-1}, \overline{U}_{j-1}, \overline{A}_{j-1}, \overline{X}_{j-1}; 0\right) \\
&\qquad\qquad\qquad +\ell'_{\epsilon}\left(U_{j-1} \mid \overline{Y}_{j-1}, \overline{Z}_{j-1}, \overline{U}_{j-2}, \overline{A}_{j-1}, \overline{X}_{j-1}; 0\right) \\
&\qquad\qquad\qquad +\ell'_{\epsilon}\left(A_{j-1} \mid \overline{Y}_{j-1}, \overline{Z}_{j-1}, \overline{U}_{j-2}, \overline{A}_{j-2}, \overline{X}_{j-1}; 0\right) \\
&\qquad\qquad\qquad +\ell'_{\epsilon}\left(X_{j-1} \mid \overline{Y}_{j-1}, \overline{Z}_{j-1}, \overline{U}_{j-2}, \overline{A}_{j-2}, \overline{X}_{j-2}; 0\right)\Big\}\Bigg].
\end{align*}

By the properties of score functions, $\E\left[ \ell'_{\epsilon}\left(A \mid B ; 0\right)\mid B\right] = 0$. Applying this, 
\begin{align*}
\mathbb{E}\left[\varphi_k \ell'_{\varepsilon}(O_k;0)\right] &= \mathbb{E}\Bigg[ Y_{k}(1-Z_{k}) W_{k-1} \left\{ \prod_{j=0}^{k-1}(1-Y_{j})(1-Z_{j}) \right\} \\
&\qquad\quad \times \sum_{j=1}^{k} \Big\{ \ell'_{\epsilon}\left(Y_{j} \mid \overline{Z}_{j}, \overline{Y}_{j-1}, \overline{U}_{j-1}, \overline{A}_{j-1}, \overline{X}_{j-1}; 0\right) \\
&\qquad\qquad +\ell'_{\epsilon}\left(Z_{j} \mid \overline{Y}_{j-1}, \overline{Z}_{j-1}, \overline{U}_{j-1}, \overline{A}_{j-1}, \overline{X}_{j-1}; 0\right) \\
&\qquad\qquad +\ell'_{\epsilon}\left(U_{j-1} \mid \overline{Y}_{j-1}, \overline{Z}_{j-1}, \overline{U}_{j-2}, \overline{A}_{j-1}, \overline{X}_{j-1}; 0\right) \\
&\qquad\qquad +\ell'_{\epsilon}\left(A_{j-1} \mid \overline{Y}_{j-1}, \overline{Z}_{j-1}, \overline{U}_{j-2}, \overline{A}_{j-2}, \overline{X}_{j-1}; 0\right) \\
&\qquad\qquad +\ell'_{\epsilon}\left(X_{j-1} \mid \overline{Y}_{j-1}, \overline{Z}_{j-1}, \overline{U}_{j-2}, \overline{A}_{j-2}, \overline{X}_{j-2}; 0\right)\Big\}\Bigg] \\
\vspace{1em}
&\quad - \mathbb{E}\Bigg[ W_{k-1} \left\{ \prod_{j=0}^{k-1}(1-Y_{j})(1-Z_{j}) \right\} G_{k}^{(k)} \\
&\qquad\qquad \times \sum_{j=1}^{k} \Big\{ \ell'_{\epsilon}\left(Y_{j-1} \mid \overline{Z}_{j-1}, \overline{Y}_{j-2}, \overline{U}_{j-2}, \overline{A}_{j-2}, \overline{X}_{j-2}; 0\right) \\
&\qquad\qquad\quad +\ell'_{\epsilon}\left(Z_{j-1} \mid \overline{Y}_{j-2}, \overline{Z}_{j-2}, \overline{U}_{j-2}, \overline{A}_{j-2}, \overline{X}_{j-2}; 0\right) \\
&\qquad\qquad\quad +\ell'_{\epsilon}\left(U_{j-1} \mid \overline{Y}_{j-1}, \overline{Z}_{j-1}, \overline{U}_{j-2}, \overline{A}_{j-1}, \overline{X}_{j-1}; 0\right) \\
&\qquad\qquad\quad +\ell'_{\epsilon}\left(A_{j-1} \mid \overline{Y}_{j-1}, \overline{Z}_{j-1}, \overline{U}_{j-2}, \overline{A}_{j-2}, \overline{X}_{j-1}; 0\right) \\
&\qquad\qquad\quad +\ell'_{\epsilon}\left(X_{j-1} \mid \overline{Y}_{j-1}, \overline{Z}_{j-1}, \overline{U}_{j-2}, \overline{A}_{j-2}, \overline{X}_{j-2}; 0\right)\Big\}\Bigg] \\
\vspace{1em}
&\quad +\sum_{t=1}^{k-1} \mathbb{E}\Bigg[ W_{t-1}(1-Y_{t})(1-Z_{t}) \left\{ \prod_{j=0}^{t-1}(1-Y_{j})(1-Z_{j}) \right\} G_{t+1}^{(k)} \\
&\qquad\qquad\quad \times \bigg( \ell'_{\epsilon}\left(X_{t} \mid \overline{Y}_{t}, \overline{Z}_{t}, \overline{U}_{t-1}, \overline{A}_{t-1}, \overline{X}_{t-1}\right)\\
&\qquad\qquad\qquad\quad + \sum_{j=1}^{t} \Big\{\ell'_{\epsilon}\left(Y_{j} \mid \overline{Z}_{j}, \overline{Y}_{j-1}, \overline{U}_{j-1}, \overline{A}_{j-1}, \overline{X}_{j-1}; 0\right) \\
&\qquad\qquad\qquad\qquad +\ell'_{\epsilon}\left(Z_{j} \mid \overline{Y}_{j-1}, \overline{Z}_{j-1}, \overline{U}_{j-1}, \overline{A}_{j-1}, \overline{X}_{j-1}; 0\right) \\
&\qquad\qquad\qquad\qquad +\ell'_{\epsilon}\left(U_{j-1} \mid \overline{Y}_{j-1}, \overline{Z}_{j-1}, \overline{U}_{j-2}, \overline{A}_{j-1}, \overline{X}_{j-1}; 0\right) \\
&\qquad\qquad\qquad\qquad +\ell'_{\epsilon}\left(A_{j-1} \mid \overline{Y}_{j-1}, \overline{Z}_{j-1}, \overline{U}_{j-2}, \overline{A}_{j-2}, \overline{X}_{j-1}; 0\right) \\
&\qquad\qquad\qquad\qquad +\ell'_{\epsilon}\left(X_{j-1} \mid \overline{Y}_{j-1}, \overline{Z}_{j-1}, \overline{U}_{j-2}, \overline{A}_{j-2}, \overline{X}_{j-2}; 0\right)\Big\}\bigg)\Bigg] \\
\vspace{1em}
&\quad -\sum_{t=1}^{k-1} \mathbb{E}\Bigg[ W_{t-1} \left\{ \prod_{j=0}^{t-1}(1-Y_{j})(1-Z_{j}) \right\} G_{t}^{(k)} \\
&\qquad\qquad\quad \times \sum_{j=1}^{k} \Big\{\ell'_{\epsilon}\left(Y_{j-1} \mid \overline{Z}_{j-1}, \overline{Y}_{j-2}, \overline{U}_{j-2}, \overline{A}_{j-2}, \overline{X}_{j-2}; 0\right) \\
&\qquad\qquad\qquad +\ell'_{\epsilon}\left(Z_{j-1} \mid \overline{Y}_{j-2}, \overline{Z}_{j-2}, \overline{U}_{j-2}, \overline{A}_{j-2}, \overline{X}_{j-2}; 0\right) \\
&\qquad\qquad\qquad +\ell'_{\epsilon}\left(U_{j-1} \mid \overline{Y}_{j-1}, \overline{Z}_{j-1}, \overline{U}_{j-2}, \overline{A}_{j-1}, \overline{X}_{j-1}; 0\right) \\
&\qquad\qquad\qquad +\ell'_{\epsilon}\left(A_{j-1} \mid \overline{Y}_{j-1}, \overline{Z}_{j-1}, \overline{U}_{j-2}, \overline{A}_{j-2}, \overline{X}_{j-1}; 0\right) \\
&\qquad\qquad\qquad +\ell'_{\epsilon}\left(X_{j-1} \mid \overline{Y}_{j-1}, \overline{Z}_{j-1}, \overline{U}_{j-2}, \overline{A}_{j-2}, \overline{X}_{j-2}; 0\right)\Big\}\Bigg] \\
\vspace{1em}
&\quad +\mathbb{E}\Bigg[ G_{1}^{(k)} \times \ell'_{\epsilon}\left(X_0; 0\right)\Bigg].
\end{align*}

{Finally, the weighted terms telescope: applying the weight-reduction identity of Lemma \ref{lemma:wt_reduction} to each summand, the second through fifth groups of terms cancel among themselves by the same recursive argument used in the g-formula derivation of Theorem \ref{theorem:g-formula_identification}, leaving only the weighted leading term}
\begin{align*}
    \E\left[\varphi_k \ell_{\varepsilon}'(O_k;0)\right] &= \mathbb{E}\Bigg[ Y_{k}(1-Z_{k})\, {W_{k-1}} \left\{ \prod_{j=0}^{k-1}(1-Y_{j})(1-Z_{j}) \right\} \\
&\qquad\quad \times \sum_{j=1}^{k} \Big\{ \ell'_{\epsilon}\left(Y_{j} \mid \overline{Z}_{j}, \overline{Y}_{j-1}, \overline{U}_{j-1}, \overline{A}_{j-1}, \overline{X}_{j-1}; 0\right) \\
&\qquad\qquad +\ell'_{\epsilon}\left(Z_{j} \mid \overline{Y}_{j-1}, \overline{Z}_{j-1}, \overline{U}_{j-1}, \overline{A}_{j-1}, \overline{X}_{j-1}; 0\right) \\
&\qquad\qquad +\ell'_{\epsilon}\left(X_{j-1} \mid \overline{Y}_{j-1}, \overline{Z}_{j-1}, \overline{U}_{j-2}, \overline{A}_{j-2}, \overline{X}_{j-2}; 0\right)\Big\}\Bigg].
\end{align*}
{Applying the weight-collapse identity of Lemma \ref{lemma:wt_collapse} to this expectation, with $g$ equal to $Y_k(1-Z_k)$ times the displayed sum of scores, converts the weighted, at-risk expectation into the regime-fixed path integral}
\begin{align*}
&= \int_{\overline{x}_{k-1}} \sum_{\overline{y}_{k}} \sum_{\overline{x}_{k}}\: y_k \left(1 - z_{k}\right) \prod_{j=1}^{k-1} \Biggl\{ \left(1 - y_{j}\right) \left(1 - z_{j}\right) \Biggr\}\\
    &\qquad \times \sum_{j=1}^{k} \Biggl\{ \ell'_{\varepsilon} \left( y_j \: \big| \: \overline{z}_j = 0, \overline{y}_{j-1} = 0, \overline{u}_{j-1} = 1, \overline{a}_{j-1} = \overline{d}_{j-1}, \overline{x}_{j-1}; 0\right)\\
    &\qquad\qquad\qquad + \ell'_{\varepsilon} \left( z_j \: \big| \: \overline{y}_{j-1} = 0, \overline{z}_{j-1} = 0, \overline{u}_{j-1} = 1, \overline{a}_{j-1} = \overline{d}_{j-1}, \overline{x}_{j-1}; 0\right)\\
    &\qquad\qquad\qquad + \ell'_{\varepsilon} \left( x_{j-1} \: \big| \: \overline{y}_{j-1} = 0, \overline{z}_{j-1} = 0, \overline{u}_{j-2} = 1, \overline{a}_{j-2} = \overline{d}_{j-2}, \overline{x}_{j-2}; 0\right)\Biggr\}\\
    &\qquad \times \prod_{j=1}^{k} \Biggl\{p \left( y_j \: \big| \: \overline{z}_j = 0, \overline{y}_{j-1} = 0, \overline{u}_{j-1} = 1, \overline{a}_{j-1} = \overline{d}_{j-1}, \overline{x}_{j-1}; 0\right)\\
    &\qquad\qquad\qquad \times p \left( z_j \: \big| \: \overline{y}_{j-1} = 0, \overline{z}_{j-1} = 0, \overline{u}_{j-1} = 1, \overline{a}_{j-1} = \overline{d}_{j-1}, \overline{x}_{j-1}; 0\right)\\
    &\qquad\qquad\qquad \times p \left( x_{j-1} \: \big| \: \overline{y}_{j-1} = 0, \overline{z}_{j-1} = 0, \overline{u}_{j-2} = 1, \overline{a}_{j-2} = \overline{d}_{j-2}, \overline{x}_{j-2}; 0\right)\Biggr\}.
\end{align*}
This equals the LHS exactly; $\left.\frac{\partial}{\partial \varepsilon} \psi_k(P_\varepsilon)\right|_{\varepsilon=0} = \E\left[\varphi_k \ell_{\varepsilon}'(O_k;0)\right]$, so $\varphi_k$ is the EIF for $\psi_k$.
\end{proof}

\section{Asymptotic Results}\label{sec:asymptotics_proof}

To prove Theorem \ref{theorem:EIF_asymptotics}, we first prove several lemmas concerning the order of the empirical process and the second-order remainder terms. Let $\hat{\P}$ be the estimated distribution of the observed data, and $\P$ be the true distribution.

\begin{lemma}\label{lemma:T1_term}
    For observed data $O_k = \{\overline{Y}_k, \overline{Z}_k, \overline{U}_{k-1}, \overline{A}_{k-1}, \overline{X}_{k-1}\}$, the empirical process term $T_1 = \left(\P_n - \P\right)\left[\varphi(O_k;\hat\P) - \varphi(O_k;\P)\right]$ is of the order $o_p(1/\sqrt{n})$ if the following conditions hold: for $t=1,\dots,k$, $\left\|\hat{W}_{t-1} - {W}_{t-1}\right\| = o_p(1)$ and $\left\|\hat{G}_{t}^{(k)} - {G}_{t}^{(k)}\right\| = o_p(1)$; and there exists an $M < \infty$ such that, for $\Bcal_t = (1-Y_t)(1-Z_t)$ and $\Bcal^*_k = Y_k(1-Z_k)$, $\left|\Bcal_k^{*} - \hat{G}_k^{(k)}\right| < M$, {$\left|W_{t-1}\right| < M$ for $t = 1,\dots,k$}, and $\left|\Bcal_t \hat{G}_{t+1}^{(k)} - \hat{G}_t^{(k)}\right| < M$ for $t=1,\dots,k-1$.
\end{lemma}

\begin{proof}[Proof of Lemma \ref{lemma:T1_term}]
First, we require $\varphi(O;\hat\P)$ to be converging to $\varphi(O;\P)$ in $L_2(\P)$ norm. In other words, we show, that under certain simple conditions, $\left\|\varphi(;\hat\P) - \varphi(;\P)\right\|^2 = o_p(1)$.
\allowdisplaybreaks
\begin{align*}
    T_1 &= \left(\P_n - \P\right) \left[\varphi(O;\hat\P) - \varphi(O;\P)\right]\\
    &= \left(\P_n - \P\right) 
    \Biggl[ \hat{W}_{k-1} \left\{ \prod_{j=0}^{k-1} (1-Y_j)(1-Z_j) \right\} \Big\{ Y_k (1-Z_k) - \hat{G}_k^{(k)} \Big\}\\
    &\qquad\qquad\qquad + \sum_{t = 1}^{k-1} \hat{W}_{t-1} \left\{ \prod_{j=0}^{t-1} (1-Y_j)(1-Z_j) \right\} \Big\{ (1 - Y_t) (1-Z_t) \hat{G}_{t+1}^{(k)} - \hat{G}_t^{(k)} \Big\}\\
    &\qquad\qquad\qquad + \hat{G}_1^{(k)} \\
    &\qquad\qquad\qquad - W_{k-1} \left\{ \prod_{j=0}^{k-1} (1-Y_j)(1-Z_j) \right\} \Big\{ Y_k (1-Z_k) - G_k^{(k)} \Big\}\nonumber\\
    &\qquad\qquad\qquad - \sum_{t = 1}^{k-1} W_{t-1} \left\{ \prod_{j=0}^{t-1} (1-Y_j)(1-Z_j) \right\} \Big\{ (1 - Y_t) (1-Z_t)G_{t+1}^{(k)} - G_t^{(k)} \Big\}\nonumber\\
    &\qquad\qquad\qquad - G_1^{(k)}\Biggr].
\end{align*}
Let $\Acal_t = \left\{\prod_{j=0}^{t-1}(1-Y_j)(1-Z_j)\right\}$, $\Bcal_t = (1-Y_t)(1-Z_t)$, and $\Bcal^*_k = Y_k(1-Z_k)$. Then,
\begin{align*}
    T_1 &= \left(\P_n - \P\right) 
    \Biggl[ \hat{W}_{k-1} \Acal_k \left\{ \Bcal^*_k - \hat{G}_k^{(k)} \right\} + \sum_{t = 1}^{k-1} \hat{W}_{t-1} \Acal_t \left\{\Bcal_t \hat{G}_{t+1}^{(k)} - \hat{G}_t^{(k)} \right\} + \hat{G}_1^{(k)} \\
    &\qquad\qquad\qquad - W_{k-1} \Acal_k \left\{ \Bcal^*_k - {G}_k^{(k)} \right\} {-} \sum_{t = 1}^{k-1} {W}_{t-1} \Acal_t \left\{\Bcal_t {G}_{t+1}^{(k)} - {G}_t^{(k)} \right\} - {G}_1^{(k)} \Biggr] \\
    &= \left(\P_n - \P\right) 
    \Biggl[ \Acal_k \left[\hat{W}_{k-1} \left\{ \Bcal^*_k - \hat{G}_k^{(k)}\right\} - W_{k-1}\left\{\Bcal^*_k - {G}_k^{(k)}\right\}\right] \tag*{\textcircled{4}}\\
    &\qquad\qquad\qquad+ \sum_{t = 1}^{k-1} \Acal_t \left[\hat{W}_{t-1}  \left\{\Bcal_t \hat{G}_{t+1}^{(k)} - \hat{G}_t^{(k)}\right\} - {W}_{t-1} \left\{\Bcal_t {G}_{t+1}^{(k)} - {G}_t^{(k)}\right\} \right]\tag*{\textcircled{5}}\\
    &\qquad\qquad\qquad+ \hat{G}_1^{(k)} - {G}_1^{(k)} \Biggr].
\end{align*}
We can rewrite {\textcircled{4}} as follows:
\begin{align*}
    & \Acal_k \left[\hat{W}_{k-1} \left\{\Bcal_k^{*} - \hat{G}_k^{(k)}\right\} - W_{k-1}\left\{\Bcal_k^{*} - G_k^{(k)}\right\} \right]\\
    &\qquad = \Acal_k \left[\hat{W}_{k-1} \left\{\Bcal_k^{*} - \hat{G}_k^{(k)}\right\} - W_{k-1}\left\{\Bcal_k^{*} - \hat{G}_k^{(k)}\right\}\right. \\
    &\qquad\qquad\qquad + \left.W_{k-1}\left\{\Bcal_k^{*} - \hat{G}_k^{(k)}\right\} - W_{k-1}\left\{\Bcal_k^{*} - G_k^{(k)}\right\} \right]\\
    &\qquad = \Acal_k \left[\left(\Bcal_k^{*} - \hat{G}_k^{(k)}\right)\left(\hat{W}_{k-1} - {W}_{k-1}\right) - {W}_{k-1}\left(\hat{G}_k^{(k)} - {G}_k^{(k)}\right) \right] \\
    &\qquad = \Acal_k\left(\Bcal_k^{*} - \hat{G}_k^{(k)}\right)\left(\hat{W}_{k-1} - {W}_{k-1}\right) - \Acal_k{W}_{k-1}\left(\hat{G}_k^{(k)} - {G}_k^{(k)}\right)
\end{align*}
Similarly, {\textcircled{5}} can be expressed as follows:
\begin{align*}
    &\sum_{t = 1}^{k-1} \Acal_t \left[\hat{W}_{t-1}  \left\{\Bcal_t \hat{G}_{t+1}^{(k)} - \hat{G}_t^{(k)}\right\} - {W}_{t-1}\left\{\Bcal_t {G}_{t+1}^{(k)} - {G}_t^{(k)}\right\} \right]\\
    &\qquad =  \sum_{t = 1}^{k-1} \Acal_t \left[\hat{W}_{t-1}  \left\{\Bcal_t \hat{G}_{t+1}^{(k)} - \hat{G}_t^{(k)}\right\} - {W}_{t-1}\left\{\Bcal_t \hat{G}_{t+1}^{(k)} - \hat{G}_t^{(k)}\right\}\right.\\
    &\qquad\qquad\qquad\quad + \left. {W}_{t-1} \left\{\Bcal_t \hat{G}_{t+1}^{(k)} - \hat{G}_t^{(k)}\right\}
    - {W}_{t-1}\left\{\Bcal_t {G}_{t+1}^{(k)} - {G}_t^{(k)}\right\} \right]\\
    &\qquad = \sum_{t = 1}^{k-1} \Acal_t \left[ \left(\Bcal_t \hat{G}_{t+1}^{(k)} - \hat{G}_t^{(k)}\right) \left(\hat{W}_{t-1} - {W}_{t-1}\right) \right.\\
    &\qquad\qquad\qquad\quad + \left.{W}_{t-1}\Bcal_t \left(\hat{G}_{t+1}^{(k)} - {G}_{t+1}^{(k)}\right) - {W}_{t-1}\left(\hat{G}_{t}^{(k)} - {G}_{t}^{(k)}\right) \right]
\end{align*}
Then, the empirical process term $T_1$ can be rewritten as:
\begin{align*}
    T_1 &= \left(\P_n - \P\right) 
    \Biggl[ \Acal_k\left(\Bcal_k^{*} - \hat{G}_k^{(k)}\right)\left(\hat{W}_{k-1} - {W}_{k-1}\right) 
    - \Acal_k{W}_{k-1}\left(\hat{G}_k^{(k)} - {G}_k^{(k)}\right) \\
    &\qquad\qquad\qquad + \sum_{t = 1}^{k-1} \Acal_t \left[ \left(\Bcal_t \hat{G}_{t+1}^{(k)} - \hat{G}_t^{(k)}\right) \left(\hat{W}_{t-1} - {W}_{t-1}\right) \right.\\
    &\qquad\qquad\qquad\qquad\qquad\quad + \left.{W}_{t-1}\Bcal_t \left(\hat{G}_{t+1}^{(k)} - {G}_{t+1}^{(k)}\right) - {W}_{t-1}\left(\hat{G}_{t}^{(k)} - {G}_{t}^{(k)}\right) \right]\\
    &\qquad\qquad\qquad + \left(\hat{G}_1^{(k)} - {G}_1^{(k)}\right) \Biggr]
\end{align*}
Note $\Acal_t, \Bcal_t,$ and $\Bcal^*_k$ are all binary. Thus, the following conditions must hold for $\varphi(O;\hat\P)$ to be converging to $\varphi(O;\P)$ in $L_2(\P)$ norm:
\begin{enumerate}
    \item There exists an $M < \infty$ such that, with probability 1: 
    \begin{enumerate}
        \item $\left|\Bcal_k^{*} - \hat{G}_k^{(k)}\right| \leq M$,
        \item $\left|W_{t-1}\right| \leq M$ for $t=1,\dots,{k}$, and
        \item $\left|\Bcal_t \hat{G}_{t+1}^{(k)} - \hat{G}_t^{(k)}\right| \leq M$ for $t=1,\dots,k-1$;
    \end{enumerate}
    \item $\left\|\hat{W}_{t-1} - {W}_{t-1}\right\| = o_p(1)$ and $\left\|\hat{G}_{t}^{(k)} - {G}_{t}^{(k)}\right\| = o_p(1)$ for $t=1,\dots,k$.
\end{enumerate}
Under these conditions, it follows that $T_1$ is a sample average of a quantity tending to zero, meaning that $\sqrt{n}T_1 = o_p(1)$. 
\end{proof}

The next lemma provides the two conditional-expectation reductions used to analyze the remainder term. Define the estimated cumulative weights and propensity ratios
\begin{equation*}
    \hat{w}_t = \prod_{j=0}^{t} \frac{U_j C_j}{\hat{\pi}^U_j \hat{\pi}^{C}_j}, \qquad \hat{w}_{-1} = 1, \qquad
    \rho_j = \frac{\pi^U_j \pi^{C}_j}{\hat{\pi}^U_j \hat{\pi}^{C}_j},
\end{equation*}
so that $\hat{W}_{t} = \hat{w}_{t}$ in the notation of the main text.

\begin{lemma}\label{lemma:T2_intermediate_result}
Suppose $\hat\pi^U_j, \hat\pi^{C}_j > 0$ almost surely. Then:
\begin{enumerate}
    \item[(i)] for any $t = 0, \dots, k-1$ and any integrable function $f$ of $\overline{X}_t$,
    \begin{equation*}
        \E\left[\hat{w}_t\, \Acal_{t+1}\, f\left(\overline{X}_t\right)\right] = \E\left[\hat{w}_{t-1}\, \rho_t\, \Acal_{t+1}\, f\left(\overline{X}_t\right)\right];
    \end{equation*}
    \item[(ii)] for $t = 1, \dots, k-1$, $\E\left[\hat{w}_{t-1}\, \Acal_t \left\{\Bcal_t\, G^{(k)}_{t+1} - G^{(k)}_t\right\}\right] = 0$, and $\E\left[\hat{w}_{k-1}\, \Acal_k \left\{\Bcal^*_k - G^{(k)}_k\right\}\right] = 0$.
\end{enumerate}
\end{lemma}

\begin{proof}[Proof of Lemma \ref{lemma:T2_intermediate_result}]
\textit{Part (i).} The argument is that of Lemma \ref{lemma:wt_reduction}, with one difference: the denominators $\hat{\pi}^U_t, \hat{\pi}^{C}_t$ are fixed functions that do not cancel against the conditional means, which instead produce the true propensities and hence the ratio $\rho_t$. In detail, write $\hat{w}_t = \hat{w}_{t-1} \times U_t C_t / (\hat{\pi}^U_t \hat{\pi}^{C}_t)$. Conditioning on $\widetilde{O}_t = \left(\overline{X}_t, \overline{Y}_t, \overline{Z}_t, \overline{A}_t, \overline{U}_{t-1}\right)$, the collection of everything realized strictly before $U_t$: every factor other than $U_t$ is a function of $\widetilde{O}_t$ ($\hat{w}_{t-1}$ and the fixed functions $\hat{\pi}^U_t, \hat{\pi}^{C}_t$ included), the integrand vanishes off the event $\left\{\overline{C}_t = 1, \overline{U}_{t-1} = 1, \overline{Y}_t = 0, \overline{Z}_t = 0\right\}$, and on that event $\E\left[U_t \mid \widetilde{O}_t\right] = \pi^U_t$ by the first display of (P3). This replaces $U_t / \hat{\pi}^U_t$ by $\pi^U_t / \hat{\pi}^U_t$ inside the expectation. Next, on $\{C_{t-1} = 1\}$ we have $C_t = \mathbbm{1}(A_t = d_t)$; conditioning on $\widetilde{O}'_t = \left(\overline{X}_t, \overline{Y}_t, \overline{Z}_t, \overline{A}_{t-1}, \overline{U}_{t-1}\right)$ and applying the second display of (P3) replaces $\mathbbm{1}(A_t = d_t) / \hat{\pi}^{C}_t$ by $\pi^{C}_t / \hat{\pi}^{C}_t$. Collecting the two ratios gives $\rho_t$ and part (i). (For $t = 0$, the same two steps apply with $\hat{w}_{-1} = 1$ and $\Acal_1 = 1$.)

\textit{Part (ii).} By the discussion preceding Lemma \ref{lemma:T2_intermediate_result}, $\hat{w}_{t-1}$ and $\Acal_t$ are functions of $\widetilde{O}_t$, and $\hat{w}_{t-1}\Acal_t = \hat{w}_{t-1}\Acal_t \mathbbm{1}(E_t)$ because the numerator of $\hat{w}_{t-1}$ contains $\prod_{j=0}^{t-1} U_j C_j$. Hence, by the law of iterated expectations and the first display of (P5),
\begin{equation*}
    \E\left[\hat{w}_{t-1} \Acal_t\, \Bcal_t G^{(k)}_{t+1}\right]
    = \E\left[\hat{w}_{t-1} \Acal_t \mathbbm{1}(E_t)\, \E\left\{\Bcal_t G^{(k)}_{t+1} \mid \widetilde{O}_t \right\}\right]
    = \E\left[\hat{w}_{t-1} \Acal_t\, G^{(k)}_t\right],
\end{equation*}
which is the first claim; the second follows identically from the second display of (P5) at $t = k$.
\end{proof}

\begin{lemma}\label{lemma:T2_term}
Let $R_2\left(\hat{\P}, \P\right) = \psi_k\left(\hat{\P}\right) - \psi_k\left(\P\right) + \int\varphi_k\left(v;\hat{\P}\right)\, \mathrm{d}\P(v)$ denote the second-order remainder in the von Mises expansion of $\psi_k$ \citep{kennedy_semiparametric_2023}. Then, 
\begin{equation}\label{eq:R2-exact}
    R_2\left(\hat{\P}, \P\right) = \sum_{t=1}^{k} b_t, \; \text{where} \;\;
    b_t = \E\left[\hat{w}_{t-2}\, \left(\rho_{t-1} - 1\right) \Acal_t \left(G^{(k)}_t - \hat{G}^{(k)}_t\right)\right],
\end{equation}
in which $\hat{w}_{t} = \prod_{j=0}^{t} U_j C_j / (\hat{\pi}^U_j \hat{\pi}^{C}_j)$ and $\rho_{j} = \pi^U_j \pi^{C}_j / (\hat{\pi}^U_j \hat{\pi}^{C}_j)$ are the estimated cumulative weights and propensity ratios defined preceding Lemma \ref{lemma:T2_intermediate_result}, with $\hat{w}_{-1} = 1$, and $\Acal_t = \prod_{j=0}^{t-1}(1-Y_j)(1-Z_j)$ is the survival indicator through time $t-1$. Consequently, if there exists $\epsilon > 0$ such that $\prod_{j=0}^{t-1}\hat{\pi}^U_j \hat{\pi}^{C}_j \geq \epsilon$ w.p.\ 1 for $t = 1, \dots, k$, and the product rate condition \eqref{eq:sequential_rate_condition} holds, then
\begin{equation*}
    R_2(\hat{\P}, \P) \leq \left|R_2\left(\hat{\P}, \P\right)\right|
    \leq \frac{1}{\epsilon}\sum_{t=1}^{k} \left\| \frac{1}{\hat{\pi}_{t-1}^U \hat{\pi}_{t-1}^{C}} - \frac{1}{\pi_{t-1}^U \pi_{t-1}^{C}} \right\| \left\| G_t^{(k)} - \hat{G}_t^{(k)} \right\|
    = o_p\left(1/\sqrt{n}\right).
\end{equation*}
\end{lemma}

\begin{proof}[Proof of Lemma \ref{lemma:T2_term}]
Substituting the form of $\varphi_k(\cdot\,;\hat\P)$ (Theorem \ref{theorem:EIF} with every nuisance replaced by its estimate) into the definition of $R_2$, the terms $\psi_k(\hat\P)$ and $-\psi_k(\hat\P)$ cancel, leaving
\begin{equation*}
    R_2\left(\hat{\P}, \P\right)
    = \E\left[\hat{G}^{(k)}_1\right] - \psi_k(\P) + \sum_{t=1}^{k} m_t, \qquad
    \begin{cases}
        m_t = \E\left[\hat{w}_{t-1} \Acal_t \left(\Bcal_t \hat{G}^{(k)}_{t+1} - \hat{G}^{(k)}_t\right)\right], & t \leq k-1,\\[0.4em]
        m_k = \E\left[\hat{w}_{k-1} \Acal_k \left(\Bcal^*_k - \hat{G}^{(k)}_k\right)\right],
    \end{cases}
\end{equation*}
where all expectations are over the analysis-sample data under $\P$ with the nuisance estimates fixed. Define, for $t = 1, \dots, k$,
\begin{equation*}
    a_t = \E\left[\hat{w}_{t-1}\, \Acal_t \left(G^{(k)}_t - \hat{G}^{(k)}_t\right)\right],
\end{equation*}
and recall $b_t$ from \eqref{eq:R2-exact}. We establish three identities.

\textit{Identity 1: $m_k = a_k$.} Decompose $\Bcal^*_k - \hat{G}^{(k)}_k = \left(\Bcal^*_k - G^{(k)}_k\right) + \left(G^{(k)}_k - \hat{G}^{(k)}_k\right)$. The first piece contributes zero under the weighted expectation by Lemma \ref{lemma:T2_intermediate_result}(ii); the second is $a_k$ by definition.

\textit{Identity 2: $m_t = a_t - a_{t+1} + b_{t+1}$ for $t = 1, \dots, k-1$.} Decompose
\begin{equation*}
    \Bcal_t \hat{G}^{(k)}_{t+1} - \hat{G}^{(k)}_t
    = \Bcal_t\left(\hat{G}^{(k)}_{t+1} - G^{(k)}_{t+1}\right) + \left(\Bcal_t G^{(k)}_{t+1} - G^{(k)}_t\right) + \left(G^{(k)}_t - \hat{G}^{(k)}_t\right).
\end{equation*}
Taking $\E\left[\hat{w}_{t-1}\Acal_t \times \cdot\,\right]$ of each of the three terms above: the middle term contributes zero by Lemma \ref{lemma:T2_intermediate_result}(ii), and the last term is $a_t$. For the first term, use $\Acal_t \Bcal_t = \Acal_{t+1}$ and apply Lemma \ref{lemma:T2_intermediate_result}(i) with $f = G^{(k)}_{t+1} - \hat{G}^{(k)}_{t+1}$, a function of $\overline{X}_t$:
\begin{equation*}
    a_{t+1} = \E\left[\hat{w}_t\, \Acal_{t+1}\left(G^{(k)}_{t+1} - \hat{G}^{(k)}_{t+1}\right)\right]
    = \E\left[\hat{w}_{t-1}\, \rho_t\, \Acal_{t+1}\left(G^{(k)}_{t+1} - \hat{G}^{(k)}_{t+1}\right)\right],
\end{equation*}
so that
\begin{align*}
    \E\left[\hat{w}_{t-1} \Acal_{t+1}\left(\hat{G}^{(k)}_{t+1} - G^{(k)}_{t+1}\right)\right]
    &= -\,\E\left[\hat{w}_{t-1} \Acal_{t+1}\left(G^{(k)}_{t+1} - \hat{G}^{(k)}_{t+1}\right)\right]\\
    &= -\,a_{t+1} + \E\left[\hat{w}_{t-1}\left(\rho_t - 1\right) \Acal_{t+1}\left(G^{(k)}_{t+1} - \hat{G}^{(k)}_{t+1}\right)\right]\\
    &= -\,a_{t+1} + b_{t+1}.
\end{align*}

\textit{Identity 3: $\E\left[\hat{G}^{(k)}_1\right] - \psi_k(\P) = -\,a_1 + b_1$.} By Theorem \ref{theorem:g-formula_identification}, $\psi_k(\P) = \E\left[G^{(k)}_1\right]$, so the left-hand side is $-\E\left[G^{(k)}_1 - \hat{G}^{(k)}_1\right]$. Applying Lemma \ref{lemma:T2_intermediate_result}(i) at $t = 0$ (with $\hat{w}_{-1} = 1$, $\Acal_1 = 1$, and $f = G^{(k)}_1 - \hat{G}^{(k)}_1$, a function of $X_0$) gives $a_1 = \E\left[\rho_0\left(G^{(k)}_1 - \hat{G}^{(k)}_1\right)\right]$, thus
\begin{equation*}
    -\,\E\left[G^{(k)}_1 - \hat{G}^{(k)}_1\right] = -\,a_1 + \E\left[\left(\rho_0 - 1\right)\left(G^{(k)}_1 - \hat{G}^{(k)}_1\right)\right] = -\,a_1 + b_1.
\end{equation*}

Summing the three identities, the $a_t$ terms telescope:
\begin{equation*}
    R_2\left(\hat{\P}, \P\right)
    = \left(-a_1 + b_1\right) + \sum_{t=1}^{k-1}\left(a_t - a_{t+1} + b_{t+1}\right) + a_k
    = \sum_{t=1}^{k} b_t,
\end{equation*}
which proves \eqref{eq:R2-exact}. This representation makes the sequential double robustness explicit: each $b_t$ is an expectation of a product of a time-$(t-1)$ propensity error, through $\rho_{t-1} - 1$, and a time-$t$ outcome-model error, through $G^{(k)}_t - \hat{G}^{(k)}_t$. In particular, $R_2 = 0$ exactly if, at every $t$, either the propensity estimates or the outcome-model estimate is correct.

For the bound, note $0 \leq \hat{w}_{t-2}\, \Acal_t \leq 1/\epsilon$ (the numerator indicators are at most one, and $\prod_{j=0}^{t-2}\hat{\pi}^U_j\hat{\pi}^{C}_j \geq \epsilon$ follows from the stated condition since each $\hat\pi \leq 1$), and
\begin{equation*}
    \left|\rho_{t-1} - 1\right|
    = \pi^U_{t-1}\pi^{C}_{t-1} \left|\frac{1}{\hat{\pi}_{t-1}^U \hat{\pi}_{t-1}^{C}} - \frac{1}{\pi_{t-1}^U \pi_{t-1}^{C}}\right|
    \leq \left|\frac{1}{\hat{\pi}_{t-1}^U \hat{\pi}_{t-1}^{C}} - \frac{1}{\pi_{t-1}^U \pi_{t-1}^{C}}\right|,
\end{equation*}
because $\pi^U_{t-1}\pi^{C}_{t-1} \leq 1$. Hence, by the Cauchy--Schwarz inequality applied to each $b_t$,
\begin{equation*}
    \left|R_2\left(\hat{\P}, \P\right)\right|
    \leq \frac{1}{\epsilon}\sum_{t=1}^{k} \left\| \frac{1}{\hat{\pi}_{t-1}^U \hat{\pi}_{t-1}^{C}} - \frac{1}{\pi_{t-1}^U \pi_{t-1}^{C}} \right\| \left\| G_t^{(k)} - \hat{G}_t^{(k)} \right\|
    = o_p\left(1/\sqrt{n}\right)
\end{equation*}
under the product rate condition \eqref{eq:sequential_rate_condition}.
\end{proof}

\begin{proof}[Proof of Theorem \ref{theorem:EIF_asymptotics}]
    For observed data $O_k = \left\{\overline{Y}_k, \overline{Z}_k, \overline{U}_{k-1}, \overline{A}_{k-1}, \overline{X}_{k-1}\right\}$, estimated distribution $\hat{\P}$, and true distribution $\P$, the one-step estimator in (\ref{eq:EIF-based_estimator}) can decomposed as follows, as shown by \citet{kennedy_semiparametric_2023}:
    \allowdisplaybreaks
\begin{align*}
    \hat\psi_k - \psi_k &= \psi\left(\hat\P\right) + \P_n\left\{\varphi_k(O_k;\hat\P)\right\} - \psi\left(\P\right)\\
    &= \left(\P_n - \P\right)\left\{\varphi_k(O_k;\hat\P)\right\} + R_2\left(\hat{\P}, \P\right)\\
    &= \left(\P_n - \P\right)\left\{\varphi_k(O_k;\P)\right\} + \left(\P_n - \P\right)\left\{\varphi_k(O_k;\hat\P)-\varphi_k(O_k;\P)\right\}+ R_2\left(\hat{\P}, \P\right)\\
    &= \frac{1}{n} \sum_{i=1}^n \varphi_{i,k} + T_1 + T_2.
\end{align*}

The first term represents a simple sample average of a fixed function. By the Central Limit Theorem (CLT), it therefore approximates a normally distributed random variable with variance $\text{Var}\left(\varphi_k\right)$, up to an error term of order $\text{o}_p({1/\sqrt{n}})$. What remains to show is that the empirical process term $T_1$ and the remainder term $T_2$ are of the order $\text{o}_p({1/\sqrt{n}})$, which would imply that $\sqrt{n}\left(\hat\psi_k - \psi_k\right) = \frac{1}{\sqrt{n}}\sum_{i=1}^n \varphi_{i,k} + o_p(1) \leadsto \text{N}\left(0, \text{Var}\left(\varphi_k\right)\right)$. 

Under the assumptions of this Theorem, by Lemma \ref{lemma:T1_term}, given that nuisance parameters are estimated from a distinct and independent sample, the empirical process term $T_1 = \left(\P_n - \P\right)\left[\varphi(O_k;\hat\P) - \varphi(O_k;\P)\right]$ is of the order $o_p(1/\sqrt{n})$. Additionally, by Lemma \ref{lemma:T2_term}, the second-order remainder term $R_2(\hat{\P}, \P)$ of the sequential doubly robust estimator for $\psi_k$ is of order $\text{o}_p(1/\sqrt{n})$. As such,
\begin{align*}
    \hat\psi_k - \psi_k &= \frac{1}{n} \sum_{i=1}^n \varphi_{i,k} + T_1 + T_2\\
    &= \frac{1}{n} \sum_{i=1}^n \varphi_{i,k} + \text{o}_p(1/\sqrt{n}).
\end{align*}

Then,
\begin{equation*}
    \sqrt{n}\left(\hat\psi_k - \psi_k\right) = \frac{1}{\sqrt{n}}\sum_{i=1}^n \varphi_{i,k} + o_p(1),
\end{equation*}
and $\hat\psi_k$ is asymptotically linear.
\end{proof}

\section{Detailed Simulation Settings}\label{sec:supp_sim_details}

This section provides the exact parameterizations of the data-generating mechanism outlined in Section \ref{sec:sims-dgp}.

\subsection{Covariates}
At baseline ($k=0$), each time-varying covariate was independently drawn such that:
\begin{equation*}
    X_{i,0}^{(m)} \sim U(-0.5, 0.5), \quad m=1, \dots, 4.
\end{equation*}
For each subsequent time point $k=1, \dots, 8$, $X_{i,0}^{(1)},\dots X_{i,0}^{(4)}$ were updated by adding a small, random perturbation, $\epsilon_{i,k}^{(m)}$, to the previous time point's values according to the formula:
\begin{equation*}
    X_{i,k}^{(m)} = X_{i,k-1}^{(m)} + \epsilon_{i,k}^{(m)}, \quad \text{where } \epsilon_{i,k}^{(m)} \sim U(-0.05, 0.05).
\end{equation*}
The four continuous, time-static covariates, $X_{i}^{*(m)}$, $m=1, \dots, 4$, were drawn from a truncated normal distribution constrained to the interval $[-3, 3]$ and were held constant for the duration of the study.

\subsection{Treatment Assignment Mechanism}
At baseline, treatment was randomly assigned such that individuals had a 2.5\% probability of being treated, i.e., $A_{i,0} \sim \text{Bernoulli}(0.025)$. For all subsequent time points ($k=1, \dots, 8$), the probability of receiving treatment was determined by the following model:
\begin{align*}
    \Pr\big(A_{i,k}=1 \mid &A_{i,k-1}, X_{i,k}, \overline{Y}_{i,k}=0, \overline{Z}_{i,k}=0, \overline{U}_{i,k-1}=0 \big) \nonumber \\ 
    &= \text{expit}\left(\beta_{0} + 0.15 \sum_{m=1}^{4} X_{i}^{*(m)} -0.35 \sum_{m=1}^{4} X_{i,k}^{(m)} + 4 A_{i,k-1} -1.5 (1-A_{i,k-1})\right),
\end{align*}
where $\text{expit}(x) = (1+e^{-x})^{-1}$ and the intercept values for the model are given by $\beta_{0} = -3.75 \text{ for } k=1, \dots, 4$ and $\beta_{0} = 0.75 \text{ for } k=5, \dots, 8$.

\subsection{Outcome and Competing Risk Mechanisms}
For individuals who had not yet experienced an event or competing event, the probabilities of an outcome were determined by the following models:
\begin{align*}
    \Pr\big(Y_{i,k}=1 \mid &A_{i,k-1}, X_{i,k-1}, \overline{Y}_{i,k-1}=0, \overline{Z}_{i,k}=0, \overline{U}_{i,k-1}=0 \big) \\
    & = \text{expit}\left(-3.75 + 0.75 \sum_{m=1}^{4} X_{i}^{*(m)} - 0.35 \sum_{m=1}^{4} X_{i,k-1}^{(m)} - 2 A_{i,k-1}\right),\text{ and} \\
    \Pr\big(Z_{i,k}=1 \mid &A_{i,k-1}, X_{i,k-1}, \overline{Y}_{i,k-1}=0, \overline{Z}_{i,k-1}=0, \overline{U}_{i,k-1}=0 \big) \\ 
    & = \text{expit}\left(-4.5 + 0.25 \sum_{m=1}^{4} X_{i}^{*(m)} - 0.25 \sum_{m=1}^{4} X_{i,k-1}^{(m)} - 1.75 A_{i,k-1}\right).
\end{align*}

\subsection{Censoring Mechanism}
 Among those who were not previously censored, the probability of remaining uncensored was determined by the following model: 
\begin{align*}
    \Pr\big(U_{i,k}=1 \mid &X_{i,k}, \overline{Y}_{i,k}=0, \overline{Z}_{i,k}=0, \overline{U}_{i,k-1}=0 \big) \\ 
    & = 1 -\text{expit}\left(-4.75 + 0.1 \sum_{m=1}^{4} X_{i}^{*(m)} -0.2 \sum_{m=1}^{4} X_{i,k}^{(m)}\right).
\end{align*}

\section{Additional Results for the Analysis of Time to Bisphosphonate (BP) Discontinuation}\label{sec:supp_rwe}

Using the longitudinal Veterans Health Administration (VHA) data with a cohort of 109,286 older male veterans, we evaluated two time-varying treatment regimes over a 4-year follow-up period. The first regime was of complete BP discontinuation after three years of treatment, while the second regime was that of sustained oral BP therapy throughout the entire follow-up period. At each time interval, study participants were deemed compatible with either of the two regimes if they had not only followed the exact treatment sequence specified by that regime up to that point, but also were uncensored, alive, and had not yet experienced a clinical fracture. Figure \ref{fig:compat_sample_sizes} depicts the number of individuals compatible with each regime over time; we see a clear reduction in the size of the risk set throughout follow-up, despite tens of thousands of compatible individuals at the start of the study.

\begin{figure}[!htb]
\centering
    \includegraphics[width=\textwidth]{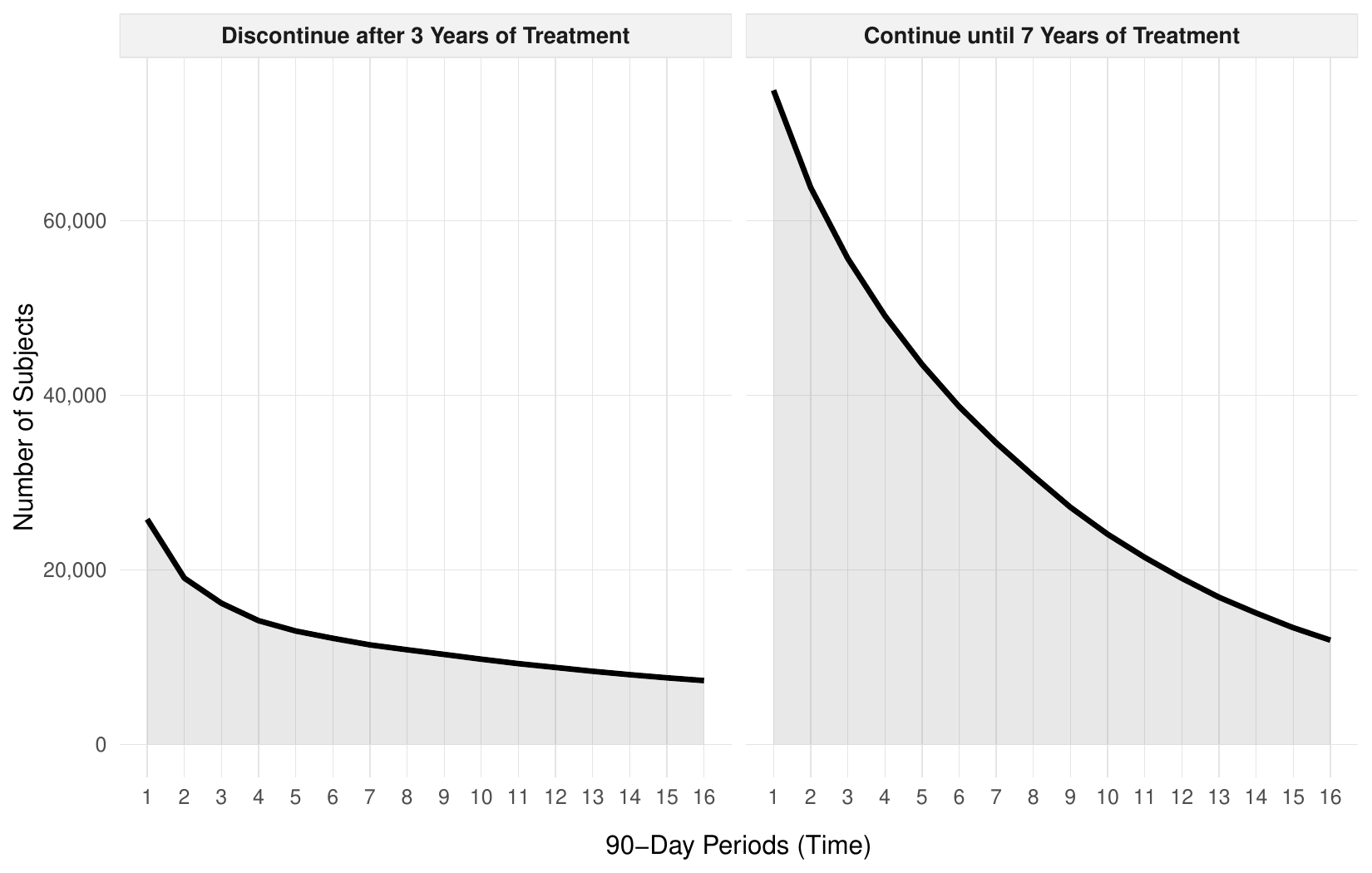}
    \caption{\textbf{Number of study participants compatible with each treatment regime across 16 90-day intervals.} Values reflect subjects remaining alive, fracture-free, uncensored, and consistent with the 3-year discontinuation regime (left) and the 7-year continuation regime (right).}
    \label{fig:compat_sample_sizes}
\end{figure}

Figure \ref{fig:ipw_vs_unweighted_aj} contrasts the cumulative incidence estimates from an unweighted Aalen-Johansen estimator against the IPW Aalen-Johansen estimator. Notably, the absolute differences between the two curves remain minimal throughout the 4-year follow-up period. At the end of the study (interval 16), the unweighted estimator yields a cumulative incidence of 7.55\% compared to 6.99\% for IPW estimator in the discontinuation regime, and 8.32\% compared to 8.75\% in the continuation regime. This indicates that inverse probability weighting alone has a negligible impact on the risk estimates.

\begin{figure}[!htb]
    \centering
    \includegraphics[width=\textwidth]{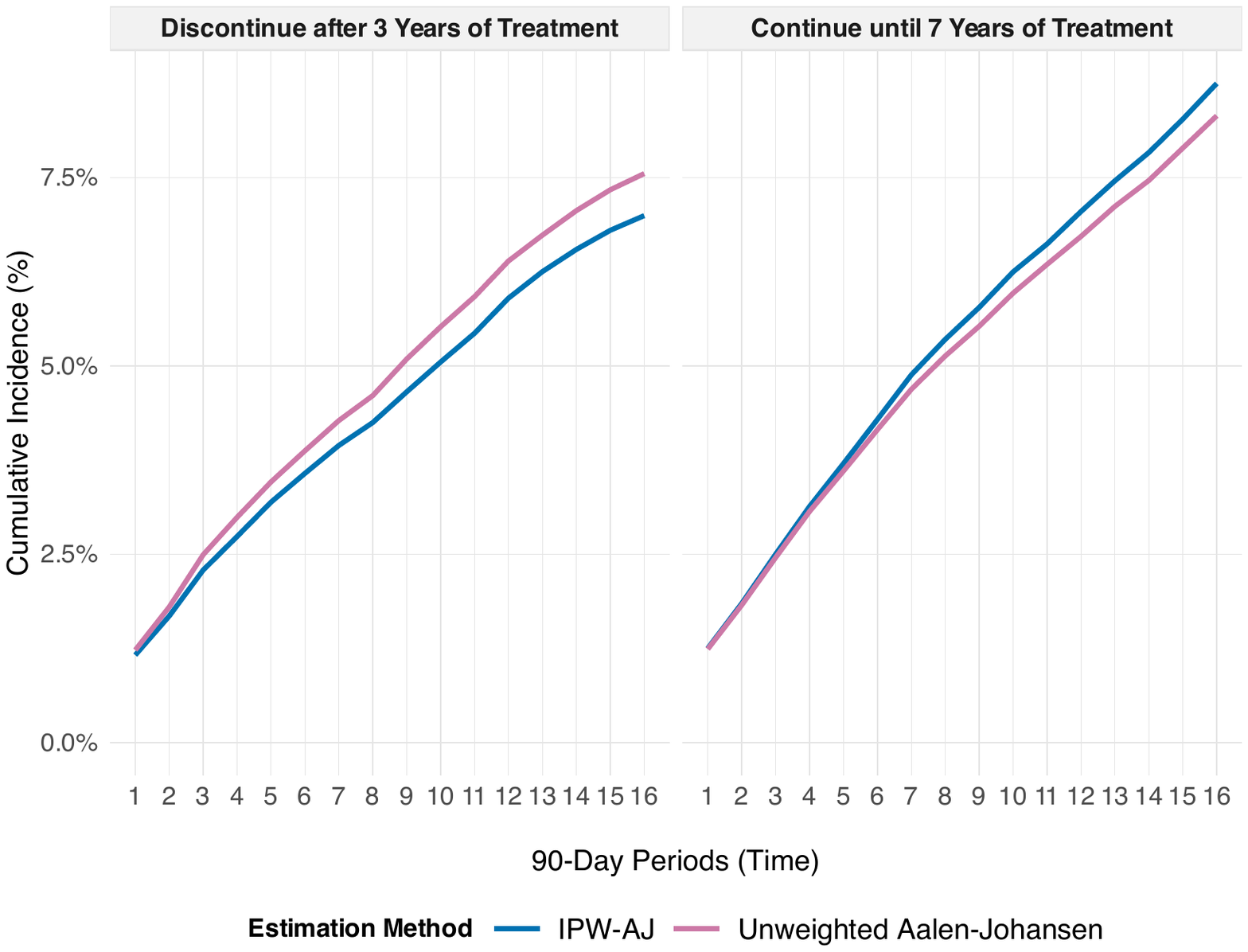}
    \caption{\textbf{Estimated fracture cumulative incidence across 16 90-day intervals using unweighted and IPW Aalen-Johansen estimators.} Results are shown under a 3-year discontinuation regime (left) and a 7-year continuation regime (right).}
    \label{fig:ipw_vs_unweighted_aj}
\end{figure}

Finally, Figure \ref{fig:three_method_comparison} compares the cumulative fracture incidence trajectories across the two treatment regimes using the IPW Aalen-Johansen (IPW-AJ) estimator, the G-formula, and the proposed EIF-based estimator. Across both regimes, all three methods capture a steady upward trend in fracture incidence over time. In the final interval, the estimated cumulative incidence of fractures under the proposed EIF-based method is approximately 10.13\% for the group discontinuing oral BP therapy after three years, compared to approximately 10.88\% for those continuing BP therapy for an additional four years. While the G-formula estimates closely track the EIF-based trajectories throughout the follow-up period, the IPW-AJ estimator diverges downward in the latter half, yielding notably lower cumulative incidence estimates by interval 16 (approximately 6.99\% and 8.75\% for the discontinuation and continuation regimes, respectively). The final estimates reflect a grouping based on model structure; the g-formula and the EIF-based approach, which both incorporate outcome modeling, yield cumulative incidence values at the end of the study that are roughly 2 to 4\%  higher than the purely propensity score-weighted IPW-AJ estimator.

\begin{figure}[!htb]
    \centering
    \includegraphics[width=\textwidth]{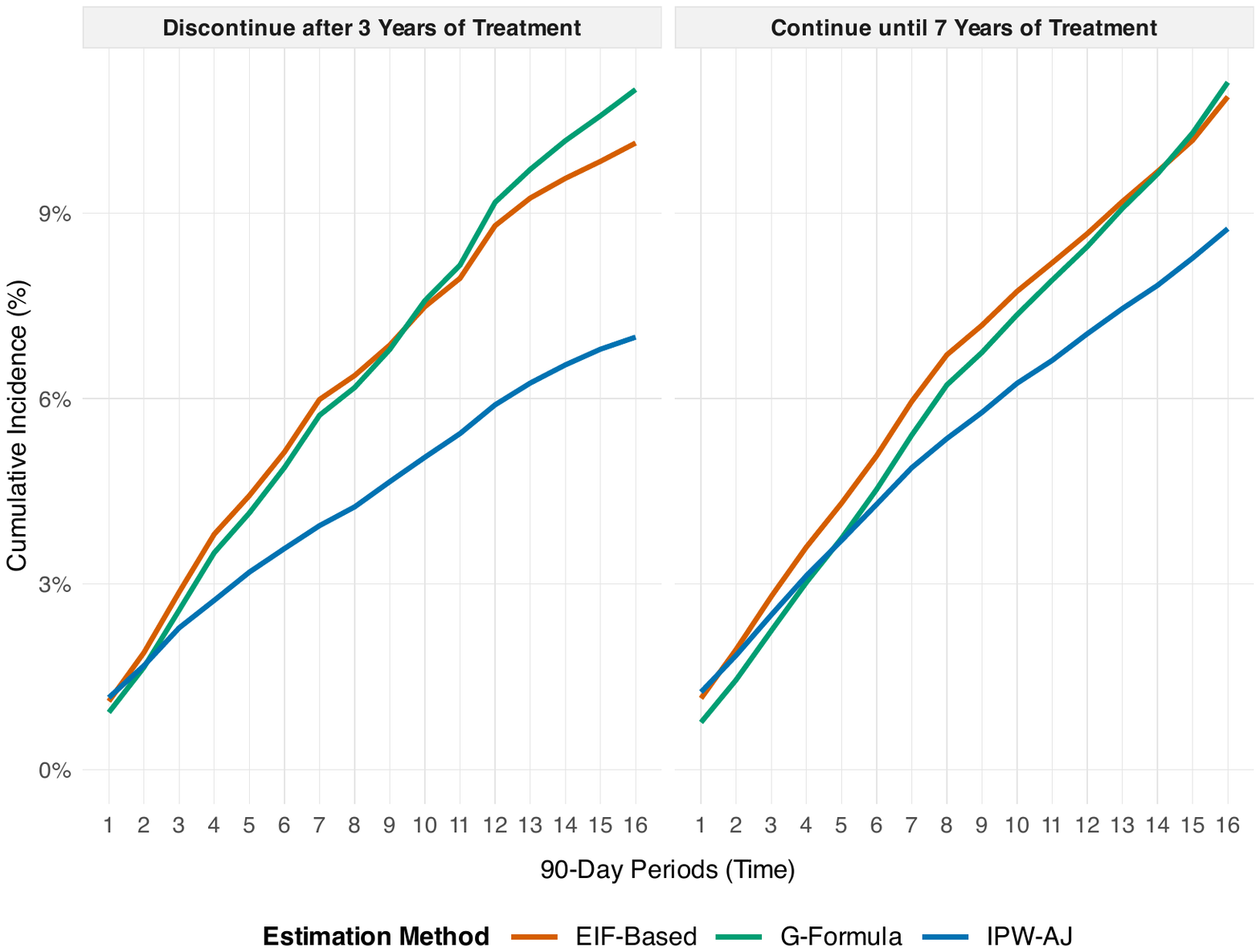}
    \caption{\textbf{Estimated fracture cumulative incidence across 16 90-day intervals using three statistical estimators.} Results compare the IPW Aalen-Johansen estimator, the g-formula, and the proposed EIF-based estimator under a 3-year discontinuation regime (left) and a 7-year continuation regime (right).}
    \label{fig:three_method_comparison}
\end{figure}

\end{document}